\documentclass{article}

\usepackage{fullpage}
\usepackage{natbib}
\usepackage{authblk}

\usepackage[utf8]{inputenc} %
\usepackage[T1]{fontenc}    %
\usepackage{hyperref}       %
\usepackage{url}            %
\usepackage{booktabs}       %
\usepackage{amsfonts}       %
\usepackage{nicefrac}       %
\usepackage{microtype}      %
\usepackage{amsmath,csquotes,xcolor,array,multicol,amsthm,thmtools}
\usepackage{enumitem,xparse}
\usepackage{footmisc}

\usepackage{thm-restate}

\usepackage{subfigure}
\usepackage{graphicx}
\graphicspath{{plots/}}
\newlength{\smallfigwidth}
\newlength{\smallfigheight}
\newlength{\smallfigsep}
\newlength{\legendheight}
\usepackage{wrapfig}

\usepackage{enumitem}
\setlist{nolistsep}
\setlist[itemize]{leftmargin=*}
\setlist[enumerate]{leftmargin=*}

\usepackage{marginnote}

\usepackage{microtype}
\usepackage{verbatim} %
\usepackage{hyperref}

\usepackage{xspace}

\usepackage{algorithm}
\usepackage{algpseudocode}
\algtext*{EndWhile}%
\algtext*{EndFor}%
\algtext*{EndProcedure}%
\algtext*{EndIf}%
\algnewcommand\algorithmicinput{\textbf{Input:}}
\algnewcommand\Input{\item[\algorithmicinput]}
\algnewcommand\algorithmicoutput{\textbf{Output:}}
\algnewcommand\Output{\item[\algorithmicoutput]}

\newtheorem{definition}{Definition}
\newtheorem{lemma}[definition]{Lemma}

\newtheorem{theorem}[definition]{Theorem}
\newtheorem{example}[definition]{Example}

\newtheorem{claim}[definition]{Claim}

\newcommand{\abs}[1]{\left|#1\right|}

\renewcommand{\lg}{\log}
\renewcommand{\varepsilon}{\epsilon}

\newcommand{\poly}{\operatorname{poly}}
\newcommand{\polylg}{\operatorname{polylg}}

\newcommand{\f}{\mathbf{f}}
\newcommand{\y}{\mathbf{y}}
\newcommand{\x}{\mathbf{x}}

\newcommand{\tO}{\tilde{O}}

\newcommand{\Gun}{G^{\textrm{un}}}
\newcommand{\tDelta}{\delta}

\newcommand{\ppp}{\mathbf{p}}
\newcommand{\qq}{\mathbf{q}}
\newcommand{\DD}{\mathbf{D}}
\newcommand{\WW}{\mathbf{W}}
\newcommand{\AAA}{\mathbf{A}}
\newcommand{\II}{\mathbf{I}}

\newcommand{\LLL}{\mathbf{L}}
\newcommand{\LL}{\ensuremath{\mathbf{\mathcal{L}}}}

\newcommand{\trvv}{\mathbf{v}'}

\newcommand{\trii}{\mathbf{b}}
\newcommand{\vv}{\mathbf{v}}

\newcommand{\mm}{\mathbf{m}}

\newcommand{\ww}{\mathbf{w}_1}
\newcommand{\1}{\mathbf{1}}

\newcommand{\Exp}[1]{\mathbf{E}\left[ #1 \right]}
\newcommand{\Var}[1]{\mathbf{Var}\left[ #1 \right]}
\newcommand{\Prob}[1]{\mathbf{Pr}\left( #1 \right)}

\DeclareRobustCommand{\ALG}{%
	\ifmmode
		\operatorname{ON}
	\else
		\text{ON}\xspace
	\fi
}
\DeclareRobustCommand{\OFF}{%
	\ifmmode
		\operatorname{OFF}
	\else
		\text{OFF}\xspace
	\fi
}
\DeclareRobustCommand{\APPROXALGO}{%
	\ifmmode
		\operatorname{APPROX}
	\else
		\text{APPROX}\xspace
	\fi
}

\providecommand{\abs}[1]{\ensuremath{\left\lvert#1\right\rvert}}
\providecommand{\norm}[1]{\ensuremath{\lVert#1\rVert}}

\def\vol{\operatorname{vol}}

\newcommand{\betain}{\beta_{\text{in}}}

\newcommand{\pinner}{p_\text{intra}}
\newcommand{\pcross}{p_\text{cross}}
\newcommand{\psign}{p_\text{sign}}
\newcommand{\qsign}{q_\text{sign}}

\newcommand{\rwseeded}{\textsc{rw-seeded}\xspace}
\newcommand{\rwunseeded}{\textsc{rw-unseeded}\xspace}
\newcommand{\rwuseeded}{\textsc{rw-u-seeded}\xspace}
\newcommand{\rwuunseeded}{\textsc{rw-u-unseeded}\xspace}
\newcommand{\FOCG}{\textsc{FOCG}\xspace}
\newcommand{\polarSeeds}{\textsc{polarSeeds}\xspace}

\newcommand{\WikiSmall}{\textsf{WikiS}\xspace}
\newcommand{\WikiMedium}{\textsf{WikiM}\xspace}
\newcommand{\WikiLarge}{\textsf{WikiL}\xspace}

\newcommand{\betainner}{\beta^{\text{inner}}}
\newcommand{\betaout}{\beta_{\text{out}}}

\usepackage[skip=0pt]{caption}

\allowdisplaybreaks

\title{Sublinear-Time Clustering Oracle for Signed Graphs}

\date{}

\author[1]{Stefan Neumann}
\author[2]{Pan Peng\footnote{Correspondence to: Pan Peng <\href{mailto:ppeng@ustc.edu.cn}{ppeng@ustc.edu.cn}>.}}
\affil[1]{KTH Royal Institute of Technology, Stockholm, Sweden}
\affil[2]{School of Computer Science and Technology, University of Science and Technology of China, Hefei, China}

\begin{document}

\maketitle

\begin{abstract}
Social networks are often modeled using \emph{signed} graphs, where vertices correspond to users and edges have a sign that indicates whether an interaction between users was positive or negative.  The arising signed graphs typically contain a clear community structure in the sense that the graph can be partitioned into a small number of \emph{polarized} communities, each defining a sparse cut and indivisible into smaller polarized sub-communities.  We provide a \emph{local clustering oracle} for signed graphs with such a clear community structure, that can answer membership queries, i.e., ``\emph{Given a vertex~$v$, which community does~$v$ belong to?}'', in sublinear time by reading only a small portion of the graph. Formally, when the graph has bounded maximum degree and the number of communities is at most $O(\log n)$, then with $\tilde{O}(\sqrt{n}\poly(1/\varepsilon))$ preprocessing time, our oracle can answer each membership query in $\tilde{O}(\sqrt{n}\poly(1/\varepsilon))$ time, and it correctly classifies a $(1-\varepsilon)$-fraction of vertices w.r.t.\ a set of hidden planted ground-truth communities. Our oracle is desirable in applications where the clustering information is needed for only a small number of vertices.  Previously, such local clustering oracles were only known for \emph{un}signed graphs; our generalization to signed graphs requires a number of new ideas and gives a novel spectral analysis of the behavior of \emph{random walks with signs}.  We evaluate our algorithm for constructing such an oracle and answering membership queries on both synthetic and real-world datasets, validating its performance in practice.
\end{abstract}

\section{Introduction}
Finding clusters (or communities) in graphs is a well-studied and fundamental problem in computer
science. While classically this problem has been studied in unsigned graphs,
several recent works have focused on \emph{signed} graphs, where each edge has
a sign indicating whether the interaction between two nodes was friendly or
hostile. This setting has been motivated by polarization in social networks,
where the users form groups that have mostly friendly interactions within each
group but there may exist hostile interactions between opposing groups (see,
e.g., \cite{bonchi2019discovering,xiao2020searching,ordozgoiti20finding,atay20cheeger}).

More concretely, in a signed graph $G=(V,E,\sigma)$, each edge $e=(u,v)\in E$ is
associated with a sign $\sigma(e) \in \{+,-\}$ indicating a positive (friendly)
or negative (hostile) relation between the two vertices $u$ and $v$. To model
polarization, \citet{harary1953notion} proposed the notion of balanced
graphs: a graph is \emph{balanced} if it can be partitioned into two
subsets $V_1$ and $V_2$ such that the induced subgraphs $G[V_1]$ and $G[V_2]$
only contain positive edges, while all edges with one endpoint in $V_1$ and the other
endpoint in $V_2$ have a negative sign.
For example, in a social network like
Twitter the groups $V_1$ and $V_2$ could correspond to users of opposing
opinions (e.g., Democrats and Republicans) that have a conflict but that behave
nicely within their respective groups.

To detect polarization in social networks, several recent works have aimed at
finding \emph{nearly-balanced} communities inside signed graphs, i.e., their
goal was to find induced subgraphs that after the removal of only few edges
become balanced and are sparsely connected to the
outside~\cite{bonchi2019discovering,mercado2019spectral,xiao2020searching}.
Often the resulting communities are \emph{minimal} in the sense that they are
nearly-balanced and they cannot be further divided into smaller nearly-balanced
communities. We will also refer to these communities as \emph{polarized}.
The main drawback of many existing methods for finding polarized communities is
that they are inherently \emph{global}, i.e., they need to process the full
graph and return a partitioning of \emph{all} vertices.
In practice, however, the graphs are often so large that methods
which aim to cluster all vertices have prohibitively high running times.
Additionally, when mining social networks, the full graph is often not available
because social network providers, such as Twitter, limit the amount of data that
is available due to privacy constraints.

Fortunately, in many settings we only require the community membership
information \emph{for a small number of vertices}, or we just want to know if
two given vertices belong to the same cluster or not. This could be the case,
for example, when an analyst wishes to find out whether two users are part of
the same polarized discussion or not. Furthermore, even in settings when the amount of
data is limited (like in the Twitter example above), it seems feasible to
explore the local neighborhood  (e.g., by performing random walks) of each user that shall be
classified.

\textbf{Our contributions.}
We provide a \emph{local clustering oracle for signed graphs}.
The oracle preprocesses a small part of the graph and after the preprocessing
finished, for a queried vertex~$v$ it can answer the following query:
\begin{itemize}
    \item \textsc{WhichCluster}($v$): Returns which cluster $v$ belongs to.
\end{itemize}
Here, we assume that there is a set of (hidden) ground-truth clusters and
\textsc{WhichCluster}($v$) returns the index of the cluster $v$ belongs to. Ideally, when two nodes $u,v$ belong to the same ground-truth cluster,
then the queries \textsc{WhichCluster}($u$) and \textsc{WhichCluster}($v$) will
return the same result, and if $u,v$ belong to different clusters, the results
will be different.

Both the preprocessing time as well as the query time of our clustering oracle
are \emph{sublinear} in the size of the input graph.
This is particularly useful when the clustering information is only required for
a small number of vertices.  More concretely, our oracle has preprocessing
time\footnote{Here, $\tO(f(n))$ denotes running times $O(f(n)\cdot\polylg(n))$.}
$\tO(\sqrt{n} \poly(1/\varepsilon))$, where $n$ is the number of vertices in the
graph and $\varepsilon>0$ is an error parameter. The query procedure has
query time $\tO(\sqrt{n} \poly(1/\varepsilon))$. Such clustering oracles have
been previously studied for unsigned
graphs~\cite{peng20robust} from a theoretical point of view but none were known
for signed graphs and it was unclear how well they perform in practice. 

In a nutshell, the query procedure \textsc{WhichCluster}($v$) performs
$\tO(\sqrt{n} \poly(1/\varepsilon))$ random walks of length $O(\lg n)$ starting
at $v$ and then aggregates the information from these random walks into a sparse
vector $\mm_v$ with $\tO(\sqrt{n} \poly(1/\varepsilon))$ non-zero entries.
We define a pseudometric $\Delta$ on the space of these vectors and show
that the vectors of vertices from the same community have ``small''
$\Delta$-distance, while the vectors of vertices from different communities
have ``large'' $\Delta$-distance.
We describe the details in Sec.~\ref{sec:algorithm}.

\begin{example}
One possible application of our oracle is the clustering of an online social
networks such as Twitter. The user interactions on Twitter can be interpreted as
a signed graph with some hidden communities. Now it is possible to label (e.g.,
		by hand) a small number of users for several nearly-balanced communities. These users can be used as seed nodes for the oracle and then the oracle can be used to efficiently classify users based on which community they belong to. This addresses the issue that the Twitter graph is too large to cluster it completely. Additionally, as researchers we do not have access to the full Twitter graph but it seems feasible to perform a small number of short random walks from each user that shall be classified.
\end{example}

We provide a theoretical analysis of the oracle for bounded-degree graphs with a
constant (or logarithmic) number of ``well-behaved'' nearly-balanced communities. We show that
when we apply \textsc{WhichCluster}($v$) to all vertices, then the oracle
classifies a $(1-\varepsilon)$-fraction of the vertices correctly.  See
Thm.~\ref{thm:oracle} for the formal statement of our result.
To obtain this result, we relate the spectrum of the graph's normalized signed
Laplacian to the random walks performed by the query procedures. Therefore, we
give a novel spectral analysis
of the behavior of \emph{random walks with signs}, which essentially keep track
whether the random walk traversed an even or an odd number of negative edges.
Then we relate the
$\Delta$-distance of signed random walk vectors to the eigenvalues/eigenvectors of the
signed Laplacian.  See Sec.~\ref{sec:main-technical-contribution} for a
precise statement of our technical contribution and
App.~\ref{sec:analysis-overview} for an overview of our analysis.
	
\emph{Our theoretical contributions.} Our theoretical results are inspired by
sublinear clustering oracles for \emph{un}signed graphs, and
some notions and lemmas look superficially similar to the counterparts in
unsigned graphs  (e.g., \cite{czumaj15testing,peng20robust}). However, to
generalize these oracles to signed graphs we need several new ideas.
We will now briefly discuss these new ideas.

First, the clustering oracle in unsigned graphs is based on the following
intuitive idea: a random walk starting from a randomly chosen vertex of some
cluster $U$ will first be trapped in $U$, and the corresponding
distribution converges to the uniform distribution on $U$ (for simplicity, we
assume the graph is regular); later, the random walk moves out of $U$ and then
the distribution converges to the uniform distribution on the whole graph. In
signed graphs, this intuition is no longer true. In particular, the distribution
of a signed random walk does not necessarily converge to a stationary
distribution (if it exists). Interestingly, we show that (informally) in a
polarized cluster $U$ with a bipartition $V_1$ and $V_2$ corresponding to the
two opposing groups, a signed random walk converges to either a scaled version
of the uniform distribution on $V_1$, or a scaled version of the uniform
distribution on $V_2$. To this end, we show that (see Lem.
\ref{lem:small-diff}) if we map vertices to the spectral embedding
defined by the first $k$ eigenvectors of the signed normalized Laplacian matrix
of the graph, then the embedded points are centered around \emph{two opposite}
centers. To contrast, in unsigned graphs, the spectral embedding of most
vertices in the same cluster are close to one \emph{single} center. To show the
existence of two opposite centers, we develop a new property relating the
eigenvectors and the polarized clusters, which may be useful for future work on
clustering in signed graphs.

Second, for unsigned oracles, it suffices to consider the $\ell_2^2$-distance
between two random walk distributions starting from any two vertices
$u,v$ to decide if $u,v$ are similar or not (i.e., if they belong to the same
cluster or not). For signed oracles, since each polarized cluster has two
opposite centers, we need to introduce a pseudometric distance $\tilde{\Delta}$ 
between the corresponding vectors to compare the similarity of two vertices. That is, for any two vertices $u,v$ with random walk vectors $\mm_u,\mm_v$, we define %
\[
	\tilde{\Delta}_{u,v}
	:=\min\{\norm{\mm_u-\mm_v}_2^2,\norm{\mm_u+\mm_v}_2^2\}.
\]
Intuitively, if $u,v$ belong to the same polarized cluster with a bipartition $V_1,V_2$, then either the distance $\norm{\mm_u-\mm_v}_2^2$ is small (corresponding to the case that $u,v$ belong to the same part in the bipartition), or $\norm{\mm_u+\mm_v}_2^2$ is small (corresponding to the case that $u,v$ belong to two different parts). Furthermore, if $u,v$ belong to two different clusters, then neither of these two distances is small.

Third, to  characterize the cluster structure of a signed graph~$G$, it is
somehow natural to use the \emph{signed bipartiteness ratio} (see
Sec.~\ref{sec:preliminaries}), a signed analogue of conductance in unsigned
graphs. However, we find that one \emph{cannot} use the signed bipartiteness
ratio of a graph to characterize the \emph{inside}
structure of a potential polarized cluster (see
App.~\ref{sec:motivation-inner}). We resolve this by introducing a new notion called
\emph{inner signed bipartiteness ratio} $\betainner(G)$ of $G$ that is a
minimization function by considering all vertex subsets of at most half the
total volume of the graph (see Eqn.~\eqref{def:innersignedbipartiteness} and
Def.~\ref{def:clustering}).

\emph{Our practical contributions.}
We provide the first implementations of signed \emph{and unsigned} oracles.
While our signed oracles with theoretical guarantees can only distinguish
between different communities, we also provide a heuristic extension which
allows for queries of the type: ``In which opposing group of a community is
vertex $v$?'' In practice, such a query could be used, e.g., to decide whether a
user in a social network is a Democrat or a Republican.

We evaluate our algorithms on synthetic and on real-world datasets
(Sec.~\ref{sec:experiments}) and show that our oracles are practical.
Our methods work well \emph{even when the graphs do not satisfy
the bounded degree assumption from our theoretical analysis}. Furthermore, our
algorithms outperform existing methods when the graphs contain large
communities.  We further provide novel real-world datasets which contain signed
graphs with a small number of large ground-truth communities; to the best of our
knowledge, these are the first public datasets with this property and we make
them freely available.

\textbf{Related work.}
Finding communities in signed graphs has received a lot of attention.
One line of works models polarized communities as (nearly) balanced
subgraphs in a signed
graph~\cite{kunegis2010spectral,chiang12scalable,bonchi2019discovering,cucuringu19sponge,cucuringu20regularized,mercado2019spectral,xiao2020searching,chu16finding,chiang14prediction,chiang12scalable}.
\citet{xiao2020searching}
provide a local algorithm for finding nearly balanced subgraphs.
The algorithm of~\cite{xiao2020searching} requires a set of seed nodes and
returns a subgraph with small signed bipartiteness ratio; we compare our
algorithm against this work in the experiments.  
\citet{ordozgoiti20finding} find large (exactly) balanced subgraphs.
In another line of work, polarized communities were modeled using $k$-way
balanced graphs~\cite{chiang12scalable} and signed stochastic block
models~\cite{mercado16clustering,mercado2019spectral} or using correlation
clustering~\cite{bansal04correlation,bianchi12correlation}; these results are
not directly comparable to our work since we consider $k$~disjoint (nearly)
$2$-way balanced subgraphs while these works try to find a single partitioning
of the graph that reveals $k$~communities.  Interestingly, many of these works
are based on spectral graph theory (e.g.,
\cite{kunegis2010spectral,chiang12scalable,xiao2020searching,ordozgoiti20finding,mercado16clustering,mercado2019spectral}).

\citet{jung16personalized,jung20random} use signed random walks with
restarts to rank users in social networks but, unlike in our work, they do not
relate the signed random walks to the spectrum of the signed graph.

Sublinear-time algorithms for clustering
\emph{unsigned} graphs have been studied using the notion of conductance (rather
than the signed bipartiteness ratio).
\citet{czumaj15testing} gave a property testing algorithm 
which can decide whether a graph is $k$-clusterable or is \emph{far from} being
$k$-clusterable in sublinear time. %
Interestingly, the algorithm by
\citet{czumaj15testing} can be adapted to a sublinear-time
clustering oracle. %
\citet{peng20robust} extended
this %
to a robust clustering oracle that reports the clustering
information of graphs with noisy partial information.
\citet{chiplunkar18testing} and \citet{gluch21spectral} provided
further improvements.

Intriguingly, the algorithm by \citet{pons06computing} is also based on
clustering the vectors of short random walks and is very popular in
practice (e.g., it is implemented in the \emph{igraph} software
package~\cite{csardi06igraph}). The similarity measure used in \citet{pons06computing} is quite similar to the one which was
independently proposed by \citet{czumaj15testing}, though the latter is focusing on using a small number of random walks to estimate the measure (rather than computing it directly) and thus achieving a sublinear-time algorithm. Therefore, one can view the
results of \citet{czumaj15testing} and of this paper as a further theoretical
justification for the practical success of the work of \citet{pons06computing}.

\section{Preliminaries}
\label{sec:preliminaries}

Let $G=(V,E,\sigma)$ be an unweighted signed graph with $n$~vertices, $m$~edges
and edge signs $\sigma(e)\in\{+,-\}$. The degree $d_G(u)$ of a vertex $u$ is 
$d_G(u) = \abs{ \{ v \colon (u,v)\in E\}}$; note that the degree does not take
into account the signs of the edges. For any set $S\subseteq V$, let 
$\vol_G(S)=\sum_{u\in S}d_G(v)$. 
The volume of $G$ is $\vol(G) = \sum_{u\in V} d_G(u)$.

For $V_1,V_2 \subseteq V$, we set
$E(V_1,V_2)=\{(u,v)\in E \colon u\in V_1, v\in V_2\}$.  Furthermore, we set
$E^+(V_1,V_2)=\{(u,v)\in E(V_1,V_2) : \sigma(uv)=+\}$ and
$E^-(V_1,V_2)=\{(u,v)\in E(V_1,V_2) : \sigma(uv)=-\}$.
When $V_1=V_2$, we set $E(V_1)=E(V_1,V_1)$. To maintain consistency with
previous works, we set $\abs{E(V_1)}$ to twice the number of edges in $G[V_1]$
(but we do not make this change for $\abs{E(V_1,V_1)}$).
We define $\abs{E^+(V_1)}$ and $\abs{E^-(V_1)}$ analogously to
$\abs{E(V_1)}$.

\textbf{Signed bipartiteness ratio.} Following the work of
\citet{xiao2020searching}, we use the \emph{signed bipartiteness ratio} to
capture polarization between two opposing groups in a graph.
A pair $(V_1,V_2)$ is a \emph{sub-bipartition} of $V$ if $\emptyset \neq V_1\cup V_2
\subseteq V$ and $V_1\cap V_2 = \emptyset$. For a sub-bipartition $(V_1,V_2)$ of $V$ we
set $\vol_G(V_1,V_2) = \sum_{u\in V_1\cup V_2} d_G(u)$.
Now the \emph{signed bipartiteness ratio} of $(V_1,V_2)$ is given by
\begin{align*}
	\beta_G(V_1,V_2)
	:= \frac{e_G(V_1,V_2)}{\vol_G(V_1,V_2)},
\end{align*}
where 
\begin{align*}
	e_G(V_1,V_2) =&
		2\abs{E_G^+(V_1,V_2)}
		+ \abs{E_G^-(V_1)}
		+ \abs{E_G^-(V_2)} \\
		&+ \abs{E_G(V_1\cup V_2, \overline{V_1\cup V_2})}.
\end{align*}

Observe that when $\beta_G(V_1,V_2)$ is small then the induced subgraph
$G[V_1\cup V_2]$ is close to balanced (i.e., $G[V_1]$ and $G[V_2]$ contain only
few negative edges and there are mostly negative edges between $V_1$ and $V_2$),
\emph{and} the vertices in $V_1\cup V_2$ are sparsely connected to
the rest of the graph (i.e., there are few edges from $V_1\cup V_2$ to
$V\setminus (V_1\cup V_2)$).

For a set of vertices $\emptyset\neq U\subseteq V$, we define
\emph{the signed bipartiteness ratio of $U$ in $G$} as 
$\beta_G(U) :=  \min_{(V_1,V_2) \colon V_1 \cup V_2 = U} \beta_G(V_1,V_2)$,
where the minimum is taken over all partitions $(V_1,V_2)$ of $U$.
For a graph $G$, we define the \emph{(classic) signed bipartiteness ratio} of
$G$ as  
\begin{align*}
	\beta(G)
	&:= \min_{\emptyset\neq U\subseteq V} \beta_G(U) \\
	&= \min_{(V_1,V_2): \textrm{sub-bipartition of $V$}} \beta_G(V_1,V_2).
\end{align*}
Observe that a set of vertices $U$ is balanced \emph{if and only if} $U$ can be
partitioned into subsets $V_1$ and $V_2$ such that $\beta_{G[U]}(V_1,V_2)=0$ 
\emph{if and only if} $\beta(G[U])=0$.
Thus, one can interpret the signed bipartiteness ratio as a measure for how
close a certain subgraph is to being balanced.
The sets $V_1$ and $V_2$ that partition $U$ are sometimes called
\emph{biclusters}.

\textbf{Spectral signed graph theory.}
Next, we introduce definitions for the spectral analysis of signed graphs. 
We use bold letters to denote vectors and matrices. Let
$\DD$ be the $n\times n$ diagonal degree matrix of $G$, i.e., $\DD_{uu}=d_G(u)$ for
all $u\in V$. Let $\AAA^\sigma$ be the $n\times n$ signed adjacency matrix, i.e.,
$\AAA^\sigma_{uv}=\sigma(uv)$ if $(u,v)\in E$, and $\AAA^\sigma_{uv}=0$, otherwise.
Let $\II$ be the $n\times n$ identity matrix.

We call $\LLL^\sigma:=\DD-\AAA^\sigma$ the \emph{signed (unnormalized) Laplacian}
matrix, and $\LL^\sigma:=\II-\DD^{-1/2}\AAA^{\sigma}\DD^{-1/2}$ the \emph{signed
normalized Laplacian} matrix.  It is well-known that all eigenvalues of
$\LL^\sigma$ are in the interval~$[0,2]$ and we list them in non-decreasing
order as $0\leq \lambda_1\leq \cdots \leq \lambda_n\leq 2$. 

For $k\in[n]$, the \emph{$k$-way signed bipartiteness ratio of $G$} is
\begin{align*}
	\beta_k(G) 
	&:= \min_{ U_1,\dots,U_k } \max_{i=1,\dots,k} \beta_G(U_i) \\
	&= \min_{ \{(V_{2i-1},V_{2i})\}_{i=1}^k } \max_{i=1,\dots,k} \beta_G(V_{2i-1},V_{2i}).
\end{align*}
Here, the minima are taken over all possible choices of $k$ non-empty, disjoint
sets $U_i\subset V$ and disjoint sub-bipartitions $(V_{2i-1},V_{2i})$,
respectively.  Intuitively, $\beta_k(G)$ is small iff $G$ contains $k$ disjoint
communities that are close to balanced iff $G$ contains $k$ polarized communities.

\citet{atay20cheeger} provided a
Cheeger-type inequality that relates the $k$-way signed bipartiteness ratio to
the eigenvalues $\lambda_k$ of the signed normalized Laplacian $\LL^\sigma$.
\begin{theorem}[Higher-order Signed Cheeger Inequality~\cite{atay20cheeger}]
\label{thm:signed-cheeger}
	There exists a constant $C_1$ such that
	for all signed graphs $G$ and $k\in[n]$,
		$\frac{\lambda_k}{2}
		\leq \beta_k(G)
		\leq C_1 k^3 \sqrt{\lambda_k}$.
\end{theorem}

Finally, $\WW=\frac{\II+\DD^{-1}\AAA^\sigma}{2}$ is the \emph{walk matrix} that
corresponds to lazy random walks in a signed graph $G$. Additionally, for
$t\in\mathbb{N}$ and $v\in V$, we set $\ppp_v^t = \1_v \WW^t$, where $\1_v$ is the
$n$-dimensional indicator vector that has a $1$ in the $v$'th entry and is $0$
in all other entries.

\section{Main Result and Algorithm}
\label{sec:result}

In this section, we formally present our main result and give the details of our
clustering oracle. To state our theorem, we first need to introduce two new
definitions.
For a signed graph $G=(V,E,\sigma)$, we define the
\emph{inner signed bipartiteness ratio} of $G$ as
\begin{align}
\label{def:innersignedbipartiteness}
	\betainner(G)
	:= \min_{\emptyset \neq U \subseteq V \colon \vol(U) \leq \frac{1}{2} \vol(G)}
		\beta_{G}(U).
\end{align}
Note that we only consider subsets~$U$
\emph{of volume at most $\frac{1}{2}\vol(G)$}; this is in contrast with the
definition of $\beta(G)$, in which the minimum is taken over \emph{all possible}
subsets~$U$. %
The definition (\ref{def:innersignedbipartiteness}) resembles the inner conductance that has been used to study the clusterability of unsigned graphs (e.g., \cite{gharan2014partitioning,czumaj15testing}), though in contrast with $\beta(G)$, it is \emph{not} directly associated with the signed Cheeger inequality. 

Next, we define the notion of clusterability under which we will obtain our
theoretical results.
\begin{definition}
\label{def:clustering}
	Let $k\in\mathbb{N}$, $\betain,\betaout\in\mathbb{R}_{> 0}$ and let
	$G=(V,E,\sigma)$ be an unweighted signed graph.
	We say that $G$ is \emph{$(k,\betain, \betaout)$-clusterable} if there
	exists a partition of $V$ into $k$ disjoint subsets
	$(U_1,\dots,U_k)$ such that
	$\beta_G(U_i) \leq \betaout$ and
	$\betainner(G[U_i]) \geq \betain$ for all $i\in[k]$.
	Each subset $U_i$ is called a
	\emph{$(\betain,\betaout)$-cluster} and the corresponding partitioning is
	called a \emph{$(k,\betain,\betaout)$-clustering}.
	Furthermore, if each subset $U_i$ satisfies that $|U_i|\geq
	\Omega(\frac{n}{k})$, then we call the partition $(U_1,\dots,U_k)$ a
	\emph{balanced}  $(k,\betain,\betaout)$-clustering.
\end{definition} 

Let us briefly explain this definition; it is handy to think of $\betaout$ as
very small and $\betaout\ll\betain$. The first condition that
$\beta_G(U_i) \leq \betaout$ for all $i=1,\dots,k$ ensures that the graph
contains $k$~communities which are nearly-balanced and that can be viewed as
polarized communities.  The second condition ($\betainner(G[U_i])\geq\betain$
for all $i=1,\dots,k$) ensures that each of the nearly-balanced communities is
minimal in the sense that it cannot be further decomposed into more balanced
communities.

We remark that at first glance it might be surprising that in
Definition~\ref{def:clustering}, we use $\betainner(G[U_i])$ instead of
$\beta(G[U_i])$ to measure the indivisibility of $G[U_i]$ into smaller
nearly-balanced communities.  However, in App.~\ref{sec:motivation-inner} we
show that if $\beta_G(U_i)$ is small, then $\beta(G[U_i])$ is also small. %
This indicates that $\beta(G[U])$ is not an appropriate measure for this characterization.

Now we state our main result.  We consider signed graphs $G$ of degree at
most $d$, where $d$ is a constant throughout the paper. We assume that we have
query access to the adjacency list of $G$, i.e., for any vertex $v$ and an index
$i\leq d$, we can query the $i$-th neighbor of $v$ in constant time if it exists
(if no such neighbor exists we get a special symbol `$\bot$').  Let
$P\triangle Q$ denote the symmetric difference. In the following, we let $d>10$,
$k\geq 1$, $\varepsilon \in (0,1)$ and $\betain\in(0,1)$.
Let $n$ be an integer such that $n\geq
\frac{1800k^2\log(k)}{\gamma\varepsilon}$, where $\gamma\in(0,1]$ is a constant.
\begin{restatable}{theorem}{oracle}
\label{thm:oracle}
	Let $G=(V,E,\sigma)$ be a signed graph with $|V|=n$ vertices and maximum
	degree at most $d$. Suppose that $G$ has a balanced
	$(k,\betain,\betaout)$-clustering $U_1,\cdots, U_k$,
	$\betaout< \frac{\varepsilon \betain^2 }{C'\log(k) k^7d^3\log n}$,
	where $C'$ is some sufficiently large constant,
	and $\abs{U_i}\geq \gamma \frac{n}{k}$ for all $i=1,\dots,k$. There exists
	an algorithm that has query access to the adjacency list of $G$ and
	constructs a clustering oracle
	in $O(\sqrt{n}\cdot \poly(\frac{kd\cdot\log n}{\varepsilon\betain}))$
	preprocessing time.  Furthermore, with probability at least
	$0.9$, the following hold:
	\begin{enumerate}
		\item Using the oracle, the algorithm can answer any
		\textsc{WhichCluster} query in
			$O(\sqrt{n}\cdot \poly(\frac{kd\cdot\log n}{\varepsilon\betain}))$ time. 
		\item Let $P_i:=\{u\in V: \textsc{WhichCluster}(u)=i\}$, $i\in [k]$,
			be the clusters defined by \textsc{WhichCluster}. Then there exists a
			permutation $\pi: [k]\to [k]$ such that for all $i\in[k]$, $\abs{P_{\pi(i)}\triangle U_{i}}\leq O(\varepsilon/\lg k)|U_i|$.
	\end{enumerate}
\end{restatable}
The theorem asserts that if the input graph $G$ has bounded degree and satisfies
the assumptions from Def.~\ref{def:clustering} with $\betaout\ll\betain$,
then our clustering oracle has preprocessing and query time
$\tO(\sqrt{n}\poly(1/\varepsilon))$. Furthermore, the second item
implies that if we pick $\varepsilon$ small enough, we can make the number
of misclassified vertices arbitrarily small. In particular, for any
$\delta>0$ we can pick $\varepsilon$ such that the oracle classifies at least a
$(1-\delta)$-fraction of the vertices correctly.

\subsection{The Algorithm}
\label{sec:algorithm}
Now we present the implementation of our clustering oracle. The main building
block of our oracle are lazy signed random walks which we will discuss first.
Based on a sequence of random walks starting at a vertex $u$, we will define a
sparse vector $\mm_u$. We will then use the vectors $\mm_u$ and $\mm_v$ to estimate
distances $\tDelta_{uv}$ between vertices $u$ and $v$ with the intuition that
$u$ and $v$ are in the same cluster iff $\tDelta_{uv}$ is small.  We will also
discuss the preprocessing of the oracle and the query procedures.

\textbf{Lazy signed random walks.} We introduce lazy signed random
walks. Intuitively, a lazy signed random walk is a lazy random
walk on the unsigned version of the graph that keeps track of the sign of the
walk. Here, the sign of the walk is the multiplication of the signs of
all traversed edges.

More formally, a \emph{lazy signed random walk} of length $t$ from a
vertex $u$ proceeds as follows. Initially, at step $T=0$, we start at vertex
$u_0:=u$ with sign $s_0:=+$. Suppose that at step $0\leq T< t$ we
are at vertex $u_T$ with sign $s_T\in\{+,-\}$. Then at the step $T+1$, with
probability~$\frac12$ we stay at $u_T$ and keep the sign unchanged (i.e.,
$u_{T+1}=v_t, s_{T+1}=s_T$), and with the remaining $\frac12$ probability, we
choose a random neighbor $v$ of $u_T$ with probability $\frac{1}{d_G(u_T)}$, and
move to $v$ and update $u_{T+1}=v, s_{T+1}=\sigma(e_{T+1}) s_T$, where
$e_{T+1}=\{u_T,v\}$.  Thus, if a walk traverses the edges $e_1,\dots,e_{t}$ then
the final sign of the walk is $\prod_{i=1}^{t}\sigma(e_i)$.
Later, we will set the number of steps to $t=\Theta(\lg n)$.

\textbf{Vectors from sequences of random walks.} Next, given a start vertex $u$,
we describe how to obtain a sparse vector $\mm_u\in\mathbb{R}^n$ based on a
sequence of lazy signed random walks.  In Sec.~\ref{sec:analysis-intuition},
we will argue that $\mm_u$ essentially serves as a (sparse) approximation of the
vector $\ppp_u^t \DD^{-1/2}$, where $\ppp_u^t = \1_u \WW^t$, $\WW$ is the walk matrix and $\DD$
is the degree matrix as defined in Sec.~\ref{sec:preliminaries}.

Suppose that we perform $R$ lazy signed random walks of length~$t$ from
vertex~$u$ and let $v_1,\dots,v_R$ be the vertices at which these random walks
finish with respective signs $s_1,\dots,s_R$. Now we define two vectors
$\mm_u^+,\mm_u^-\in\mathbb{R}_{\geq 0}^n$ as follows. We set $\mm_u^+(v)$ to the fraction of
random walks that ended in vertex $v$ with sign $+$, i.e.,
$\mm_u^+(v)=\frac{\abs{\{ i \colon v_i=v, \, s_i=+ \}}}{R}$ for all $v\in V$.
Similarly, we set 
$\mm_u^-(v)=\frac{\abs{\{ i \colon v_i=v, \, s_i=- \}}}{R}$.  Finally, we set
$\mm_u = (\mm_u^+ - \mm_u^-) \DD^{-1/2}$. 

Note that $\mm_u$ can have positive \emph{and negative} entries.
Furthermore, the $R$ random walks can end in at most~$R$
different vertices, and thus $\mm_u$ has at most $R$~non-zero
entries. Later, we will set $R=\tO(\sqrt{n} \poly(1/\varepsilon))$.

We now introduce our key subroutine \textsc{EstDotProd}($u$,$v$,$t$,$\alpha$)
and the pseudocode of the routine is presented in
Alg.~\ref{alg:estimate-dot-product}.
Consider two vertices $u$ and $v$, the random walk length $t$ and a technical parameter
$\alpha$ that we will set in the proof of Thm.~\ref{thm:oracle}. Then
\textsc{EstDotProd}($u$,$v$,$t$,$\alpha$) computes the two vectors $\mm_u$
and $\mm_v$ and calculates their dot product. This is repeated $h=O(\lg n)$
times and then the median of these dot products is returned. We will later
(Lem.~\ref{lem:estimate-dot-product}) show that the output of
\textsc{EstDotProd}($u$,$v$,$t$,$\alpha$) gives an approximation of
$\langle \ppp_u^t \DD^{-1/2}, \ppp_v^t \DD^{-1/2}\rangle$ with small error.

\textbf{Preprocessing.}
We present the preprocessing phase of the oracle 
in Alg.~\ref{alg:mainoracle}. 
Here, we make use of the fact (see Sec.~\ref{sec:analysis-intuition}
for details) that when two vertices $u$ and $v$ are from the same cluster, then
\begin{align*}
	\Delta_{uv}
		:= \min\{
			&\norm{\ppp_u^t \DD^{-1/2}- \ppp_v^t \DD^{-1/2}}^2_2, \\
			&\norm{\ppp_u^t \DD^{-1/2}+ \ppp_v^t \DD^{-1/2}}^2_2
			\}
\end{align*}
should be small, whereas if $u$ and $v$ are from different clusters then
$\Delta_{uv}$ should be large. However, since our algorithm has no
access to the vectors $\ppp_u^t$ and $\ppp_v^t$, we will need to use
\textsc{EstDotProd} to obtain an estimate $\tDelta_{uv}$ as approximation of
$\Delta_{uv}$.

The preprocessing starts by sampling a set $S$~of $O(k \log k)$ vertices.
Essentially this ensures that from each of the
$k$~clusters, $S$~contains at least one vertex.
Now for all pairs of vertices $u,v\in S$, we compute $\tDelta_{uv}$ as
approximation of $\Delta_{uv}$ by rewriting
the norms inside the definition $\Delta_{uv}$ as dot products and then
estimating each of these dot products using \textsc{EstDotProd}. More concretely,
we observe that
\begin{equation}
\label{eq:Delta_uv_dot_product}
	\hspace{-0.15cm}
\begin{aligned}
	\Delta_{uv}
	= \min\{ Y_{uu} + Y_{vv} - 2 Y_{uv},
			 Y_{uu} + Y_{vv} + 2 Y_{uv} \},
\end{aligned}
\end{equation}
where $Y_{ab} = \langle \ppp_a^t \DD^{-1/2}, \ppp_b^t \DD^{-1/2} \rangle$ for
$a,b\in V$.
Next, we let $X_{ab} = \textsc{EstDotProd}(a,b)$
and we define
\begin{align}
\label{eq:delta_uv}
	\hspace{-0.5cm}
	\tDelta_{uv} = \min\{ X_{uu}+X_{vv}-2X_{uv}, X_{uu}+X_{vv}+2X_{uv} \}.
\end{align}

We cluster the vertices in $S$ as follows. We create an auxiliary graph $H$
with vertex set $S$ and without edges. Then we add edges for each pair of
vertices $u$ and $v$ such that $\tDelta_{uv}$ is ``small'', i.e., if
$\tDelta_{uv} < \frac{1}{2dn}$.
In our proof of Thm.~\ref{thm:oracle}, we will show that if the conditions of
the theorem hold, then $H$ consists of $k$~cliques corresponding to the
$k$~clusters in the $(k,\betain,\betaout)$-clustering. Thus, the
preprocessing will correctly identify at least one vertex from
each cluster.

\textbf{Query procedure.}
For a query \textsc{WhichCluster}($v$), we proceed similarly to
how we clustered the vertices in $S$ in the preprocessing. More concretely,
given $v$ as input to the query, we compute $\tDelta_{uv}$ for all $u \in S$.
If $\tDelta_{uv}\leq\frac{1}{2dn}$ for some $u\in S$ then we return that $v$
belongs to the cluster of~$u$.  For the full details, see
Alg.~\ref{alg:answeringquery}.

\subsection{Main Technical Contribution}
\label{sec:main-technical-contribution}
To prove Thm.~\ref{thm:oracle}, our global proof strategy is as follows.
First, we show that two
vertices $u$ and $v$ are from the same cluster \emph{if and only if}
$\Delta_{uv}$ is small (see Lem.s~\ref{lem:close}
and~\ref{lem:dissimilar-distributions}). Second, we show that
$X_{uv}$ does not introduce too much error for estimating 
$\langle \ppp_v^t \DD^{-1/2}, \ppp_u^t \DD^{-1/2}\rangle$
(see Lem.~\ref{lem:estimate-dot-product}).
This then implies that with a large probability $\delta_{uv}$ is close to
$\Delta_{uv}$ and, thus, $\delta_{uv}$ is small
\emph{iff} $u$ and $v$ are from the same cluster.
We give more intuition and details of our proof strategy
in App.~\ref{sec:analysis-overview}.

This strategy is similar to the one by \citet{czumaj15testing} for unsigned
graphs. However, even though our global proof strategy is similar, we still have
to contribute a significant amount of new ideas to extend the clustering oracle
from the unsigned to the signed setting. We will now discuss some of these
challenges in more detail.

One particular challenge was showing that $\Delta_{uv}$ is small if $u$
and $v$ are from the same cluster (see Lem.~\ref{lem:close}).
To prove this, we need two technical lemmas that constitute our
main technical contribution.  One of them (Lem.~\ref{lem:small-diff}) provides
a connection between bicluster membership and the entries of the
eigenvectors of the signed normalized Laplacian. We believe this result is of
independent interest and will find further applications in the analysis of signed
graphs.

Let $\vv_1,\cdots,\vv_n$ be the orthonormal row eigenvectors of $\LL^\sigma$
s.t.\ $\lambda_i \vv_i = \vv_i \LL^\sigma$. Thus $\vv_i\vv_j^\top = 1$ if
$i=j$ and $\vv_i \vv_j^\top = 0$ otherwise. Let $\trvv_i= \vv_i\DD^{-1/2}$. For a
subset $U\subseteq V$, let $\mu_U=\vol_{G[U]}(U)$ be the total volume of the
induced graph $G[U]$, i.e., the sum of degrees of all vertices in $G[U]$. 

The first lemma says there is a gap between $\lambda_k$ and $\lambda_{k+1}$ if a
graph is $(k,\betain,\betaout)$-clusterable, which allows us to
use the first $k$ eigenvectors $\vv_1,\dots,\vv_k$ to bound
$\Delta_{uv}$.
The proof makes use of our new definition of inner signed bipartiteness ratio of
a graph. We defer the proof details to App.~\ref{app:proofofeigenvalues}.
\begin{restatable}{lemma}{eigenvalues}
\label{lem:cheeger-eigenvalues}
	If $G$ is signed and
	$(k,\betain,\betaout)$-clusterable, then
	$\lambda_i \leq 2\betaout$ for all $i\leq k$ and
	$\lambda_i \geq \frac{\betain^2}{C_1^2 (k+1)^6}$ for all $i\geq k+1$, where
	$C_1$ is the constant from Thm.~\ref{thm:signed-cheeger}.
\end{restatable}

For our second lemma consider a cluster $U$ with $\betainner(G[U])\geq \betain$
and $\beta_G(U)\leq \betaout$. Then there exists a partition of $U$ into
biclusters $V_1$ and $V_2$ such that $\beta_G(V_1,V_2)=\beta_G(U)\leq \betaout$,
i.e., $V_1$ and $V_2$ correspond to the two polarized groups inside cluster $U$.
Intuitively, our lemma asserts that for most vertices $u\in U$, the sign of the
entry $\trvv_i(u)$ ($i=1,\dots,k$) reveals whether $u\in V_1$ or $u\in
V_2$.  In other words, we show that each vector $\trvv_i$ ($i=1,\dots,k$),
approximately reveals a biclustering of $U$ into two polarized communities.

Slightly more precisely, we will show that if $i\leq k$ then for each
``typical'' vertex pair $u,v\in U$, it holds that $\trvv_i(u)\approx \trvv_i(v)$
if $u,v\in V_1$ or $u,v\in V_2$, and $\trvv_i(u)\approx - \trvv_i(v)$ otherwise.
To do so, we establish a novel connection between the structure of each balanced
cluster and the geometric embedding from these $k$ eigenvectors: we relate the
signed indicator vector $\1_{V_1,V_2}$ to the first eigenvector $\ww$ of (the
normalized signed Laplacian of) the subgraph $G[U]$, and we use the variational
characterization of the second eigenvector of $G[U]$ to analyze the $\trvv_i$
restricted on~$U$.  Here, $\1_{V_1,V_2}\in \mathbb{R}^{U}$ is the vector with
$\1_{V_1,V_2}(u)=1$ if $u\in V_1$ and $\1_{V_1,V_2}(u)=-1$ if $u\in V_2$.
We present the proof of the lemma in App.~\ref{app:proofsmalldiff}.

\begin{restatable}{lemma}{smalldiff}
\label{lem:small-diff}
Let $\alpha\in (0,1)$.  Suppose $G=(V,E,\sigma)$ is signed and
	$(k,\betain,\betaout)$-clusterable. Let $U$ be a cluster of $G$ with
	$\betainner(G[U])\geq \betain$ and $\beta_G(U)\leq \betaout$.
	Then there exists a partition $V_1,V_2$ of $U$, a subset $\widetilde{U}$ of
	$U$, and constants $c_i, 1\leq i\leq k$, such that $|c_i|\leq 3d$,
	$|\widetilde{U}|\geq (1-\alpha)|U|$, and for each $i\leq k$, 
\begin{itemize}
\item if $u\in V_1\cap \widetilde{U}$, then 
$\abs{\trvv_i(u)-c_i \cdot \frac{1}{\sqrt{\mu_U}}} \leq \frac{64d C_1}{\betain}\cdot\sqrt{\frac{\betaout}{\alpha \cdot \mu_U}}$
\item if $u\in V_2\cap \widetilde{U}$, then 
$\abs{\trvv_i(u)+ c_i \cdot \frac{1}{\sqrt{\mu_U}}} \leq \frac{64d C_1}{\betain}\cdot\sqrt{\frac{\betaout}{\alpha \cdot \mu_U}}$
\end{itemize}
	where $C_1$ is the constant from Thm.~\ref{thm:signed-cheeger}.
\end{restatable}

\section{Experiments}
\label{sec:experiments}
We experimentally evaluated our algorithms on a MacBook
Pro with 16~GB RAM and a 2~GHz Quad-Core Intel Core~i5.  Our
algorithms were implemented in C++11.%
We always performed 8~\textsc{WhichCluster}-queries in
parallel. %
See App.~\ref{sec:implementation} for more implementation details.
Our source code is available on
github.\footnote{\url{https://github.com/StefanResearch/signed-oracle}\label{fn:github}}

\textbf{Quality measure.}
In our evaluation, we consider a
set of ground-truth clusters $U_1,\dots,U_k$ and the output of an algorithm
$\tilde{U}_1,\dots,\tilde{U}_s$. We assume w.l.o.g.\ that $s\geq k$.
The \emph{accuracy} of the clustering
is $\min_{\pi} \frac{1}{m} \sum_{i=1}^k \abs{U_i \cap \tilde{U}_{\pi(i)}}$,
where $\pi\colon[k]\to[s]$ is an injective function and $m=\sum_i
\abs{U_i}$ is the number of elements in the ground-truth clusters. Thus, the
accuracy %
measures how many elements were classified correctly; since $\pi$ is injective,
each $U_i$ must be mapped to a different $\tilde{U}_j$.  When a cluster
$\tilde{U}_j$ contains an element $v \not\in \bigcup_i U_i$, we remove $v$ from
$\tilde{U}_j$ (this can happen %
when we do not have ground-truth information for all vertices).

\textbf{Algorithms.}
First, we consider our signed oracles \rwseeded and \rwunseeded, resp., obtain
ground-truth seed nodes or randomly picked vertices in the preprocessing (see
App.~\ref{sec:implementation}).  Second, we implemented two \emph{un}signed
oracles which operate on the underlying unsigned graph (see
App.~\ref{sec:implementation}); we denote them \rwuseeded and \rwuunseeded,
depending on their initialization.

As baselines we use \FOCG by \citet{chu16finding} and
\polarSeeds by \citet{xiao2020searching}.
\FOCG is a \emph{global} algorithm that requires access to the full graph and
enumerates nearly-balanced communities. \polarSeeds is a local algorithm that
requires some seed nodes as input and explores the graph locally to find a
subgraph with small signed bipartiteness ratio.
We used the implementations provided by the authors and ran them with the
default parameters.

Since our focus was on practically efficient oracle data structures, we did not
compare against the oracle by \citet{gluch21spectral}, as
it is of highly theoretical nature, and involves several subroutines that hinder
the implementation in practice.\footnote{For instance, Alg.~10 in the arxiv
	version of \citet{gluch21spectral} samples a set of $\Omega(k^4)$ vertices
	and then enumerates all possible partitions of this set. This would be
	infeasible in practice.}

We consider two types of experiments:
(1)~\emph{Clustering} a graph $G$ into polarized communities $U_1,\dots,U_k$ (as
per Def.~\ref{def:clustering}) and (2)~\emph{biclustering} $G$ into
opposing polarized groups $(V_{1},V_{2}),(V_3,V_4),\dots,(V_{2k-1},V_{2k})$. For
the biclustering setting, we consider a heuristic of our oracle which does not
take absolute values when computing $\mm_x$ (see
App.~\ref{sec:implementation} for details). In both cases, we use the
corresponding versions of the algorithms; to avoid blowing up the
notations, we use the same algorithm names for the clustering and
biclustering versions.

\textbf{Evaluation on synthetic data.} Due to lack of space, we present our
experiments on synthetic data in App.~\ref{sec:experiments-synthetic}.
Our experiments show that our oracles outperform the baselines
when the clusters are large. The unsigned oracle works well for clustering
but only the signed oracle can recover the biclusters $(V_{2i-1},V_{2i})$.
Also, the seeded oracles outperform the unseeded oracles and our
oracle scale linearly in the number the number of steps and the walk lengths.

We also note that \FOCG{} and \polarSeeds{} do not perform very well on the
synthetic datasets.  We believe the main reason is that these algorithms
were built to find ``small'' clusters which do not necessarily partition the
graph; in contrast, our method is strongest in the presence of large clusters
that partition the graph.  This large-vs-small cluster intuition is also
corroborated by our experiments on synthetic data in
Figure~\ref{fig:results_randomGraph_k_accuracy} in the appendix: as the
number~$k$ of clusters increases, the clusters get smaller and the performance
of \polarSeeds{} improves. If we increased~$k$ further, the performance of
\polarSeeds{} would improve further and eventually outperform our methods.

\textbf{Evaluation on real-world data.}
Since we are not aware of public signed graph datasets with a small number of
large ground-truth communities, we created our own real-world data. We make them
available on github.\footref{fn:github}

We obtained
our graphs from English-language Wikipedia by considering Wikipedia pages about
politicians and the articles linked on their pages.  We selected five countries
(UK, Germany, Austria, Spain and US) and for each we selected a number
of politicians (UK:~2307, Germany:~1444, Austria:~190, Spain:~546, US:~5053);
this gives the ground-truth clusters $U_i$.  The politicians from the clusters
$U_i$ belong to one of two opposing parties, which splits each $U_i$ into
opposing groups $V_{2i-1}$ and $V_{2i}$.  In our graphs, the vertices
correspond to Wikipedia pages of politicians and articles that are linked
on the politicians' pages.  An edge~$(u,v)$ indicates that page~$u$ has a link
to page~$v$.  The sign for an edge~$(u,v)$ is $-$ if $u$ and $v$ are
politicians from the same country and they are in opposing parties (e.g.,
Democrats and Republicans in the US); otherwise, we set the sign to~$+$.

\begin{table*}[t]
  \caption{Statistics for real-world datasets. Here, $\abs{E^-}$ denotes the number
	  of negative edges, $\overline{\text{deg}}$ and $\text{deg}_{\max}$
	  denote the average degree and maximum degree, and
	  $\abs{V}_{\text{labeled}}$ denotes the number of vertices with
	  ground-truth labels.
  }
  \centering
  \label{tab:exp:properties}
  \smallskip
  \begin{tabular}{lrrrrrr}
    \toprule
    Dataset & $\abs{V}$ & $\abs{E}$ & $\abs{E^-}/\abs{E}$ &
	$\overline{\text{deg}}$ &  $\text{deg}_{\max}$ & $\abs{V}_{\text{labeled}}$ \\
    \midrule
	\WikiSmall & 9\,211 & 395\,038 & 0.39 & 140.3 & 1\,503 & 9\,211 \\
	\WikiMedium & 34\,404 & 904\,768 & 0.28 & 52.6 & 3\,407 & 9\,453 \\
	\WikiLarge & 258\,259 & 3\,187\,096 & 0.08 & 24.7 & 6\,017 & 9\,540 \\
    \bottomrule
  \end{tabular}
\end{table*}

We created three dataset. \WikiLarge contains all politicans and all
articles linked on their pages; we included all edges that contain at least one
politician. %
On \WikiLarge, the signed bipartiteness ratios of the three large communities
(UK, Germany, US) is $\approx0.66$ and for the smaller communities (Austria,
Spain) it is $\approx0.9$.  We also consider two smaller versions of \WikiLarge:
\WikiSmall is the largest connected component of $G[P]$, where $G$ is the graph
given by \WikiLarge and $P$ is the set of politician nodes in \WikiLarge;
$\WikiMedium$ is the largest connected component of $G[P \cup V']$, where $V'$
is a randomly sampled set containing 10\% of the non-politician nodes from
\WikiLarge.
Note that for \WikiLarge and
\WikiMedium we only have a ground-truth clustering for a \emph{subset} of the vertices
(namely, for the politician nodes). We present statistics of the datasets
in Table~\ref{tab:exp:properties}.  
In our experiments, we use the undirected versions
of these datasets, i.e., each directed edge is made undirected.

\begin{figure*}[!htb]
  \begin{center}
  \subfigure{
	  \includegraphics[width=4\smallfigwidth]{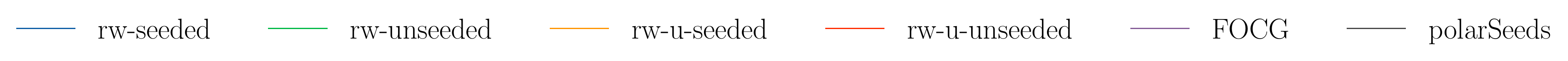}
  }
  \addtocounter{subfigure}{-1}

  \centering
  \subfigure[\small Vary \#steps: clustering]{
    \label{fig:wikiSmall_vary_numSteps_accuracy}
    \includegraphics[height=\smallfigheight,width=\smallfigwidth]{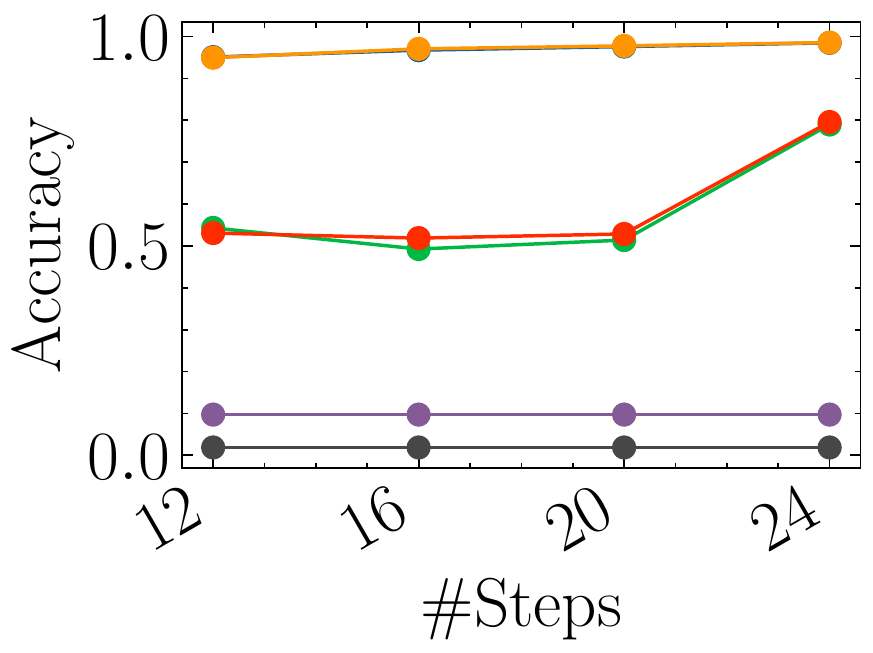}
  }\hspace*{\smallfigsep}
  \subfigure[\small Vary \#walks: clustering]{
    \label{fig:wikiSmall_vary_numWalks_accuracy}
    \includegraphics[height=\smallfigheight,width=\smallfigwidth]{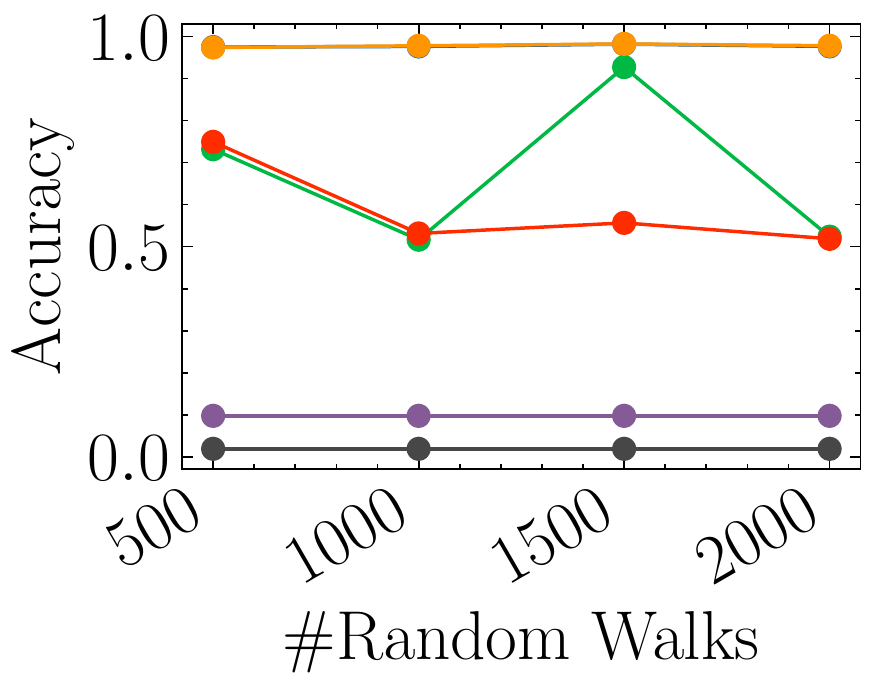}
  }\hspace*{\smallfigsep}
  \subfigure[\small Vary \#steps: biclustering]{
    \label{fig:wikiSmall_vary_numSteps_accuracy_biclustering}
    \includegraphics[height=\smallfigheight,width=\smallfigwidth]{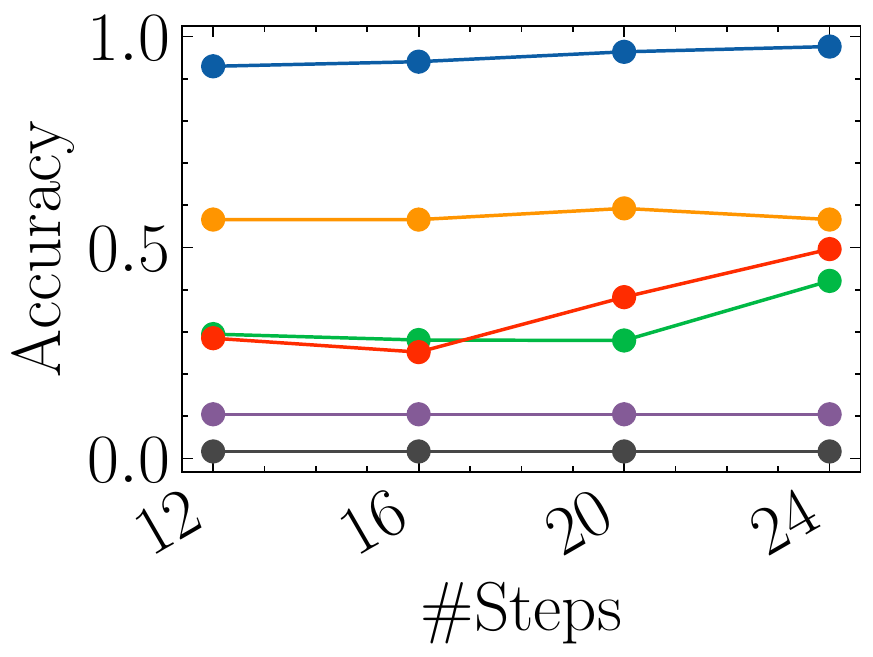}
  }\hspace*{\smallfigsep}
  \subfigure[\small Vary \#walks: biclustering]{
    \label{fig:wikiSmall_vary_numWalks_accuracy_biclustering}
    \includegraphics[height=\smallfigheight,width=\smallfigwidth]{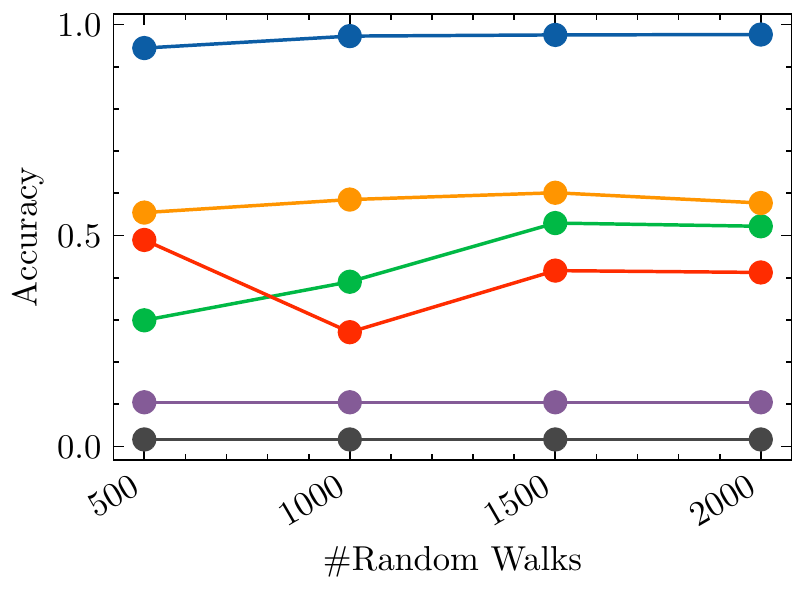}
  } \\
  \end{center}
  \caption{Accuracies of the algorithms on \WikiSmall.
	  We consider clustering
		  (Figs.~\subref{fig:wikiSmall_vary_numSteps_accuracy} and~\subref{fig:wikiSmall_vary_numWalks_accuracy})
		and biclustering
		(Figs.~\subref{fig:wikiSmall_vary_numSteps_accuracy_biclustering}
		 and~\subref{fig:wikiSmall_vary_numWalks_accuracy_biclustering}) experiments,
		with varying number of steps
			(Figs.~\subref{fig:wikiSmall_vary_numSteps_accuracy} and~\subref{fig:wikiSmall_vary_numSteps_accuracy_biclustering})
		and varying number of random walks
		(Figs.~\subref{fig:wikiSmall_vary_numWalks_accuracy} and~\subref{fig:wikiSmall_vary_numWalks_accuracy_biclustering}).
		Since \FOCG and \polarSeeds do not use the number of steps and random
		walks as input, we only ran them once.
  }
  \label{fig:wikiS-experiments}
\end{figure*}

\begin{figure*}[!htb]
  \begin{center}
  \subfigure{
	  \includegraphics[width=3\smallfigwidth]{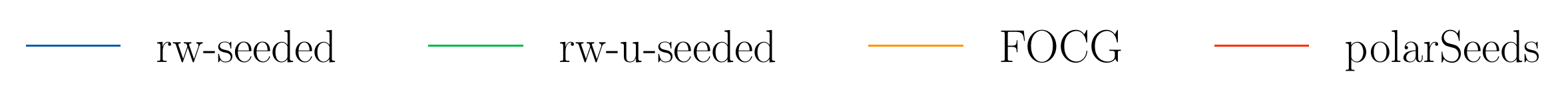}
  }
  \addtocounter{subfigure}{-1}

  \centering
  \subfigure[\small \WikiLarge: clustering]{
    \label{fig:wikiLarge_vary_numWalks_accuracy}
    \includegraphics[height=\smallfigheight,width=\smallfigwidth]{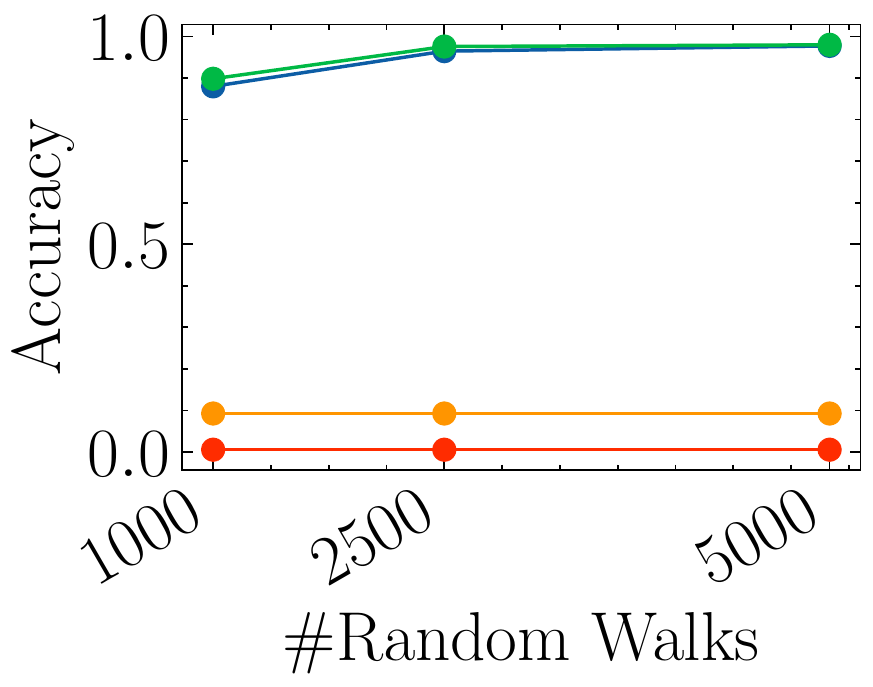}
  }\hspace*{\smallfigsep}
  \subfigure[\small \WikiLarge: avg.\ query time]{
    \label{fig:wikiLarge_vary_numWalks_time-norm}
    \includegraphics[height=\smallfigheight,width=\smallfigwidth]{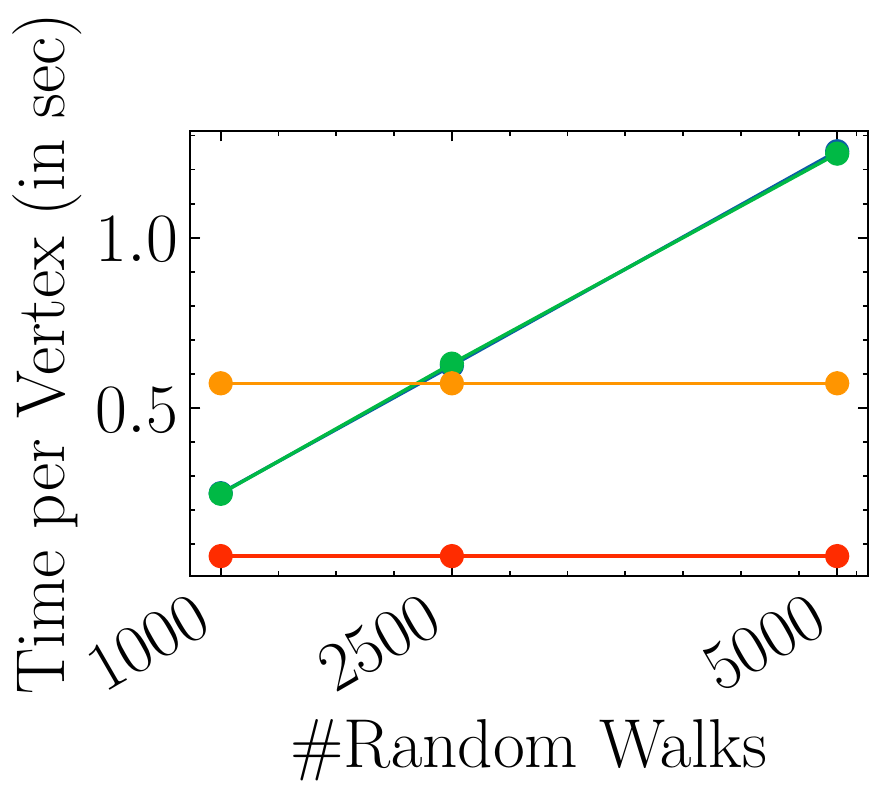}
  }\hspace*{\smallfigsep}
  \subfigure[\small \WikiMedium: biclustering]{
    \label{fig:wikiMedium_vary_numWalks_accuracy_biclustering}
    \includegraphics[height=\smallfigheight,width=\smallfigwidth]{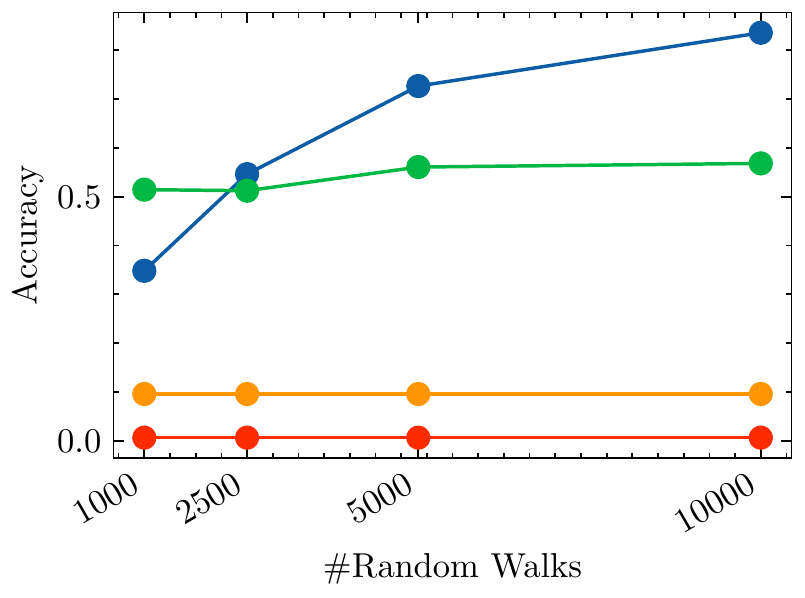}
  }\hspace*{\smallfigsep}
  \subfigure[\small \WikiLarge: biclustering]{
    \label{fig:wikiLarge_vary_numWalks_accuracy_biclustering}
    \includegraphics[height=\smallfigheight,width=\smallfigwidth]{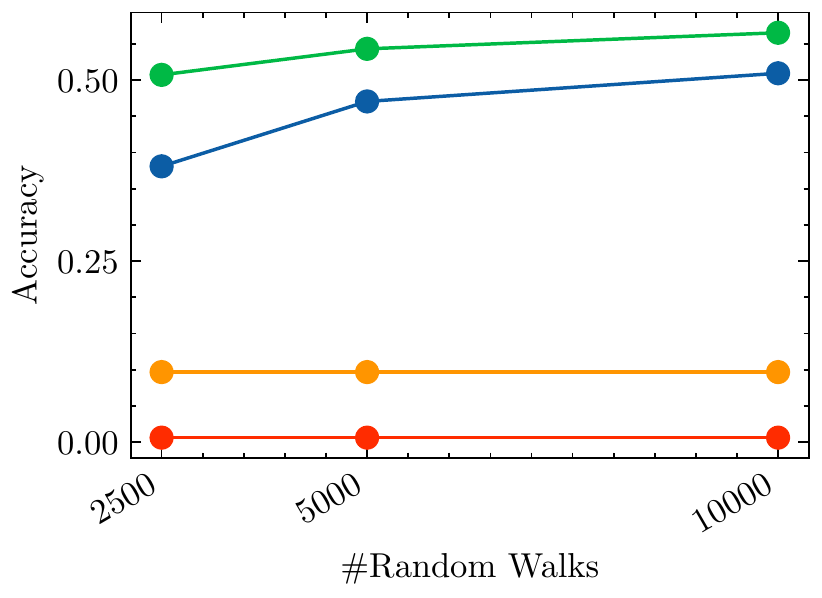}
  } \\
  \end{center}
  \caption{Results on \WikiMedium and \WikiLarge for varying numbers of random
	  walks.  We present accuracies for clustering on \WikiLarge
		  (Fig.~\subref{fig:wikiLarge_vary_numWalks_accuracy}),
		  and for biclustering in \WikiMedium (Fig.~\subref{fig:wikiMedium_vary_numWalks_accuracy_biclustering})
		  and in \WikiLarge (Fig.~\subref{fig:wikiLarge_vary_numWalks_accuracy_biclustering}).
	  The running time per vertex for clustering in 
	  \WikiLarge is given in Fig.~\subref{fig:wikiLarge_vary_numWalks_time-norm}.
		Since \FOCG and \polarSeeds do not use random
		walks as input, we only ran them once.
	}
  \label{fig:wikiML-experiments}
\end{figure*}

\emph{Experiments on \WikiSmall.}
Fig.~\ref{fig:wikiS-experiments} shows our results on \WikiSmall.
Unless stated otherwise, our
oracles used $1000$~random walks of length~$20$. For clustering (biclustering)
experiments, the seeded oracles obtained
$10$~($5$)~seeds for each $U_i$ ($V_i$); the unseeded oracles
sampled $5k$~vertices, where $k=5$ for clustering and $k=10$ for biclustering.  For the clustering
experiments~(Figs.~\ref{fig:wikiSmall_vary_numSteps_accuracy}
		and~\ref{fig:wikiSmall_vary_numWalks_accuracy}), the seeded oracles
achieve almost perfect accuracy; the unseeded methods are worse and benefit from
longer random walks (Fig.~\ref{fig:wikiSmall_vary_numSteps_accuracy}).
For the biclustering
experiments~(Figs.~\ref{fig:wikiSmall_vary_numSteps_accuracy_biclustering}
		and~\ref{fig:wikiSmall_vary_numWalks_accuracy_biclustering}),
\rwseeded is by far the best method and achieves excellent
accuracies. \rwuseeded consistently achieves accuracies above 50\% but below
60\%; this suggests that \rwuseeded successfully finds
the clusters $U_i$ (as shown by the clustering results) but, as it
ignores the edge signs, it places the vertices into the biclusters $V_{2i-1}$ and
$V_{2i}$ only slightly better than random.  For biclustering, the
unseeded methods do not perform well. \FOCG and
\polarSeeds achieve low accuracies, since they return too small clusters.

\emph{Experiments on \WikiMedium and \WikiLarge.}
Fig.~\ref{fig:wikiML-experiments} presents our results on the larger
datasets. We did not run the unseeded oracles, since in \WikiMedium and
\WikiLarge not all vertices are contained in ground-truth communities.
We used random walks of length~$24$ for \WikiMedium and of length~$20$ for
\WikiLarge; we initialized the seeds as for \WikiSmall.
Fig.~\ref{fig:wikiLarge_vary_numWalks_accuracy}
shows that even on \WikiLarge, the seeded oracles find the clusters $U_i$ with
almost perfect accuracy.  Furthermore, the algorithms
scale linearly in the number of random walks and on average queries takes less
than 1.4~seconds
(Fig.~\ref{fig:wikiLarge_vary_numWalks_time-norm}).
However, compared to \WikiSmall, we obtain lower accuracies for the biclustering
experiments: while on \WikiMedium, \rwseeded still achieves accuracies over 83\%
with enough random walks
(Fig.~\ref{fig:wikiMedium_vary_numWalks_accuracy_biclustering}),
for \WikiLarge even with 10\,000 random walks, \rwseeded only achieves an
accuracy of 50\%
(Fig.~\ref{fig:wikiLarge_vary_numWalks_accuracy_biclustering}).
We blame this on the fact that in \WikiLarge 
the ground-truth clusters are relatively small (they contain only 3.7\% of the
vertices). Additionally, in \WikiLarge the fraction of negative edges is
only~8\%, and thus only few random walks will encounter a negative edge and
\rwseeded cannot benefit from the edge-sign information. As before, \FOCG and
\polarSeeds have low accuracies.%

\emph{Conclusion.} Our oracles %
are successful for finding polarized communities $U_i$, even when the graphs do
not have bounded degrees (as in our theoretical analysis).  \rwseeded is
successful in finding the biclusters $(V_{2i-1},V_{2i})$, as long as they are
large enough and there are enough negative edges.  

\section{Conclusion}
\label{sec:conclusion}
We presented a local clustering oracle for signed graphs. Given a vertex~$u$, the oracle can
return the cluster membership of $u$ in sublinear time. Such a data structure is
desirable when the input graphs are large and the cluster membership is
only required for a small number of vertices. We proved that if the graph
satisfies a clusterability assumption, then the oracle returns correct cluster
memberships for a $(1-\varepsilon)$-fraction of the vertices w.r.t.\ a hidden
set of ground-truth clusters. We also evaluated the oracle practically,
showing that it achieves good results for large clusters.

In the future it will be interesting to provide a theoretical analysis for our
biclustering heuristic; here, ideas from \citet{trevisan2012max} might be
helpful. From a practical point of view it would be interesting to obtain
improvements for our biclustering heuristic, which allow to find small
biclusters $(V_{2i-1},V_{2i})$ when there are few negative edges.  Two
directions for this might be as follows: (1)~For a node $u$, first identify the
cluster $U_i$ of $u$ using the unsigned oracle and then use auxiliary
information to decide whether $u$ belongs to $V_{2i-1}$ or to $V_{2i}$. (2)~When
the underlying graph contains few negative edges, bias the random walks of
\rwseeded such that it takes disproportionately many negative edges.

\section*{Acknowledgements}
This research is supported by the the ERC Advanced Grant REBOUND~(834862), the
EC H2020 RIA project SoBigData++~(871042), and the Wallenberg AI, Autonomous
Systems and Software Program~(WASP) funded by the Knut and Alice Wallenberg
Foundation. P.P. is supported by ``the Fundamental Research Funds for the Central Universities''.

\bibliographystyle{icml2022}
\bibliography{main}

\pagebreak

\appendix

\section{Overview of the Appendix}
The appendix is organized as follows:
\begin{itemize}
	\item Appendix~\ref{sec:motivation-inner}: We provide further motivation
		for our choice of the inner signed bipartiteness ratio.
	\item Appendix~\ref{sec:pseudocode}: We present the pseudocode for our
		algorithms.
	\item Appendix~\ref{sec:analysis-overview}: We give an overview of our proof
		strategy.
	\item Appendix~\ref{sec:proofs}: We present the full proofs for all claims in
		the main text.
	\item Appendix~\ref{sec:implementation}: We give details on the implementations
		of our algorithms, including parameter tuning.
	\item Appendix~\ref{sec:experiments-synthetic}: We evaluate our algorithm on
		synthetically generated data.
\end{itemize}

\section{Further Motivation of the Inner Signed Bipartiteness Ratio}
\label{sec:motivation-inner}
We provide further motivation for the inner signed bipartiteness ratio.
Recall that we set 
$$ \betainner(G)
	:= \min_{\emptyset \neq U \subseteq V \colon \vol(U) \leq \frac{1}{2} \vol(G)}
		\beta_{G}(U).$$

First, let us justify our intuition that if $\betainner(G)$ is large, then
$\betainner(G)$ cannot be decomposed into two nearly-balanced communities.  To
make this more formal, recall the definition of
$\beta_2(G) = \min_{U_1,U_2} \max_{i=1,2} \beta_G(U_i)$, where the minimum is
taken over all partitions $U_1,U_2$ of $V$ with $U_1,U_2\neq\emptyset$.  Now
note that if we could split $G$ into two nearly-balanced communities, then we
would have that $\beta_2(G)$ is small.  Thus, we show that if $\betainner(G)$ is
large then $\beta_2(G)$ is at least as large. This justifies our informal
intuition above.
\begin{lemma}
\label{lem:beta-inner-beta-two}
	It holds that $\beta_2(G) \geq \betainner(G)$.  In particular, if
	$\betain \in \mathbb{R}_{\geq 0}$ and $\betainner(G) \geq \betain$ then
	$\beta_2(G)\geq\betain$.
\end{lemma}
\begin{proof}
	We prove the first claim by contradiction, i.e., suppose that
	$\betainner(G)=\betain$ and $\beta_2(G)<\betain=\betainner(G)$. Then there
	exists a partition $U_1,U_2$ of $V$ such that $\beta_G(U_1) < \betain$ and
	$\beta_G(U_2) < \betain$. Next, since $U_1$ and $U_2$ form a partition of
	$V$, we must have that $\vol(U_1)\leq\frac{1}{2}\vol(G)$ or
	$\vol(U_2)\leq\frac{1}{2}\vol(G)$.  This implies that there exists a set
	$U\in\{U_1,U_2\}$ such that $\vol(U)\leq\frac{1}{2}\vol(G)$ and
	$\beta_G(U) < \betain$. Thus, $\betainner(G)<\betain$. But this contradicts
	our assumption that $\betainner(G)=\betain$.

	The second claim of the lemma immediately follows from applying the first
	claim.
\end{proof}

\textbf{Why do we not assume that $\beta(G[U_i])$ is large?}
Next, we discuss why we cannot replace the inner signed bipartiteness ratio
$\betainner(G)$ with the (classic) signed bipartiteness ratio $\beta(G[U_i])$ in
Def.~\ref{def:clustering}.  To see this, consider
Def.~\ref{def:clustering} with $\betainner(G[U_i]) \geq \betain$ replaced
by $\beta(G[U_i])\geq \betain$ for all $i=1,\dots,k$. Again, we consider the the
setting with $\betain > \betaout$.  Intuitively, this would mean that subgraph
$G[U_i]$ is ``far from balanced'' on its inside (since $\beta(G[U_i])$ is large)
while outside it is close to balanced (since $\beta_G(U_i)$ is small). 

Unfortunately, we show that if $\betain > \betaout$ then no graph can satisfy
this new definition. Before we give a more general proof, consider the following
example. Consider a graph $G$ which consists of two positive cliques among
vertices $V_1$ and $V_2$ and in between $V_1$ and $V_2$ there is biclique
consisting only of negative edges. Then $\beta(G)=0$ and $\beta(V_1,V_2)=0$ but
also $\beta(G[V_1])=0$ and $\beta(G[V_2])=0$ (since graphs with only positive
edges are balanced).

Now we give a more general result showing that for any $U\subseteq V$, we have
that $\beta(G[U])\leq \beta_G(U)$. Therefore, the previously proposed definition
that avoids the inner signed bipartiteness ratio cannot work: it implies that
$\betain \leq \beta(G[U_i]) \leq \beta_G(U_i)\leq \betaout$ but this contradicts
our assumption that $\betain > \betaout$.
\begin{lemma}
\label{lem:betain-betaout}
For any two disjoint subsets $V_1,V_2\subseteq V$, it holds that
$\beta(G[V_{1}\cup V_{2}]) \leq \beta_G(V_{1}, V_{2})$.
Furthermore, for any $U\subseteq V$, it holds that 
	$\beta(G[U]) \leq \beta_G(U)$.
\end{lemma}
\begin{proof}
	We have that
	\begin{align*}
		&\beta(G[V_{1}\cup V_{2}]) \\
		&\leq \beta_{G[V_{1}\cup V_{2}]}(V_{1},V_{2}) \\
		&=
			\frac{
				2\abs{E_G^+(V_{1},V_{2})}
				+ \abs{E_G^-(V_{1})}
				+ \abs{E_G^-(V_{2})}
			}{ \vol_{G[V_{1}\cup V_{2}]}(V_{1}, V_{2}) } \\
		&\leq 
			\frac{
				2\abs{E_G^+(V_{1},V_{2})}
				+ \abs{E_G^-(V_{1})}
				+ \abs{E_G^-(V_{2})}
				+ \abs{E_G(V_{1}\cup V_{2}, \overline{V_{1}\cup V_{2}})}
			}{
				\vol_{G[V_{1}\cup V_{2}]}(V_{1}, V_{2})
				+ \abs{E_G(V_{1}\cup V_{2}, \overline{V_{1}\cup V_{2}})}
			} \\
		&= \beta_G(V_{1}, V_{2}),
	\end{align*}
	where we have used that for $a,b,c > 0$ and $a \leq b$, it
	holds that $\frac{a}{b} \leq \frac{a+c}{b+c}$.

Now for any subset $U$, let $V_1,V_2$ be a partition of $U$ such that $\beta_G(U)=\beta_{G}(V_1,V_2)$. Then by the above calculation, we have that 
\[
\beta(G[U])\leq \beta_{G[U]}(V_1,V_2)\leq \beta_G(V_1,V_2) =\beta_G(U).
\qedhere
\]
\end{proof}

\textbf{Is the underlying unsigned graph clusterable?}
Next, we argue that our condition from Def.~\ref{def:clustering} is
not implied by previous definitions that were based on the conductance of
\emph{unsigned} graphs (e.g., in~\cite{czumaj15testing}).
In other words, it is possible that our methods finds the planted clusters in
the signed graph while this would not be possible by only looking at the
underlying unsigned graph.

First, recall that for an unsigned graph $\Gun=(V,E)$ and a set of vertices
$S\subseteq V$, the \emph{conductance} of $S$ is given by
$\phi_{\Gun}(S)=\frac{\abs{E(S,V\setminus S)}}{\vol(S)}$. Furthermore, we let
$\phi(\Gun)$ denote the conductance of $\Gun$ which is given by
$\phi(\Gun) = \min\{\phi_{\Gun}(S) \colon \vol(S) \leq \vol(\Gun)/2 \}$.

Now consider the unsigned graph $\Gun$ that is obtained by removing
all the signs on the edges from $G$. We argue that a
$(k,\betain,\betaout)$-clustering of $G$ does not imply a
$(k,\betain,\betaout)$-clustering of $\Gun$, i.e., $\Gun$ contains a
$k$-partition $C_1,\dots,C_k$ such that the \emph{inner conductance} of $C_i$,
denoted $\phi(\Gun[C_i])$, is at least $\betain$ and the \emph{outer
conductance} of $C_i$, denoted $\phi_G(C_i)$, is at most $\betaout$ for each
$i$. Therefore, \emph{one could \textbf{not} apply the previous clustering
oracle for conductance-based clustering of $\Gun$ and recover the underlying
clusters in our problem}. We give a brief explanation next.

Suppose that $U_1,\dots,U_k$ is a $(k,\betain,\betaout)$-clustering of $G$.
Then, indeed, it is true that each $U_i$ has small outer conductance in $\Gun$,
since we have that $\phi_{\Gun}(U_i)\leq \beta_G(U_i)\leq \betaout$.  However,
the inner conductance of $U_i$ in $\Gun$ can be arbitrarily small: Even though
the inner signed bipartiteness ratio of $U_i$ is large, it can happen that there
is a very small subset $S_i\subseteq U_i$ (say, of size $O(\log n)$) such that
there are almost no edges leaving $S_i$ but all edges in $S_i$ have sign $-$
and $S_i$ is far from being balanced. Thus, there exists a subset $S_i$ in $U_i$
whose (outer) conductance is almost $0$ in $\Gun$ but the inner signed
bipartiteness ratio of $S_i$ and $U_i$ is large.

Note that the previous example with the set $S_i$ also shows that it is possible
that $G$ contains a sparse cut and, therefore, $\ppp_u^t$ does not converge to
the uniform distribution of $V$.

\section{Pseudocode for Our Algorithms}
\label{sec:pseudocode}

\begin{algorithm}[H]
\caption{Estimating the dot product $\langle \ppp_u^t \DD^{-1/2}, \ppp_v^t \DD^{-1/2}\rangle$}
\label{alg:estimate-dot-product}
	\begin{algorithmic}[1]
	\Procedure{\textsc{EstDotProd}}{$u,v,t,\alpha$} 
	\State $R\gets\frac{40000d^2 k^{1.5}\sqrt{n}}{\alpha^{1.5}}$
	\ForAll{$i=1,\cdots,h=O(\log n)$}
		\For{$x\in \{u,v\}$}
			\State Perform $R$ lazy signed random walks of length $t$ starting at vertex~$x$ with sign $+$.
			\For{each $w\in V$}
				\State $\mm^+_x(w)\gets $ the fraction of walks that end at $w$ with sign~$+$.
				\State $\mm^-_x(w)\gets $ the fraction of walks that end at $w$ with sign~$-$.
				\State\label{line:m_x} $\mm_x \gets (\mm_x^+ - \mm_x^-) \DD^{-1/2}$.
			\EndFor
		\EndFor
		\State $\chi_i\gets \langle \mm_u, \mm_v \rangle$.
	\EndFor
	\State Let $X_{uv}$ be the median value of $\chi_1,\cdots,\chi_h$. 
	\State \Return $X_{uv}$
	\EndProcedure
	\end{algorithmic}
\end{algorithm}

\begin{algorithm}[H]
\caption{Preprocessing: Constructing a clustering oracle}
\label{alg:mainoracle}
	\begin{algorithmic}[1]
		\Procedure{\textsc{BuildOracle}} {$G,k,d,\betain,\varepsilon,\gamma$}
			\State $s\gets \frac{20k \log({k})}{\gamma}$,
				$\alpha\gets\frac{\varepsilon}{90s}$,
				$t\gets \frac{C''k^6 d^3 \log n}{\varepsilon\cdot\betain^2}$
				for constant $C''$
			\State Sample a set $S$ of $s$ vertices independently and uniformly at
					random from~$V$.

		\For{$v\in S$} \label{alg:l2norm}
			\Comment{test if $||\ppp_v^t \DD^{-1/2}||_2^2=O(\frac{k^2\log(k)}{\gamma\varepsilon\cdot n})$}
			\State $X_{vv}\gets $\textsc{EstDotProd}($v,v,t,\alpha$) 
			\If{$X_{vv} \geq \frac{4000k^2\log(k)}{\gamma\varepsilon\cdot n}$}
				\State\label{alg:l22normest} abort and \Return \textbf{Fail}
			\EndIf
		\EndFor
		\State let $H$ be an empty graph with vertex set $S$
		\For{each pair $u, v \in S$}  \label{alg:distribution-closeness}
			\Comment{test if $\Delta_{uv}\le \frac{1}{4nd}$}
			\State $X_{vv}\gets $\textsc{EstDotProd}($v,v,t,\alpha$) 
			\State $X_{uu}\gets $\textsc{EstDotProd}($u,u,t,\alpha$) 
			\State $X_{uv}\gets $\textsc{EstDotProd}($u,v,t,\alpha$) 
			\State\label{alg:deltaestimate} $\delta_{uv}\gets \min\{X_{vv}+X_{uu}-2X_{uv},X_{vv}+X_{uu}+2X_{uv}\}$
			\If{$\delta_{uv}\leq \frac{1}{2dn}$}
				\State\label{alg:addedge} add edge $(u,v)$ to $H$
			\EndIf
		\EndFor
		\If{$H$ is the union of $k$ connected components (CCs)}
		\State label the components by ``$1,2,\dots, k$' 
		\State label each vertex $u\in V_H$ with the same index as its component, denoted $\ell(u)$
		\State \Return $H$ and its vertex labeling $\ell$
		\Else
		\State \Return \textbf{Fail}
		\EndIf
		\EndProcedure
	\end{algorithmic}
\end{algorithm}  

\begin{algorithm}[H]
\caption{Answering the community membership of a vertex $v$}
\label{alg:answeringquery}
\begin{algorithmic}[1]
\Procedure{\textsc{WhichCluster}}{$G,v,H,\ell$} 
	\For{$u\in V_H$} 
		\State $X_{vv}\gets $\textsc{EstDotProd}($v,v,t,\alpha$) 
		\State $X_{uu}\gets $\textsc{EstDotProd}($u,u,t,\alpha$) 
		\State $X_{uv}\gets $\textsc{EstDotProd}($u,v,t,\alpha$) 
		\State\label{alg:estimatedeltaagain} $\delta_{uv}\gets \min\{X_{vv}+X_{uu}-2X_{uv},X_{vv}+X_{uu}+2X_{uv}\}$
		\If{$\delta_{uv}\leq \frac{1}{2dn}$}
		\State abort and \Return the label $\ell(u)$
		\EndIf
	\EndFor
	\State \Return a random number from $\{1,2,\dots,k\}$
\EndProcedure
\end{algorithmic}
\end{algorithm}

\section{Analysis Overview}
\label{sec:analysis-overview}
In this section, we give an overview of the analysis and provide the main
technical lemmas of our analysis.  All missing proofs can be found in
App.~\ref{sec:proofs}.

\textbf{Intuition.}
\label{sec:analysis-intuition}
We begin by providing some intuition for our algorithm and our analysis.

We start by establishing some properties of lazy signed random walks.
First, suppose that we perform the lazy random walks \emph{without the signs}
on the underlying unsigned graph $\Gun$ with the corresponding
transition probability matrix
$\WW^{\textrm{un}}:=\frac{\II+\DD^{-1}\AAA^{\textrm{un}}}{2}$, where $\AAA^{\textrm{un}}$ is
the adjacency matrix of $\Gun$. Then the probability that a random
walk started at vertex~$u$ ends in vertex~$v$ is $\Prob{v_t=v}=\qq_u^t(v)$, where
$\qq_u^t = \1_u (\WW^{\textrm{un}})^t$.  Second, when we \emph{add the sign}, then we
are interested in the quantity
\begin{align*}
	\ppp_{u}^t(v) = \Prob{v_t=v, s_t=+} - \Prob{v_t=v, s_t=-},
\end{align*}
i.e., $\ppp_u^t(v)$ is the probability of reaching $v$ with a positive sign
\emph{minus} the probability of reaching $v$ with a negative sign.  Observe that
$\ppp_{u}^t(v)$  can be described by the walk matrix
$\WW=\frac{\II+\DD^{-1}\AAA^\sigma}{2}$, i.e., $\ppp_{u}^t(v)=[\1_{u}\WW^t](v)$ for all $v$.
Note that while $\qq_u^t$ (for unsigned random walks)
gives a distribution, this is not the case for $\ppp_u^t$. In fact, $\ppp_u^t$ can
potentially even contain negative entries and $\ppp_u^t$ does not necessarily
converge to the uniform distribution of $V$ as it can happen that $V$ contains a
sparse cut (we discuss this in App.~\ref{sec:motivation-inner}).
In the following,
we call such a vector $\ppp_{u}^t$ the \emph{discrepancy vector} of a lazy
signed random walk of length $t$ starting from $u$.  

Next, consider a signed graph $G$ and a $(k,\betain,\betaout)$-clustering
$U_1,\dots,U_k$ of $G$ as per Def.~\ref{def:clustering}. Let
$U\in\{U_1,\dots,U_k\}$ be one of the clusters.
Since by assumption we have that $\beta_G(U)\leq \betaout$, there exists a
partition of $U$ into subsets $V_1$ and $V_2$ with $\beta_G(V_1,V_2)\leq \betaout$. 
Then, intuitively, for a typical vertex $u\in U$, a \emph{short} random walk
starting from $u$ of has the following properties:
\begin{enumerate}[label=(\roman*)]
	\item\label{item:trapped}
		Since the walk is short (this is crucial), the walk will be ``trapped''
		in $U$, as there are only few edges leaving $U$. Thus, for the
		discrepancy vector it should hold that
		$\sum_{v\in U} \abs{\ppp_u^t(v)} \gg \sum_{v\in V\setminus U} \abs{\ppp_u^t(v)}$.
	\item\label{item:discrepancy-signs}
	   	If $u\in V_1$ then most walks ending at vertices $v\in V_1$
		should have a positive sign and most walks ending at vertices $V\in V_2$
		should have negative sign. Thus, for the discrepancy vector $\ppp_u^t$ it
		should hold that $\ppp_u^t(v) > 0$ if $v\in V_1$ and 
		$\ppp_u^t(v) < 0$ if $v\in V_2$. Similarly, if $u\in V_2$ then the same
		holds with flipped signs.
	\item\label{item:largeinnerbeta} Let $u\in U$. If two vertices $x,y$ are
		from the same sub-communities (i.e., $x,y\in V_1$ or $x,y\in V_2$), then
		$\ppp_u^t(x)\approx \ppp_u^t(y)$, i.e., the discrepancy on $x$ is close
		to the discrepancy on $y$.
\end{enumerate}

Now let us discuss how we can use the discrepancy vectors $\ppp_u^t$ and $\ppp_v^t$ to
decide whether $u$ and $v$ are in the same cluster or not. Our goal is to find a
distance function $\Delta_{uv} = \Delta_{uv}(\ppp_u^t,\ppp_v^t)$ such that
$\Delta_{uv}$ is small iff $u$ and $v$ are from the same cluster $U$.

First, consider vertices $u\in U_i$ and $v\in U_j$, $i\neq j$, from
different clusters.  Then Property~\ref{item:trapped} suggests that $\ppp_u^t$ and
$\ppp_v^t$ have most of their mass in different entries and thus
$\norm{\ppp_u^t - \ppp_v^t}_2^2$ should be large. 

Second, consider vertices $u,v\in U$ from the same cluster. Let $U=V_1 \cup V_2$
as above. If $u,v\in V_1$ or $u,v\in V_2$ then the properties %
above suggest 
$\ppp_u^t \approx \ppp_v^t$ and thus $\norm{\ppp_u^t - \ppp_v^t}_2^2 \approx 0$ is small.
However, if $u\in V_1$ and $v\in V_2$, then $\ppp_u^t \approx -\ppp_v^t$
by Property~\ref{item:discrepancy-signs} and \ref{item:largeinnerbeta} and thus $\norm{\ppp_u^t - \ppp_v^t}_2^2
\approx \norm{2 \ppp_u^t}_2^2$ will still be large.  However, this issue can be
mitigated if we use $\mathbf{r}_u^t = \abs{\ppp_u^t}$ and $\mathbf{r}_v^t = \abs{\ppp_v^t}$ instead of
$\ppp_u^t$ and $\ppp_v^t$. Here, the absolute values are applied component-wise, i.e.,
$\mathbf{r}_u^t(v) = \abs{\ppp_u^t(v)}$ for all $v$.

Therefore, in our analysis we would like to use $\norm{\mathbf{r}_u^t - \mathbf{r}_v^t}^2_2$ as a
distance measure. However, due to some technical difficulties in the analysis of
this quantity, we cannot use the vectors $\mathbf{r}_u^t$ and instead we consider
$\min\{\norm{\ppp_u^t - \ppp_v^t}^2_2,\norm{\ppp_u^t+ \ppp_v^t}^2_2\}$ which behaves
similar to taking the absolute values and also fixes the issue with
Property~\ref{item:discrepancy-signs}. After adding some degree corrections, we
arrive at our final distance measure
$\Delta_{uv}=\min\{\norm{\ppp_u^t \DD^{-1/2}- \ppp_v^t \DD^{-1/2}}^2_2,\norm{\ppp_u^t \DD^{-1/2}+ \ppp_v^t \DD^{-1/2}}^2_2\}$.
Note that $\Delta$ is a pseudometric distance, i.e., one may have that
$\Delta_{uv}=0$ for distinct vectors $\ppp_u^t$ and $\ppp_v^t$.

\subsection{Main Technical Lemmas}
\label{sec:main-technical-lemmas}
We give a technical overview of our analysis. Note that while our high-level
proof strategy is relatively similar to the one used in~\cite{czumaj15testing}
for unsigned graphs, the concrete proofs are often quite different. It required
a substantial amount of work and new ideas to obtain our results for signed
graphs.

\textbf{Random walks from the same cluster.} First, we show that if $U$ is a
polarized community with large inner signed
bipartiteness ratio and small (outer) bipartiteness ratio, then for most of the
vertices $u,v\in U$ their distance $\Delta_{uv}$ is small.
\begin{restatable}{lemma}{vectorclose}
\label{lem:close}
	Let $\alpha,\gamma\in(0,1)$.
	Let $G=(V,E,\sigma)$ be a signed, $d$-bounded degree,
	$(k,\betain,\betaout)$-clusterable graph.
	Let $U$ be a subset of $V$ such that
	$\abs{U}\geq \gamma n$, $\betainner(G[U])\geq\betain$
	and $\beta_G(U)\leq \betaout$.
	Then for all $t\geq \frac{4 C_1^2 (k+1)^6 \lg n}{\betain^2}$,
	$\betaout<\alpha_1\betain^2$ 
	where $\alpha_1=\frac{\alpha \gamma}{800000 C_1^2 k d^3}$, there exists a
	subset $\widetilde{U}\subseteq U$ with
	$\abs{\widetilde{U}} \geq (1-\alpha)\abs{U}$ and for all
	$u,v \in \widetilde{U}$, it holds that
		$\Delta_{uv} \leq \frac{1}{4nd}$.
\end{restatable}

Proving Lem.~\ref{lem:close} was one of the main obstacles for obtaining our
results. We will give an overview of its proof and our main technical
contribution in Sec.~\ref{sec:main-technical-contribution}.

\textbf{Random walks from two different clusters.} %
We further show that if we have two disjoint
communities $U_1$ and $U_2$, then for most vertices $u\in U_1$ and $v\in U_2$
their distance $\Delta_{uv}$ must be large.
Note that while Lem.~\ref{lem:close} assumes a lower bound on the walk
length~$t$, Lem.~\ref{lem:dissimilar-distributions} assumes an upper bound
on~$t$. Thus, it is crucial that the random walks have the correct
length~$\Theta(\lg n)$. %
\begin{restatable}{lemma}{vectorfaraway}
\label{lem:dissimilar-distributions}
	Let $U_1$ and $U_2$ be two disjoint subsets with
	$\beta_G(U_1),\beta_G(U_2)\leq \betaout$. Let $0<\alpha<1$. For any
	$0\leq t\leq \frac{\alpha}{8\betaout}$, there exist subsets
	$\widehat{U_1}\subseteq U_1,\widehat{U_2}\subseteq U_2$ such that
	$|\widehat{U_1}|\geq (1-\alpha)|U_1|$,
	$|\widehat{U_2}|\geq (1-\alpha)|U_2|$, and for any $u\in \widehat{U_1}$ and
	$v\in \widehat{U_2}$, it holds that
		$\Delta_{uv} \geq \frac{1}{nd}$.
\end{restatable}

\textbf{$\ell_2^2$-norm of the vector $\ppp_v^t \DD^{-1/2}$.} Based on the previous two lemmas, our goal will be to test whether $\Delta_{uv}
\leq \frac{1}{4nd}$ or $\Delta_{uv}\geq \frac{1}{nd}$. To this end, we wish to
use \textsc{EstDotProd} to approximate $\Delta_{uv}$ via
Eqn.s~\eqref{eq:Delta_uv_dot_product} and~\eqref{eq:delta_uv}.
To prove this, we first give a useful bound on the $\ell_2^2$-norm of the
vectors $\ppp_v^t \DD^{-1/2}$.
\begin{restatable}{lemma}{vectorsmallnorm}
\label{lem:small-norm}
	Let $\alpha \in (0,1)$. Suppose $G=(V,E,\sigma)$ is a signed
	and $(k,\betain,\betaout)$-clusterable graph. Then there exists a set
	$V'\subseteq V$ of size $\abs{V'}\geq(1-\alpha)\abs{V}$ such that for all
	$u\in V'$ and all $t\geq \frac{C_1^2 (k+1)^6 \lg n}{\betain^2}$, we have that
		$\norm{\ppp_v^t \DD^{-1/2}}^2_2 \leq \frac{2k}{\alpha n}$.
\end{restatable}

\textbf{Estimating the dot product.} Now we prove that
	\textsc{EstDotProd}($u$,$v$,$t$,$\alpha$) estimates $\langle
\ppp_u^t \DD^{-1/2}, \ppp_v^t \DD^{-1/2}\rangle$ with small error.
\begin{restatable}{lemma}{dotproduct}
\label{lem:estimate-dot-product}
	Let $\alpha\in (0,1)$ be a number such that $\frac{2k}{\alpha}\leq n$.
	Suppose $G=(V,E,\sigma)$ is a signed and
	$(k,\betain,\betaout)$-clusterable graph.
	Let $t \geq \frac{C_1^2 (k+1)^6 \lg n}{\betain^2}$.
	Let $V'\subseteq V$ be the set of vertices satisfying the property given by
	Lem. \ref{lem:small-norm}.
	Then %
	\textsc{EstDotProd}($u,v,t,\alpha$) %
	outputs $X_{uv}$ such that with probability $1-1/n^3$,
	it holds that
	\[
		\abs{ X_{uv}  - \langle \ppp_v^t \DD^{-1/2}, \ppp_u^t \DD^{-1/2}\rangle} \leq \frac{1}{20nd}
	\]
	for all $u,v\in V'$.
	Furthermore, \textsc{EstDotProd}($u,v,t,\alpha$) %
	runs in time
	$O(\frac{d^2k^{1.5}t\log n}{\alpha^{1.5}}\cdot\sqrt{n})$.
\end{restatable}

To prove Thm.~\ref{thm:oracle}, we first show that an
overwhelming fraction of the vertices are ``well-behaved'' in the senses of
Lem.s~\ref{lem:close}--\ref{lem:small-norm}. Then, if we only consider these
``well-behaved'' vertices, we can apply Lem.~\ref{lem:estimate-dot-product} and
this will classify all of these vertices correctly with high probability.

\section{Deferred Proofs}\label{appsec:analysis}
\label{sec:proofs}

\subsection{Proof of Lemma~\ref{lem:cheeger-eigenvalues}}\label{app:proofofeigenvalues}
Now we prove Lem. \ref{lem:cheeger-eigenvalues}, which is restated in the
following for the sake of readability.
\eigenvalues*

\begin{proof}[Proof of Lemma~\ref{lem:cheeger-eigenvalues}]
	Since $G$ is $(k,\betain,\betaout)$-clusterable, there exists
	a partition of $V$ into clusters
	$(C_1,\dots,C_k)$
	such that $\betainner(G[C_i]) \geq \betain$ and
	$\beta_G(C_i)\leq\betaout$ for all $i\in[k]$.
	Therefore, $\beta_k(G) \leq \max_{i} \beta_G(C_i) \leq \betaout$.
	Now Thm.~\ref{thm:signed-cheeger} implies that 
	$\lambda_k \leq 2 \beta_k(G) \leq 2\betaout$.
	Thus, $\lambda_1\leq\cdots\leq\lambda_k\leq 2\betaout$.

	The next part shows that $\lambda_{k+1}\geq \frac{\betain^2}{C_1^2 (k+1)^6}$.
	Consider any $k+1$ disjoint subsets $V_1,\dots,V_{k+1}$.

	Now observe that there must exist a subset $V_{i_0}\in\{V_1,\dots,V_{k+1}\}$
	with following property: for all $i\in[k]$,
	$\vol(V_{i_0}\cap C_i) \leq \frac{1}{2} \vol(C_i)$.
	To see that this is the case, suppose that the statement is false, i.e., no such subset exists.
	Then by pigeonhole principle there must exist indices $j_1$, $j_2$ with $j_1\neq j_2$ and index $j_3$
	such that  $V_{j_1}\cap C_{j_3}$ and
	$V_{j_2}\cap C_{j_3}$ have volume more than
	$\frac{1}{2} \vol(C_{j_3})$.
	This is a contradiction since the sets $V_{j_1}$ and
	$V_{j_2}$ are mutually disjoint.

	For the rest of the proof, consider the subset $V_{i_0}$ with the above property. %
Consider an arbitrary partition $L,R$ of $V_{i_0}$, such that $L\cup R=V_{i_0}$ and $L\cap R=\emptyset$. Let $L_i=C_i\cap L$ and $R_i=C_i\cap R$. Observe that 	$V_{i_0} = \bigcup_i (L_i\cup R_i)$ since $V = \bigcup_i 	C_i$. Since $\betainner(G[C_i])\geq \betain$ and 
	$\vol(L_i\cup R_i) \leq \frac{1}{2} \vol(C_i)$  for all $i\leq k$,
	we get that
	$\beta_{G[C_i]}(L_i \cup R_i) \geq \betain$ 
	for all $i\in[k]$.
	In particular, this implies that
	$\beta_{G[C_i]}(L_i, R_i) = \frac{e_{G[C_i]}(L_i,R_i)}{\vol(L_i\cup R_i)} \geq \betain$ 
	for all $i\in[k]$.
	Thus,
	\begin{align*}
		\beta_G(L,R)
		&= \frac{e_G(L,R)}{\vol(L \cup R)} \\
		&= \frac{2 \abs{E^+_G(L,R)} + \abs{E^-_G(L)} + \abs{E^-_G(R)}
				+ \abs{E_G(L\cup R,V\setminus(L\cup R))}}{\vol(L \cup R)} \\
		&\geq 
			\frac{1}{\vol(L \cup R)}
			\sum_{i\in[k]} 
			\Big[
				2 \abs{E^+_G(L_{i},R_{i})}
				+ \abs{E^-_G(L_{i})}
				+ \abs{E^-_G(R_{i})}
			+ \abs{E_G(L_{i}\cup R_{i},C_i\setminus(L_{i}\cup R_{i}))} \Big] \\
		&= \frac{\sum_{i\in[k]} e_{G[C_{i}]}(L_{i},R_{i})}
			{\sum_{i\in[k]} \vol(L_{i}\cup R_{i})} \\
		&\geq \frac{\sum_{i\in[k]} \betain \vol(L_{i} \cup R_{i})}
			{\sum_{i\in[k]} \vol(L_{i}\cup R_{i})} \\
	   	&\geq \betain.
	\end{align*}
Since $L,R$ is an arbitrary partition of $V_{i_0}$, it holds that $\beta_G(V_{i_0})\geq \betain$. 
	Thus, $\beta_{k+1}(G)\geq \betain$. Now Thm.~\ref{thm:signed-cheeger}
	implies that $\lambda_{k+1} \geq \frac{\betain^2}{C_1^2 (k+1)^6}$.
\end{proof}

\subsection{Proof of Lemma~\ref{lem:small-diff}}\label{app:proofsmalldiff}
Now we prove Lem. \ref{lem:small-diff}, which is restated in the
following for the sake of readability.
\smalldiff*

\begin{proof}
	The proof is based on the following intuition. Recall the definitions of vectors $\trvv_i,\1_{V_1,V_2},\ww$ from the above discussion. Since both  $\trvv_i$ and a scalar
	multiplications of $\1_{V_1,V_2}$ have small total `discrepancy'  over the
	set of all edges (i.e., Ineq. (\ref{eq:lambda-i-betaout}) and (\ref{eqn:about_indicator_vector})), and the ratio between the total `discrepancy' over all the edges
	and the total `discrepancy' over all vertices (w.r.t. some centers defined by $\ww$) is large (i.e., Ineq. (\ref{eqn:lambda_2_H})), one
	can guarantee both that $\vv_i'$ and a scaled multiplication of $\1_{V_1,V_2}$
	are close to (another scalar multiplication of) $\ww$. We now give the details.  

	Since $G$ is $(k,\betain,\betaout)$-clusterable, we can apply
	Lem.~\ref{lem:cheeger-eigenvalues} to obtain
	$\lambda_{k+1} \geq \frac{\betain^2}{C_1^2 (k+1)^6}$ and
	$\lambda_i \leq 2\betaout$ for all $i\in[k]$.
	
	Since for any $i\leq k$, $\vv_i \LL^\sigma=\lambda_i \vv_i$, we have 
	$\vv_i \DD^{-1/2} (\DD-\AAA^\sigma) \DD^{-1/2} =\lambda_i \vv_i$, and thus 
\begin{align*}
\lambda_i &=\vv_i \DD^{-1/2} (\DD-\AAA^\sigma) \DD^{-1/2}\vv_i^\top \\
			&=\trvv_i (\DD-\AAA^\sigma) {\trvv_i}^\top \\
			&=\sum_{(u,v)\in E} \left(\trvv_i(u) -\sigma(uv)\trvv_i(v)\right)^2.
\end{align*}

Thus, for any $i\leq k$, 
\begin{align}
\label{eq:lambda-i-betaout}
\sum_{(u,v)\in E} \left(\trvv_i(u) -\sigma(uv)\trvv_i(v)\right)^2\leq 2\betaout.
\end{align}

	Now consider the subgraph $H=G[U]$ induced by the
	subset $U$. Denote the set of vertices of $H$ as
	$V_H=U$ and the set of edges of $H$ as $E_H$.
	By assumption on $U$, 
	$\betainner(H)\geq\betain$ and $\beta_G(V[H])\leq \betaout$.

	Let $\lambda_1(H),\lambda_2(H)$ be the first and second eigenvalues of the
	normalized Laplacian matrix $\LL_H^\sigma$ of $H$. Next, let $\DD_H$ be
	the degree matrix of $H$ and let $\AAA_H^\sigma$ denote the signed adjacency
	matrix of $H$.
	
Now we consider a partition $(V_1,V_2)$ of $U$ such that
$\beta_G(U)=\beta_G(V_1,V_2)\leq \betaout$.
Then it holds that $\beta(H)\leq \beta_H(V_1,V_2) 
 \leq  \beta_G(V_1,V_2) \leq  \betaout$. Therefore, by Thm.
\ref{thm:signed-cheeger}, 
\[\lambda_1(H)\leq 2\betaout.\]

Let $\1_{V_1,V_2}\in \mathbb{R}^{V_H}$ such that $\1_{V_1,V_2}(u)=1$ if $u\in
V_1$ and $-1$ otherwise. Set $\trii :=\frac{\1_{V_1,V_2}}{\sqrt{\mu_U} }$.
Then it holds that
\begin{align}
\sum_{(u,v)\in E}(\trii(u)-\sigma(uv)\trii(v))^2
&= \frac{\1_{V_1,V_2}}{\sqrt{\mu_U}}(\DD_H-\AAA_H^\sigma)\frac{\1_{V_1,V_2}\top}{\sqrt{\mu_U}}\nonumber\\
&=\frac{1}{\mu_U}\sum_{(u,v)\in E(H)}(\1_{V_1,V_2}(u)-\sigma(u,v)\1_{V_1,V_2}(v))^2\nonumber \\
&=\frac{2|E^+(V_1,V_2)|+|E^-(V_1)|+|E^-(V_2)|}{\mu_U} \\
&= \beta_H(V_1,V_2) \\
&\leq \betaout.
  \label{eqn:about_indicator_vector}
\end{align}
	
Let $\ww\in \mathbb{R}^{V_H}$ be an eigenvector corresponding to $\lambda_1(H)$ such that $\ww \ww^\top =1$, then it holds that 
	$\ww \LL_H^\sigma = \lambda_1(H) \ww$.

Now since $\betainner(H)=\betainner(G[U])\geq \betain$, we know
(see App.~\ref{sec:motivation-inner})
that
$\beta_2(H) = \beta_2(G[U])\geq \betainner(G[U])\geq \betain$. Thus, by
Thm.~\ref{thm:signed-cheeger}, 
\[
\lambda_2(H)\geq \frac{\beta_2(H)^2}{64C_1^2}\geq \frac{\betain^2}{64C_1^2}.
\]

Now by the variational characterization of the eigenvalues (see, e.g.,
Eqn.~(1.7) in \cite{chung97spectral}), we have 
\begin{align}
\label{eqn:lambda_2_H}
\lambda_2(H) = \min_{\f\in \mathbb{R}^{V_H} } \max_{c\in \mathbb{R}}
				\frac{\sum_{(u,v)\in E_H}  (\f(u) - \sigma(uv)\f(v))^2}
					{\sum_u (\f(u)-c\cdot \ww(u))^2 d_H(u)}.
\end{align}

Now for any $i\leq k$, if we let $\f= {\trvv_i}_{| H}$, i.e., the function
$\trvv_i$ that is restricted on $H$, in Inequality~\eqref{eqn:lambda_2_H}, and
let $$c_{1,i}:=\arg \min_{c\in \mathbb{R}: c\neq 0}{\sum_u (\f(u)-c\cdot \ww(u))^2 d_H(u)},$$
then we have that
\begin{align}
\label{eqn:lambda_2_v_i_H}
    \frac{\sum_{(u,v)\in E_H}  (\trvv_i(u) - \sigma(uv)\trvv_i(v))^2}
					{\sum_u (\trvv_i(u)-c_{1,i}\cdot \ww(u))^2 d_H(u)}
	\geq \lambda_2(H)
	\geq \frac{\betain^2}{64C_1^2}.
\end{align}

If we let $\f= {\trii}$ in Inequality~\eqref{eqn:lambda_2_H}, and let
$$c_2:=\arg \min_{c\in \mathbb{R}: c\neq 0}{\sum_u (\trii(u)-c\cdot \ww(u))^2 d_H(u)}.$$
Thus,
\begin{align}
\label{eqn:lambda_2_indicator_H}
    \frac{\sum_{(u,v)\in E_H}  (\trii(u) - \sigma(uv)\trii(v))^2}
					{\sum_u (\trii(u)-c_2\cdot \ww(u))^2 d_H(u)}
	\geq \lambda_2(H)
	\geq \frac{\betain^2}{64C_1^2}.
\end{align}

Therefore, by Inequalities~\eqref{eq:lambda-i-betaout},
\eqref{eqn:about_indicator_vector}, \eqref{eqn:lambda_2_v_i_H} and
\eqref{eqn:lambda_2_indicator_H}, it holds that 
\begin{align}
\label{eqn:difference}
\begin{split}
	&\sum_u (\trvv_i(u)-c_{1,i}\cdot \ww(u))^2 d_H(u) \leq \frac{128C_1^2\betaout}{\betain^2},
	\text{ and } \\
	&\sum_u (\trii(u)-c_2\cdot \ww(u))^2 d_H(u) <\frac{128C_1^2\betaout}{\betain^2}.
\end{split}
\end{align}
Equivalently, 
\begin{align*}
 &\norm{\trvv_i \DD^{1/2} - c_{1,i} \ww \DD^{1/2}}_2^2\leq \frac{128C_1^2\betaout}{\betain^2}<\frac14,
	\text{ and } \\
 &\norm{c_2 \ww \DD^{1/2}-\trii \DD^{1/2}}_2^2\leq  \frac{128C_1^2\betaout}{\betain^2} <\frac14,
\end{align*}
where we make use the fact that $512C_1^2\betaout< \betain^2$.

Recall that $\trvv_i= \vv_i\DD^{-1/2}$ and $\trii =\frac{\1_{V_1,V_2}}{\sqrt{\mu_U} }$. Therefore,
\begin{align*}
	\norm{\trvv_i \DD^{1/2}}_2^2
	&=\sum_{u\in V_H} {\trvv_i}^2(u) d_H(u) \\
	&=\sum_{u\in V_H} {\vv_i}^2(u) d_H(u)^{-1} \cdot d_H(u) \\
	&\leq \sum_{u\in V} \vv_i^2(u) \\
	&= 1
\end{align*}
and 
\[
	\norm{\trii \DD^{1/2}}_2^2=\sum_{u\in V_H}\trii^2(u)d_H(u)=1.
\]
By the above inequalities, we have 
\[|c_{1,i}|\cdot \norm{\ww \DD^{1/2}}_2\leq \frac{1}{2} + \norm{\trvv_i \DD^{1/2}}_2=\frac32,\]
and
\[\frac12=\norm{\trii \DD^{1/2}}_2-\frac{1}{2}\leq |c_2|\cdot \norm{\ww \DD^{1/2}}_2\leq \frac{1}{2} + \norm{\trii \DD^{1/2}}_2=\frac32.
\]
By the fact that $d\geq d_u\geq 1$ for any vertex $u\in H$ and that $\norm{\ww \DD^{1/2}}_2^2=\sum_u \ww^2(u)d_u$, we have 
\[d=\sum_u d\ww^2(u)\geq \norm{\ww \DD^{1/2}}_2^2\geq \sum_u \ww^2(u)=1.
\]
Furthermore, we have that 
\[
	|c_{1,i}|\leq \frac{3}{2},\quad
\]
and
\[
	\frac{1}{2d}\leq |c_2|\leq \frac{3}{2}.
\]

Let $B_1=\{u \colon \abs{\trvv_i(u)-c_{1,i}\cdot \ww(u)}^2 \geq \frac{256C_1^2\betaout}{\alpha\betain^2\mu_U}\}$, and
$B_2=\{u \colon \abs{\trii(u)-c_2\cdot \ww(u)}^2 \geq \frac{256C_1^2\betaout}{\alpha\betain^2\mu_U}\}$.
By Inequality~\eqref{eqn:difference}, the fact that $d_H(u)\geq 1$ for any $u\in
V_H$ and an averaging argument, we have $|B_1|\leq \frac{\alpha |U|}{2}$, and $|B_2|\leq
\frac{\alpha |U|}{2}$. Therefore, by letting $\widetilde{U}=U\setminus (B_1\cup
		B_2)$, we have $|\widetilde{U}|\geq (1-\alpha)|U|$, and
it holds that
\begin{align*}
\abs{\trvv_i(u)-c_{1,i}\cdot \ww(u)}^2 \leq \frac{256C_1^2\betaout}{\alpha\betain^2\mu_U},
\end{align*}
and
\begin{align*}
\abs{\trii(u)-c_2\cdot \ww(u)}^2 \leq \frac{256C_1^2\betaout}{\alpha\betain^2\mu_U}.
\end{align*}
for any $u\in \widetilde{U}$. Thus, 
\begin{align*}
	\abs{\trvv_i(u)-c_{1,i} \cdot \ww(u)} \leq \frac{16C_1}{\betain}\cdot\sqrt{\frac{\betaout}{\alpha \cdot \mu_U}},
\end{align*}
and
\begin{align*}
	\abs{\frac{c_{1,i}}{c_2}\trii(u)-c_{1,i} \cdot \ww(u)} 
		&\leq \frac{|c_{1,i}|}{|c_2|}\frac{16C_1}{\betain}\cdot\sqrt{\frac{\betaout}{\alpha \mu_U}} \\ 
		&\leq \frac{3/2}{1/(2d)} \frac{16C_1}{\betain}\cdot\sqrt{\frac{\betaout}{\alpha \mu_U}} \\
		&\leq \frac{48dC_1}{\betain}\cdot\sqrt{\frac{\betaout}{\alpha \mu_U}}.
\end{align*}
Now for each $i\leq k$ we let $c_{i}:=\frac{c_{1,i}}{c_2}$. Then $|c_i|\leq
\frac{3/2}{1/(2d)}=3d$, and 
for any vertex $u\in \widetilde{U}$, it holds that 
\begin{align*}
	\abs{\trvv_i(u)-c_i \cdot \trii(u)}
	&\leq \abs{\trvv_i(u)-c_{1,i} \cdot \ww(u)} + \abs{\frac{c_{1,i}}{c_2}\trii(u)-c_{1,i} \cdot \ww(u)} \\
	&\leq \frac{64d C_1}{\betain}\cdot\sqrt{\frac{\betaout}{\alpha \mu_U}}.
\end{align*}
Finally, by the definition of $\trii=\frac{\1_{V_1,V_2}}{\sqrt{\mu_U}}$, we know that for each $i\leq k$, 
\begin{itemize}
\item if $u\in V_1\cap \widetilde{U}$, then 
$\abs{\trvv_i(u)-c_i \cdot \frac{1}{\sqrt{\mu_U}}} \leq \frac{64d C_1}{\betain}\cdot\sqrt{\frac{\betaout}{\alpha \mu_U}}$,
\item if $u\in V_2\cap \widetilde{U}$, then 
$\abs{\trvv_i(u)+ c_i \cdot \frac{1}{\sqrt{\mu_U}}} \leq \frac{64d
	C_1}{\betain}\cdot\sqrt{\frac{\betaout}{\alpha \mu_U}}$.
	\qedhere
\end{itemize}
\end{proof}

\subsection{Proof of Lemma~\ref{lem:close}}
\label{sec:proof-lem-close}
Now we prove Lem.~\ref{lem:close}, which is restated in the
following for the sake of readability.
\vectorclose*
\begin{proof}
Since
\begin{align*}
\ppp_u^t=\1_u \WW^t & = \1_u\left(\frac{\II+\DD^{-1}\AAA^\sigma}{2}\right)^t\\ &= \1_u \left(\frac{\DD^{-1/2}(\II+ \DD^{-1/2}\AAA^\sigma \DD^{-1/2})\DD^{1/2}}{2}\right)^t \\
&= \1_u \DD^{-1/2}\left(\frac{\II+\DD^{-1/2}\AAA^\sigma \DD^{-1/2}}{2}\right)^t \DD^{1/2}
\\
&= \1_u \DD^{-1/2}\left(\II-\frac{\LL_G}{2}\right)^t \DD^{1/2} \\
&= \1_u \DD^{-1/2} \sum_{i}(1-\lambda_i/2)^t \vv_i^\top \vv_i \DD^{1/2}.
\end{align*}
Recall that $\trvv_i=\vv_i \DD^{-1/2}$. We have that, 
\begin{align*}
\ppp_u^t \DD^{-1/2}
&= \1_u\DD^{-1/2} \sum_{i}(1-\lambda_i/2)^t \vv_i^\top \vv_i \\
&= \sum_{i}(1-\lambda_i/2)^t \frac{\vv_i(u)}{\sqrt{d_G(u)}} \vv_i \\
&=\sum_{i}(1-\lambda_i/2)^t \trvv_i(u) \vv_i.
\end{align*}
Therefore we get that
\begin{align*}
&\norm{\ppp_u^t \DD^{-1/2}- \ppp_v^t \DD^{-1/2}}^2_2\\
		&=\norm{\sum_{i=1}^n  (\trvv_i(u)-\trvv_i(v)) (1-\lambda_i/2)^t \vv_i}_2^2\\
		&=\sum_{i=1}^n  (\trvv_i(u)-\trvv_i(v))^{2} (1-\lambda_i/2)^{2t}\\
		&\leq \sum_{i=1}^k (\trvv_i(u)-\trvv_i(v))^{2}
			+ \sum_{i=k+1}^n (\trvv_i(u)-\trvv_i(v))^{2} \left(1 - \lambda_i/2\right)^{2t} \\
		&\leq \sum_{i=1}^k (\trvv_i(u)-\trvv_i(v))^{2}
			+ \left(1 - \lambda_{k+1}/2\right)^{2t}\sum_{i=k+1}^n 2\left(\frac{\vv_i(u)^2}{d_G(u)}+\frac{\vv_i(v)^2}{d_G(v)}\right) \\
&\leq \sum_{i=1}^k (\trvv_i(u)-\trvv_i(v))^{2}
			+ 4 \left(1 - \frac{\betain^2}{2C_1^2 (k+1)^6} \right)^{2t} \\
	&\leq \sum_{i=1}^k (\trvv_i(u)-\trvv_i(v))^{2}
			+ 4 \exp\left( - \frac{ \betain^2 t}{C_1^2 (k+1)^6} \right),
\end{align*}
	where in the fourth step we used that
	$(\trvv_i(u)-\trvv_i(v))^2
		\leq 2 (\trvv_i(u)^2+\trvv_i(v)^2)$
	by the Cauchy--Schwarz inequality
	and the definition of $\trvv_i$.
	In the fifth step we used that 
	$\sum_{i=k+1}^n \vv_i(u)^2 \leq
	\sum_{i=1}^n \vv_i(u)^2=1$
	and that
	$\lambda_{k+1} \geq \frac{\betain^2}{C_1^2 (k+1)^6}$ by
	Lem.~\ref{lem:cheeger-eigenvalues}, as well as $d_G(u)\geq 1$, for any $u\in V$.
Similarly, 
\begin{align*}
\norm{\ppp_u^t \DD^{-1/2}+ \ppp_v^t \DD^{-1/2}}^2_2
\leq \sum_{i=1}^k (\trvv_i(u)+\trvv_i(v))^{2}
			+  4 \exp\left( - \frac{ \betain^2 t}{C_1^2 (k+1)^6} \right).
\end{align*}
We apply Lem. \ref{lem:small-diff} on the set $U$ and let
$\widetilde{U}\subseteq U$,  $V_1,V_2$ be the partition of $U$ with the property
specified in Lem. \ref{lem:small-diff}. Now we consider two vertices
$u,v\in \widetilde{U}$. We now distinguish four cases:
\begin{itemize}
\item Case~1: If $u,v\in V_1$, then for all $i\in [k]$,
	\begin{align*}
		\abs{\trvv_i(u)-\trvv_i(v)}
			&\leq \abs{\trvv_i(u)-c_i \cdot \frac{1}{\sqrt{\mu_U}}} + \abs{\trvv_i(v)-c_i \cdot \frac{1}{\sqrt{\mu_U}}} \\ 
			&\leq \frac{128d C_1}{\betain}\cdot\sqrt{\frac{\betaout}{\alpha \mu_U}}.
	\end{align*}
	Thus, 
	\begin{align*}
		\norm{\ppp_u^t \DD^{-1/2}- \ppp_v^t \DD^{-1/2}}^2_2
		&\leq \sum_{i=1}^k (\trvv_i(u)-\trvv_i(v))^{2}
					+  4 \exp\left( - \frac{ \betain^2 t}{C_1^2 (k+1)^6} \right) \\
		&\leq \frac{20000k d^2 C_1^2\betaout}{\betain^2\cdot \alpha \mu_U}	+  4 \exp\left( - \frac{ \betain^2 t}{C_1^2 (k+1)^6} \right)\\
		&\leq \frac{20000k d^2 C_1^2\betaout}{\betain^2\cdot \alpha |U|}	+  4 \exp\left( - \frac{ \betain^2 t}{C_1^2 (k+1)^6} \right)\\
		&\leq \frac{20000k d^2 C_1^2\betaout}{\betain^2\cdot \alpha \gamma n}	+  4 \exp\left( - \frac{ \betain^2 t}{C_1^2 (k+1)^6} \right)\\
		& \leq\frac{1}{4 nd},
	\end{align*}
	where in the second to last inequality, we use the assumption that $\betaout < \alpha_1 \betain^2$ such that 
		$\alpha_1 = \alpha_1(k,\alpha,\gamma,d) = \frac{\alpha \gamma}{800000 C_1^2 k d^3}$ and that
		$t \geq \frac{4 C_1^2 (k+1)^6 \lg n}{\betain^2}$.
\item Case~2: If $u,v\in V_2$, then for all $i\in [k]$, 
	\begin{align*}
		\abs{\trvv_i(u)-\trvv_i(v)}
			&\leq \abs{\trvv_i(u)+ c_i \cdot \frac{1}{\sqrt{\mu_U}}} + \abs{\trvv_i(v)+ c_i \cdot \frac{1}{\sqrt{\mu_U}}} \\
			&\leq \frac{128d C_1}{\betain}\cdot\sqrt{\frac{\betaout}{\alpha \mu_U}}.
	\end{align*}
	Thus, similarly as above, 
	\begin{align*}
		\norm{\ppp_u^t \DD^{-1/2}- \ppp_v^t \DD^{-1/2}}^2_2
		&\leq \frac{20000k d^2 C_1^2\betaout}{\betain^2\cdot \alpha \mu_U}	+  4 \exp\left( - \frac{ \betain^2 t}{C_1^2 (k+1)^6} \right) \\
		&\leq \frac{1}{4nd}.
	\end{align*}
\item Case~3: If $u\in V_1,v\in V_2$, then for all $i\in [k]$, 
	\begin{align*}
		\abs{\trvv_i(u)+\trvv_i(v)}
			&\leq \abs{\trvv_i(u)- c_i \cdot \frac{1}{\sqrt{\mu_U}}} + \abs{\trvv_i(v)+ c_i \cdot \frac{1}{\sqrt{\mu_U}}} \\
			&\leq \frac{128d C_1}{\betain}\cdot\sqrt{\frac{\betaout}{\alpha \mu_U}}.
	\end{align*}
	Thus, similarly as above, 
	\begin{align*}
		\norm{\ppp_u^t \DD^{-1/2}+ \ppp_v^t \DD^{-1/2}}^2_2
		&\leq \sum_{i=1}^k (\trvv_i(u)+\trvv_i(v))^{2}
				+  4 \exp\left( - \frac{ \betain^2 t}{C_1^2 (k+1)^6} \right) \\
		&\leq \frac{20000k d^2 C_1^2\betaout}{\betain^2\cdot \alpha \mu_U}	+  4 \exp\left( - \frac{ \betain^2 t}{C_1^2 (k+1)^6} \right)\\
		&\leq \frac{1}{4nd}.
	\end{align*}

\item Case~4: If $u\in V_2,v\in V_1$ then for all $i\in [k]$, 
	\begin{align*}
		\abs{\trvv_i(u)+\trvv_i(v)}
		&\leq \abs{\trvv_i(u)+ c_i \cdot \frac{1}{\sqrt{\mu_U}}} + \abs{\trvv_i(v)-c_i \cdot \frac{1}{\sqrt{\mu_U}}} \\
		&\leq \frac{128d C_1}{\betain}\cdot\sqrt{\frac{\betaout}{\alpha \mu_U}}.
	\end{align*}

	Thus, similarly as above
	\begin{align*}
		\norm{\ppp_u^t \DD^{-1/2}+ \ppp_v^t \DD^{-1/2}}^2_2
		&\leq \frac{20000k d^2 C_1^2\betaout}{\betain^2\cdot \alpha \mu_U}	+  4 \exp\left( - \frac{ \betain^2 t}{C_1^2 (k+1)^6} \right) \\
		&\leq \frac{1}{4nd}.
	\end{align*}
	\end{itemize}

Therefore, for any two vertices $u,v\in \widetilde{U}$, we have that 
\begin{align*}
&\Delta_{uv}=\min\{\norm{\ppp_u^t \DD^{-1/2}- \ppp_v^t \DD^{-1/2}}^2_2,\norm{\ppp_u^t \DD^{-1/2}+ \ppp_v^t \DD^{-1/2}}^2_2\}
\leq \frac{1}{4nd}.
\end{align*}

\end{proof}

\subsection{Proof of Lemma~\ref{lem:dissimilar-distributions}}

Now we prove Lem. \ref{lem:dissimilar-distributions}, which is restated in the
following for the sake of readability.
\vectorfaraway*
\begin{proof}
Let $0<\alpha<1$. Consider a subset $C=(V_1,V_2)$ with $\beta_G(V_1,V_2)\leq \betaout$. We first show that for any $t\geq 0$, there exists a subset $\widehat{C}\subseteq C$ such that $\vol(\widehat{C})\geq (1-\alpha)\vol(C)$ and for any $v \in \widehat{C}$, $\sum_{w\in C}\abs{\1_v \WW^t(w)}\ge 1 - \frac{2 t \betaout}{\alpha}$. 
To do so, we first introduce some notations. For any vertex subset $C=(V_1,V_2)\subseteq V$, %
we define vectors $y_{V_1,V_2}$ and $1_{V_1,V_2}$ as 
\begin{align*}
\y_{V_1,V_2}(u) =
\begin{cases}
\frac{d_u}{\vol(C)}       & \quad \text{if } u\in V_1,\\
-\frac{d_u}{\vol(C)}  & \quad \text{if } u\in V_2,\\
0 & \quad \text{otherwise},
\end{cases}
\\
\1_{V_1,V_2}(u) =
\begin{cases}
1       & \quad \text{if } u\in V_1,\\
-1 & \quad \text{if } u\in V_2,\\
0 & \quad \text{otherwise}.
\end{cases}
\end{align*}

We first show the following result.
\begin{claim}
For all $t\geq 0$, $\y_{V_1,V_2} \WW ^t \1_{V_1,V_2}^\top \geq 1-t\betaout$.
\end{claim}
\begin{proof}
We prove for any $t\geq 0$, 
\[
\y_{V_1,V_2} \WW^t \1_{V_1,V_2}^\top - \y_{V_1,V_2} \WW^{t+1} \1_{V_1,V_2}^\top \leq \betaout. 
\]
Note that once the above inequality is proven, the claim follows from the fact that $\y_{V_1,V_2} \WW^0 \1_{V_1,V_2}^\top = \y_{V_1,V_2} \1_{V_1,V_2}^\top=1.$

Let $\x$ be the vector such that $\x^\top=\WW^t \1_{V_1,V_2}^\top$. Note that for any
vertex $w\in V$, it holds that $|\x(w)|\leq 1$.
Therefore,  
\begin{align*}
& \y_{V_1,V_2} \WW ^t \1_{V_1,V_2}^\top - \y_{V_1,V_2} \WW ^{t+1} \1_{V_1,V_2}^\top\\
&= \y_{V_1,V_2} (\II-\WW)W^t \1_{V_1,V_2}^\top \\
&=  \y_{V_1,V_2} \DD^{-1}\frac{\DD-\AAA^{\sigma}}{2}\WW^t \1_{V_1,V_2}^\top\\
&=\frac{1}{2\vol(C)}\1_{V_1,V_2}(\DD-\AAA^{\sigma})\WW^t \1_{V_1,V_2}\\
&=\frac{1}{2\vol(C)}\sum_{(u,v)\in E}\left(\1_{V_1,V_2}(u)-\sigma(u,v)\1_{V_1,V_2}(v)\right)\cdot \left(\x(u)-\sigma(u,v)\x(v)\right) \\
&\leq\frac{1}{2\vol(C)}\sum_{(u,v)\in E}\abs{\1_{V_1,V_2}(u)-\sigma(u,v)\1_{V_1,V_2}(v)}\cdot \abs{\x(u)-\sigma(u,v)\x(v)} \\
&\leq\frac{1}{2\vol(C)}\sum_{(u,v)\in E} 2 \cdot \abs{\1_{V_1,V_2}(u)-\sigma(u,v)\1_{V_1,V_2}(v)} \\
&= \frac{1}{\vol(C)}\cdot (2|E_G^+(V_1,V_2)|+2|E_G^-(V_1)|+2|E_G^-(V_2)|+ |E_G(C,\overline{C})|)\\
&\leq 2\beta_G(V_1,V_2)\leq 2\betaout. \qedhere
\end{align*}

\end{proof}

By the above claim, we have
\begin{align*}
&\sum_{v\in C}\frac{d_v}{\vol(C)}\sum_{w\in C}\abs{\1_v \WW^t (w)} \\
\geq &  \sum_{v\in V_1}\frac{d_v}{\vol(C)}{\1_v \WW^t \1_{V_1,V_2}^\top}-\sum_{v\in V_2}\frac{d_v}{\vol(C)}{\1_v \WW^t \1_{V_1,V_2}^\top} \\
= &
\y_{V_1,V_2}\WW^t \1_{V_1,V_2}^\top 
	   \ge  1-2t\betaout.
\end{align*}
Thus, 
\begin{align*}
	\sum_{v\in C}\frac{d_v}{\vol(C)}(1-\sum_{w\in C}|\1_v\WW^t(w)|)
	&=1-\sum_{v\in C}\frac{d_v}{\vol(C)}\sum_{w\in C}\abs{\1_v \WW^t (w)} \\
	&\leq 2t\betaout.
\end{align*}

Let $Q_C = \{v : \sum_{w\in C}\abs{\1_v \WW^t(w)}\le 1 - \frac{2d t \betaout}{\alpha}\}$. Then,
\begin{align*}
\sum_{v\in C}\frac{d_v}{\vol(C)}(1-\sum_{w\in C}\abs{\1_v \WW^t (w)})
&\ge  \sum_{v \in Q_C}\frac{d_v}{\vol(C)}(1 - \sum_{w\in C}\abs{\1_v \WW^t (w)})\\
&	   \ge
	\frac{\vol(Q_C)}{\vol(C)}\frac{2dt \betaout}{ \alpha}\\
& \geq \frac{|Q_C|}{d |C|} \frac{2dt \betaout}{\alpha}
	\enspace.  
\end{align*}

Thus, $|Q_C|\le \alpha |C|$. Therefore, if we set $\widehat{C} = C\setminus Q_C$, then $|\widehat{C}| \ge (1-\alpha) |C|$, and for any $v \in \widehat{C}$, 
\[\sum_{w\in C}\abs{\1_v \WW^t(w)}\ge 1 - \frac{2 d t \betaout}{\alpha}.
\]

Now for any two disjoint sets $C_1=(V_1,V_2)$ and $C_2=(V_1',V_2')$, we define $\widehat{C_1}$ and $\widehat{C_2}$ for $C_1$ and $C_2$, respectively. Thus, $|\widehat{C_1}|\geq (1-\alpha)|C_1|$ and $|\widehat{C_2}|\geq (1-\alpha)|C_2|$. Furthermore, for any $t \ge 1$ and $0 < \alpha < 1$, for any $u \in \widehat{C_1}$ and $v \in \widehat{C_2}$:%
\begin{align*}
	\sum_{w\in C_1}|\1_u \WW^t(w)| &\ge   1-\frac{2dt\betaout}{\alpha},
\end{align*}
and
\begin{align*}
	\sum_{w\in C_2}|\1_v \WW^t(w)| &\ge   1-\frac{2dt\betaout}{\alpha}.
\end{align*}
	
Since $C_1$ and $C_2$ are disjoint, we have \[
\sum_{w\in C_1}|\1_v \WW^t(w)| \le  1 - \sum_{w\in C_2}|\1_v \WW^t(w)|\le \frac{2d t \betaout}{\alpha},\]
and
\[
\sum_{w\in C_2}|\1_u \WW^t(w)| \le  1 - \sum_{w\in C_1}|\1_u \WW^t(w)|\le \frac{2 dt \betaout}{\alpha}.\]
Let $\qq_u^t$ be the vector such that $\qq_u^t(w)=\abs{\1_v \WW^t(w)}$.
Therefore, for any $t \ge 0$,
\begin{align*}
	&\norm{(\qq_u^t-\qq_v^t)\DD^{-1/2}}_2^2 \\
	&= \sum_{w\in V}(\qq_u^t(w)-\qq_v^t(w))^2\frac{1}{d_G(w)}\\
	&=	\left(\sum_{w\in V}(\qq_u^t(w)-\qq_v^t(w))^2\frac{1}{d_G(w)}\right)
			\cdot \left(\sum_{w\in V}d_G(w))\cdot \frac{1}{\vol(G)}\right) \\
	&\geq \frac{(\sum_{w\in V}\abs{\qq_u^t(w)-\qq_v^t(w)})^2}{\vol(G)}\qquad \textrm{(by Cauchy Schwarz inequality)}\\
	&\geq \frac{(\sum_{w\in V}\abs{\abs{\1_u\WW^t(w)}-\abs{\1_v\WW^t(w)}})^2}{\vol(G)} \\
	&\ge \frac{1}{\vol(G)}
		\left[
			\sum_{w\in C_1}(|\1_u\WW^t(w)|-|\1_v\WW^t(w)|)
				+ \sum_{w\in C_2}(|\1_v\WW^t(w)|-|\1_u\WW^t(w)|) \right]^2\\
	&=\frac{1}{\vol(G)}
		\left[ \sum_{w\in C_1}|\1_u\WW^t(w)|-\sum_{w\in C_1}|\1_v\WW^t(w)|
			+ \sum_{w\in C_2}|\1_v\WW^t(w)|-\sum_{w\in C_2}|\1_u\WW^t(w)| \right]^2
	\\ & \ge \frac{(2\cdot (1 - \frac{2d t \betaout}{ \alpha} -\frac{2d t \betaout}{ \alpha}) )^2}{\vol(G)} \\
	&= \frac{(2\cdot(1 - \frac{4dt \betaout}{ \alpha}))^2}{\vol(G)}.
\end{align*}
	
In particular, if $t \le \frac{\alpha}{8d\betaout}$, then  $\norm{(\qq_u^t - \qq_v^t)\DD^{-1/2}}^2_2 \ge \frac{1}{\vol(G)}\geq \frac{1}{nd}$.

The lemma then follows from the fact that
\begin{align*}
	\Delta_{uv}
	&=\min\{\norm{\ppp_u^t \DD^{-1/2}- \ppp_v^t \DD^{-1/2}}^2_2,\norm{\ppp_u^t \DD^{-1/2}+ \ppp_v^t \DD^{-1/2}}^2_2\} \\
	&\geq \norm{\qq_u^t \DD^{-1/2} - \qq_v^t \DD^{-1/2}}_2^2.
\end{align*}
\end{proof}

\subsection{Proof of Lemma~\ref{lem:small-norm}}

Now we prove Lem. \ref{lem:small-norm}, which is restated in the following for the sake of readability.
\vectorsmallnorm*

\begin{proof}
Recall that $\vv_i$ is the $i$-th eigenvector of $\LL_G^\sigma$, and $\trvv_i = \vv_i \DD^{-1/2}$. For all $u\in V$, we set $\delta(u)=\sum_{i=1}^k \trvv_i(u)^2$. Since we have
	that $\norm{\vv_i}_2^2=1$ and $d_G(u)\geq 1$,
	\begin{align*}
		\sum_{u\in V} \delta(u)
		= \sum_{u\in V} \sum_{i=1}^k \frac{\vv_i(u)^2}{d_G(u)}
		= \sum_{i=1}^k \sum_{u\in V} \frac{\vv_i(u)^2}{d_G(u)}
		\leq k.
	\end{align*}
	Thus, the average value of $\delta(u)$ over all $u$ is $\frac{k}{n}$. This
	implies that there exists a subset of vertices $V'\subseteq V$ of size
	$\abs{V'}\geq(1-\alpha)\abs{V}$ such that $\delta(u) \leq \frac{k}{\alpha n}$
	for all $u\in V'$.

	Furthermore, we have that $\1_u = \sum_{i=1}^n \vv_i(u)\vv_i$ and
	$\ppp_u^t \DD^{-1/2} = \sum_{i=1}^n \trvv_i(u) \left(1-\frac{\lambda_i}{2}\right)^t \vv_i$.
	We now get that
	\begin{align*}
		\norm{\ppp_u^t\DD^{-1/2}}_2^2
		&= \norm{ \sum_{i=1}^n \trvv_i(u) \left(1-\frac{\lambda_i}{2}\right)^t \vv_i }_2^2 \\
		&= \sum_{i=1}^n \trvv_i(u)^2 \left(1-\frac{\lambda_i}{2}\right)^{2t} \\
		&= \sum_{i=1}^k \trvv_i(u)^2 \left(1-\frac{\lambda_i}{2}\right)^{2t}
			+ \sum_{i=k+1}^n \trvv_i(u)^2 \left(1-\frac{\lambda_i}{2}\right)^{2t} \\
		&= \sum_{i=1}^k \trvv_i(u)^2
			+ \left(1-\frac{\lambda_{k+1}}{2}\right)^{2t} \sum_{i=k+1}^n \trvv_i(u)^2 \\
		&\leq \delta(u) + \left(1-\frac{\lambda_{k+1}}{2}\right)^{2t} \\
		&\leq \frac{k}{\alpha n} + \left(1-\frac{\betain^2}{2C_1^2 (k+1)^6}\right)^{2t} \\
		&\leq \frac{k}{\alpha n} + \exp\left( - \frac{ \betain^2 t}{C_1^2 (k+1)^6} \right) \\
		&\leq \frac{2k}{\alpha n},
	\end{align*}
	where we used that $\lambda_{k+1}\geq \frac{\betain^2}{C_1^2 (k+1)^6}$ by
	Lem.~\ref{lem:cheeger-eigenvalues} and in the last step we used that
	$t \geq \frac{C_1^2 (k+1)^6 \lg n}{\betain^2}$. %
\end{proof}

\subsection{Proof of Lemma~\ref{lem:estimate-dot-product}}
Now we prove Lem.~\ref{lem:estimate-dot-product}, which is restated in the following for the sake of readability.
\dotproduct*

\begin{proof}
This proof is based on~\citet[Lem.~19]{chiplunkar18testing}. However, since the
vectors we are analyzing may contain negative entries, we need to give a more refined analysis on the variance of the corresponding estimator.  

Let $u,v\in V'$. Recall that $\ppp_u^t = \1_u \WW^t$ and $\ppp_v^t = \1_v \WW^t$. By Lem.~\ref{lem:small-norm} we have that $\norm{\ppp_u^t\DD^{-1/2}}_2, \norm{\ppp_v^t\DD^{-1/2}}_2  \leq
\sqrt{\frac{2k}{\alpha n}}$.
Let $\eta\in (0,1)$ be a parameter that will be specified later. Let $R$ be an integer such that $R^2 \geq \frac{6}{\eta^2} \cdot \frac{2k}{\alpha n}$
	and $R \geq \frac{24}{\eta^2} \cdot \left( \frac{2k}{\alpha n}\right)^{1.5}$. 
	
For $x\in\{u,v\}$, we perform $R$ lazy signed random walks from $x$ of length $t$. 
Let $X_{x,w}^{r,s}$ be a random variable that is $\frac{1}{\sqrt{d_G(w)}}$ if
the $r$'th walk that starts at vertex $x$ ends at vertex $w$ with sign
$s\in\{+,-\}$.
Set $X_{x,w}^r = X_{x,w}^{r,+} - X_{x,w}^{r,-}$.
Observe that $\Exp{X_{x,w}^r} = \frac{\ppp_x^t(w)}{\sqrt{d_G(w)}}$ for
all $w\in V$. 

Let $\mm^s_x(w)$ be the fraction of walks that start at $x$ and end at $w$ with sign~$s$, for $x\in\{u,v\}$. Let $\mm_x=\mm^+_x \DD^{-1/2} - \mm^-_x \DD^{-1/2}$. 

Now for any pair of vertices $u,v\in V'$, observe that
\begin{align*}
	\langle \mm_u, \mm_v \rangle
	= \frac{1}{R^2}
		\sum_{w \in V}
			\left( \sum_{r_u=1}^{R} X_{u,w}^{r_u} \right)
			\left( \sum_{r_v=1}^{R} X_{v,w}^{r_v} \right).
\end{align*}
This implies that
\begin{align*}
	\Exp{ \langle \mm_u, \mm_v \rangle }
	&= \Exp{ \frac{1}{R^2} \sum_{w \in V} \left( \sum_{r_u=1}^{R} X_{u,w}^{r_u} \right) \left(							 \sum_{r_v=1}^{R} X_{v,w}^{r_v} \right)} \\
	&= \frac{1}{R^2} \sum_{w \in V} R \frac{\ppp_u^t(w)}{\sqrt{d_G(w)}} \cdot R \frac{\ppp_v^t(w)}{\sqrt{d_G(w)}} \\
	&= \sum_{w \in V} \frac{\ppp_u^t(w)}{\sqrt{d_G(w)}}  \frac{\ppp_v^t(w)}{\sqrt{d_G(w)}} 
	= \langle \ppp_u^t \DD^{-1/2}, \ppp_v^t \DD^{-1/2} \rangle.
\end{align*}

Next, we wish to compute
$\Var{\langle \mm_u, \mm_v \rangle} = \Exp{\langle \mm_u, \mm_v \rangle^2}- \Exp{\langle \mm_u, \mm_v \rangle}^2$.
We start by computing $\Exp{\langle \mm_u, \mm_v \rangle^2}$:
\begin{align*}
	\Exp{\langle \mm_u, \mm_v \rangle^2}
	&= \Exp{ \frac{1}{R^4}
			\sum_{w \in V}
			\sum_{w' \in V}
			\sum_{r_u=1}^{R}
			\sum_{r_u'=1}^{R} 
			\sum_{r_v=1}^{R}
			\sum_{r_v'=1}^{R}
			X_{u,w}^{r_u} X_{u,w'}^{r_u'} X_{v,w}^{r_v} X_{v,w'}^{r_v'} } \\
	&= \frac{1}{R^4}
			\sum_{w \in V}
			\sum_{w' \in V}
			\sum_{r_u=1}^{R}
			\sum_{r_u'=1}^{R} 
			\sum_{r_v=1}^{R}
			\sum_{r_v'=1}^{R}
			\Exp{ X_{u,w}^{r_u} X_{u,w'}^{r_u'} X_{v,w}^{r_v} X_{v,w'}^{r_v'} }.
\end{align*}
We perform a case distinction in order to bound
$\Exp{ X_{u,w}^{r_u} X_{u,w'}^{r_u'} X_{v,w}^{r_v} X_{v,w'}^{r_v'} }$:
\begin{itemize}
	\item If $w \neq w'$, then 
	\begin{align*}
	\Exp{ X_{u,w}^{r_u} X_{u,w'}^{r_u'} X_{v,w}^{r_v} X_{v,w'}^{r_v'} }
	=
	\begin{cases}
	\frac{\ppp_u^t(w)}{\sqrt{d_G(w)}} \cdot \frac{\ppp_v^t(w)}{\sqrt{d_G(w)}} \cdot \frac{\ppp_{ u}^t(w')}{{\sqrt{d_G(w')}}} \cdot \frac{\ppp_{v}^t(w')}{\sqrt{d_G(w')}}      & \quad \text{if $r_u\neq r_{u}'$ and $r_v\neq
			 r_{v'}$},\\
	0 & \quad \text{otherwise}.
	\end{cases}
\end{align*}
	\item If $w = w'$, $r_u=r_{u}'$ and $r_v = r_{v}'$ then 
\[\Exp{ X_{u,w}^{r_u} X_{u,w}^{r_u} X_{v,w}^{r_v} X_{v,w}^{r_v} }
	\leq \frac{\abs{\ppp_u^t(w)}}{\sqrt{d_G(w)}} \cdot \frac{\abs{\ppp_v^t(w)}}{\sqrt{d_G(w)}}\cdot \frac{1}{\sqrt{d_G(w)}}\cdot \frac{1}{\sqrt{d_G(w)}}.
	\]
	\item If $w = w'$, $r_u=r_{u}'$ and $r_v \neq r_{v}'$ then 
\[
\Exp{ X_{u,w}^{r_u} X_{u,w}^{r_u} X_{v,w}^{r_v} X_{v,w}^{r_v'} }
\leq \frac{\abs{\ppp_u^t(w)}}{\sqrt{d_G(w)}} \cdot \frac{\abs{\ppp_v^t(w)}}{\sqrt{d_G(w)}}\cdot \frac{1}{\sqrt{d_G(w)}}\cdot \frac{\abs{\ppp_v^t(w)}}{\sqrt{d_G(w)}}.
\]
	\item If $w = w'$, $r_u\neq r_{u}'$ and $r_v = r_{v}'$ then 
\[
\Exp{ X_{u,w}^{r_u} X_{u,w}^{r_u'} X_{v,w}^{r_v} X_{v,w}^{r_v} }
\leq \frac{\abs{\ppp_u^t(w)}}{\sqrt{d_G(w)}} \cdot \frac{\abs{\ppp_v^t(w)}}{\sqrt{d_G(w)}}\cdot \frac{\abs{\ppp_u^t(w)}}{\sqrt{d_G(w)}}\cdot \frac{1}{\sqrt{d_G(w)}}.
\]
	\item If $w = w'$, $r_u\neq r_{u}'$ and $r_v\neq r_{v}'$ then 
\[
\Exp{ X_{u,w}^{r_u} X_{u,w}^{r_u'} X_{v,w}^{r_v} X_{v,w}^{r_v'} }
\leq \frac{\abs{\ppp_u^t(w)}}{\sqrt{d_G(w)}} \cdot \frac{\abs{\ppp_v^t(w)}}{\sqrt{d_G(w)}}\cdot \frac{\abs{\ppp_u^t(w)}}{\sqrt{d_G(w)}}\cdot \frac{\abs{\ppp_v^t(w)}}{\sqrt{d_G(w)}}.
\]
\end{itemize}
Thus, we obtain that
\begin{align*}
	&\quad\Exp{\langle \mm_u, \mm_v \rangle^2}\\
	&= \frac{1}{R^4}
			\sum_{w \in V}
			\sum_{w' \in V}
			\sum_{r_u=1}^{R}
			\sum_{r_u'=1}^{R} 
			\sum_{r_v=1}^{R}
			\sum_{r_v'=1}^{R}
			\Exp{ X_{u,w}^{r_u} X_{u,w'}^{r_u'} X_{v,w}^{r_v} X_{v,w'}^{r_v'} } \\
	&\leq \frac{R^2(R-1)^2}{R^4}\sum_{w\in V} \sum_{w' \neq w} \frac{{\ppp_u^t(w)}}{\sqrt{d_G(w)}} \cdot \frac{{\ppp_v^t(w)}}{\sqrt{d_G(w)}} \cdot \frac{{\ppp_{u}^t(w')}}{\sqrt{d_G(w')}}\cdot \frac{{\ppp_{v}^t(w')}}{\sqrt{d_G(w')}} \\
		&\quad+ \frac{1}{R^2} \sum_{w\in V} \frac{\abs{\ppp_u^t(w)}\cdot \abs{\ppp_v^t(w)}}{{d_G^2(w)}} 
			+ \frac{1}{R} \sum_{w\in V} \frac{\abs{\ppp_u^t(w)} \cdot \abs{\ppp_v^t(w)}^2}{d_G^2(w)} \\
		&\quad+ \frac{1}{R} \sum_{w\in V} \frac{\abs{\ppp_u^t(w)}^2 \cdot \abs{\ppp_v^t(w)}}{d_G^2(w)}
		 + \sum_{w\in V} \frac{\abs{\ppp_u^t(w)}^2 \cdot \abs{\ppp_v^t(w)}^2}{d_G^2(w)} \\
	&= \sum_{w,w'\in V} \frac{{\ppp_u^t(w)}\cdot {\ppp_v^t(w)}\cdot {\ppp_{u}^t(w')} \cdot {\ppp_{v}^t(w')}}{{d_G(w)} d_G(w')} \\
		&\quad- \left(\frac{2R-1}{R^2}\right)\sum_{w\in V}\sum_{w'\neq w} \frac{{\ppp_u^t(w)}\cdot {\ppp_v^t(w)}\cdot {\ppp_{u}^t(w')} \cdot {\ppp_{v}^t(w')}}{{d_G(w)} d_G(w')} \\
		&\quad+ \frac{1}{R^2} \sum_{w\in V} \frac{\abs{\ppp_u^t(w)}\cdot \abs{\ppp_v^t(w)}}{{d_G^2(w)}}  + \frac{1}{R} \sum_{w\in V} \frac{\abs{\ppp_u^t(w)} \cdot \abs{\ppp_v^t(w)}^2}{d_G^2(w)}\\
		 &\quad+ \frac{1}{R} \sum_{w\in V} \frac{\abs{\ppp_u^t(w)}^2 \cdot \abs{\ppp_v^t(w)}}{d_G^2(w)},
\end{align*}

This implies that
\begin{align*}
	&\quad\Var{\langle \mm_u, \mm_v \rangle}
		= \Exp{\langle \mm_u, \mm_v \rangle^2}- \Exp{\langle \mm_u, \mm_v \rangle}^2 \\
	&\leq \sum_{w,w'\in V} \frac{{\ppp_u^t(w)}\cdot {\ppp_v^t(w)}\cdot {\ppp_{u}^t(w')} \cdot {\ppp_{v}^t(w')}}{{d_G(w)} d_G(w')} \\
		&\quad- \left(\frac{2R-1}{R^2}\right)\sum_{w\in V}\sum_{w'\neq w} \frac{{\ppp_u^t(w)}\cdot {\ppp_v^t(w)}\cdot {\ppp_{u}^t(w')} \cdot {\ppp_{v}^t(w')}}{{d_G(w)} d_G(w')} \\
		&\quad+ \frac{1}{R^2} \sum_{w\in V} \frac{\abs{\ppp_u^t(w)}\cdot \abs{\ppp_v^t(w)}}{{d_G^2(w)}} 
			+ \frac{1}{R} \sum_{w\in V} \frac{\abs{\ppp_u^t(w)} \cdot \abs{\ppp_v^t(w)}^2}{d_G^2(w)} \\
		&\quad+ \frac{1}{R} \sum_{w\in V} \frac{\abs{\ppp_u^t(w)}^2 \cdot \abs{\ppp_v^t(w)}}{d_G^2(w)}
			- \left( \sum_{w \in V} \frac{\ppp_u^t(w)}{\sqrt{d_G(w)}} \frac{\ppp_v^t(w)}{\sqrt{d_G(w)}} \right)^2 \\
	&\leq \frac{1}{R^2} \sum_{w\in V} \frac{\abs{\ppp_u^t(w)}}{\sqrt{d_G(w)}} \cdot \frac{\abs{\ppp_v^t(w)}}{\sqrt{d_G(w)}}
		+ \frac{1}{R} \sum_{w\in V} \frac{\abs{\ppp_u^t(w)}}{\sqrt{d_G(w)}} \cdot \left(\frac{\abs{\ppp_v^t(w)}}{\sqrt{d_G(w)}}\right)^2 \\
		&\quad+ \frac{1}{R} \sum_{w\in V} \left(\frac{\abs{\ppp_u^t(w)}}{\sqrt{d_G(w)}}\right)^2 \cdot \frac{\abs{\ppp_v^t(w)}}{\sqrt{d_G(w)}} \\
	&\quad + \frac{2}{R} \sum_{w,w'\in V}\frac{\abs{\ppp_u^t(w)}\cdot \abs{\ppp_v^t(w)}\cdot \abs{\ppp_{u}^t(w')} \cdot \abs{\ppp_{v}^t(w')}}{{d_G(w)} d_G(w')}\\
	&\leq \frac{1}{R^2} \norm{\ppp_u^t \DD^{-1/2}}_2 \cdot \norm{\ppp_v^t \DD^{-1/2}}_2
			+ \frac{1}{R} \norm{\ppp_u^t \DD^{-1/2}}_2 \cdot \norm{\ppp_v^t \DD^{-1/2}}^2_4 \\
		&\quad+ \frac{1}{R} \norm{\ppp_u^t \DD^{-1/2}}_4^2 \cdot \norm{\ppp_v^t \DD^{-1/2}}_2
			+\frac{2}{R} \left(\sum_{w\in V}\frac{\abs{\ppp_u^t(w)}\cdot \abs{\ppp_v^t(w)}}{{d_G(w)}}\right)^2\\
	&\leq \frac{1}{R^2} \norm{\ppp_u^t \DD^{-1/2}}_2 \cdot \norm{\ppp_v^t \DD^{-1/2}}_2
		 + \frac{1}{R} \norm{\ppp_u^t \DD^{-1/2}}_2 \cdot \norm{\ppp_v^t \DD^{-1/2}}^2_2 \\
	 &\quad+ \frac{1}{R} \norm{\ppp_u^t \DD^{-1/2}}_2^2 \cdot \norm{\ppp_v^t \DD^{-1/2}}_2
		 +\frac{2}{R} \norm{\ppp_u^t \DD^{-1/2}}_2^2 \cdot \norm{\ppp_v^t \DD^{-1/2}}_2^2,
\end{align*}
where in the last step we have used the Cauchy-Schwarz inequality and the monotonicity of the
$\ell_p$-norms\footnote{It is known that for $p,q\in(0,\infty)$ with $p \leq q$,
	it holds that $\norm{\x}_q \leq \norm{\x}_p$ for all vectors $\x$.
}, which gives $\norm{\x}_4 \leq \norm{\x}_2$, for any vector $\x$.

Now recall that $\norm{\ppp_u^t\DD^{-1/2}}_2, \norm{\ppp_v^t\DD^{-1/2}}_2  \leq
\sqrt{\frac{2k}{\alpha n}}$. Recall that $\frac{2k}{\alpha}\leq n$ and thus $\frac{2k}{\alpha n} \leq 1$.
This implies that
\begin{align*}
	\Var{\langle \mm_u, \mm_v \rangle}
	&\leq \frac{1}{R^2} \cdot \frac{2k}{\alpha n}
		+ \frac{2}{R} \cdot \left( \frac{2k}{\alpha n} \right)^{1.5}
		+ \frac{2}{R} \cdot \left( \frac{2k}{\alpha n} \right)^{2} \\
	&\leq \frac{1}{R^2} \cdot \frac{2k}{\alpha n}
		+ \frac{4}{R} \cdot \left( \frac{2k}{\alpha n} \right)^{1.5}.
\end{align*}
Now Chebyshev's Inequality implies that:
\begin{align*}
&\Prob{\abs{ \langle \mm_u, \mm_v \rangle - \langle \ppp_u^t\DD^{-1/2}, \ppp_v^t\DD^{-1/2} \rangle} \geq \eta } \\
	&=  \Prob{ \abs{ \langle \mm_u, \mm_v \rangle - \Exp{\langle \mm_u, \mm_v \rangle}} \geq \eta }\\
	&\leq \frac{\Var{\langle \mm_u, \mm_v \rangle}}{\eta^2} \\
	&\leq \frac{1}{\eta^2} \cdot
		\left( \frac{1}{R^2} \cdot \frac{2k}{\alpha n}
		+ \frac{4}{R} \cdot \left( \frac{2k}{\alpha n} \right)^{1.5}\right) \\
	&\leq \frac{1}{3}.
\end{align*}
In the last inequality we have used that
$R^2 \geq \frac{6}{\eta^2} \cdot \frac{2k}{\alpha n}$
	and $R \geq \frac{24}{\eta^2} \cdot \left( \frac{2k}{\alpha n}\right)^{1.5}$.
	
Now we let $\eta= \frac{1}{20 nd}$ and $R=\frac{40000d^2 k^{1.5}\sqrt{n}}{\alpha^{1.5}}$ so that the above conditions on $R$ are satisfied. Thus, with probability at least $1-\frac13=\frac23$, the estimate $\langle \mm_u, \mm_v\rangle$ satisfies that \[
\abs{ \langle \mm_u, \mm_v\rangle - {\langle \ppp_u^t\DD^{-1/2}, \ppp_v^t\DD^{-1/2} \rangle}} \leq \frac{1}{20nd}. 
\]
Now note that the algorithm \textsc{EstDotProd}($u,v,t,\alpha$) repeatedly invokes the above subroutine for $h=O(\log n)$ times and outputs the median of the corresponding estimates $\langle \mm_u, \mm_v\rangle$, we are guaranteed that with probability at least $1-1/n^3$, the output $X_{uv}$ satisfies that
\[
\abs{ X_{uv} - {\langle \ppp_u^t\DD^{-1/2}, \ppp_v^t\DD^{-1/2} \rangle}} \leq \frac{1}{20nd}.
\]

To obtain the runtime result, observe that for each run of the subroutine, the algorithm only performs $R$ random walks of length $t$ from both $u$ and $v$, which can be done in $O(R t)$ time. Thus, each of the vectors $\mm_u$ and $\mm_v$ has at most
$R$ non-zero entries and the dot product $\langle \mm_u,\mm_v\rangle$ can be computed in time
$O(R)$. Finally, since we run the subroutine for $O(\log n)$ times, the total running time is thus $O(R t \log n)=O(\frac{d^2k^{1.5}t\sqrt{n}\log n}{\alpha^{1.5}})$.

This finishes the proof of the lemma.
\end{proof}

\subsection{Proof of Thm.~\ref{thm:oracle}}
Now we prove Thm.~\ref{thm:oracle}, which is restated in the following for the sake of readability.
\oracle*

\begin{proof}
Given the above lemmas, we can prove our main theorem as follows. Recall that
$C''>0$ is some large constant. Note that we have selected the random walk length
$t=\frac{C'' k^6 d^3 \log n}{\betain^2}$. 

Let
$s=\frac{20k}{\gamma}\log({k})$, $\alpha =\frac{\varepsilon}{90s}$.
Recall that
$\betaout< \frac{\varepsilon\betain^2 }{C'\log(k) k^7d^3\log n}$.
Note that $t\leq \frac{\alpha }{8\betaout}$ by changing appropriately large
$C''$. Note that
$\frac{2k}{\alpha}=\frac{1800k^2\log(k)}{\gamma\varepsilon}\leq n$.

\textbf{Correctness.} 
Let $U_1,\dots,U_k$ be a $(k,\betain,\betaout)$-clustering of $G$ such that
each cluster has size at least $\frac{\gamma n}{k}$, for some universal constant
$\gamma >0$. For any $u\in V$, let $U_u$ be the cluster that contains $u$. We
call a vertex $u$ \emph{bad}, if either:
\begin{itemize}
	\item $u\in V\setminus V'$, where $V'$ is the set as defined in Lem.
		\ref{lem:small-norm} with $\alpha =\frac{\varepsilon}{90s}$,
	\item $u\in U_u \setminus \widetilde{U_u}$, where $\widetilde{U_u}$ is as
		defined in Lem. \ref{lem:close} with $\alpha =\frac{\varepsilon}{90s}$,
		or 
	\item $u\in U_u \setminus \widehat{U_u}$, where $\widehat{U_u}$ is as
		defined in Lem. \ref{lem:dissimilar-distributions} with
		$\alpha=\frac{\varepsilon}{90s}$.
\end{itemize}
Let $B$ denote the set of all bad vertices. Note that  
\[|B|\leq \left( \frac{\varepsilon}{90s} +\frac{\varepsilon}{90s}+\frac{\varepsilon}{90s}\right) \cdot n=\frac{\varepsilon\cdot n}{30s}.
\]
We call a vertex $u$ \emph{good}, if it is not bad. 

Note that since each cluster $U$ has size at least $\frac{\gamma n}{k}$, it
holds that
\[
|B| 
\leq \frac{\gamma \varepsilon \cdot n}{600k\log({k}) }
\leq O\left(\frac{\varepsilon}{\log({k})}\right)|U|.
\]

Thus, with probability
at least $1-\frac{\varepsilon}{30s}\cdot s \geq 1-\frac{1}{30}$, all the
vertices in $S$ are good. In the following, we will assume that this is the
case. 

Note further that since each cluster $U$ satisfies that $|U|\geq  \frac{\gamma n}{k}$ for
some $\gamma=\Omega(1)$, it holds that with probability at least $1-
(1-\frac{\gamma}{k})^s\geq 1-\frac{1}{30k}$, there exists at least one
vertex in $S$ that is from cluster $U$. Thus, for all the $k$ clusters $U$, with
probability at least $1-\frac{1}{30}$, there exists at least one vertex
in $S$ that is from cluster $U$.

By Lem. \ref{lem:small-norm}, we know that for any $v\in S$,
$\norm{\ppp_v^t \DD^{-1/2}}_2^2
	\leq \frac{2k}{n}\cdot \frac{90s}{\varepsilon}
	=\frac{3600k^2\log(k)}{\gamma\varepsilon\cdot n}$.
Let $u,v$ be two different vertices in $S$. By
Lem.~\ref{lem:estimate-dot-product}, with probability at least
$1-\frac{1}{n^3}$, we can estimate each term
$\langle \ppp_x^t \DD^{-1/2}, \ppp_y^t \DD^{-1/2} \rangle$ within an additive error at
most $\frac{1}{20nd}$, for any $\{x,y\}\in \{u,v\}$. This also implies that with
probability at least $1-\frac{|S|^2}{n^3}\geq 1-\frac{1}{n^2}$, for all vertex
pairs $x,y\in S$, we have an estimate $X_{xy}$ such that 
\[
	\abs{ X_{xy}  - \langle \ppp_x^t \DD^{-1/2}, \ppp_y^t \DD^{-1/2}\rangle} \leq \frac{1}{20nd}.
\]
In the following, we will assume the above inequality holds for any $x,y\in S$. 

Since $v$ is good for each $v\in S$, we know that
$X_{vv}\leq \frac{3600k^2\log(k)}{\gamma\varepsilon\cdot n}+\frac{1}{20nd}
	\leq \frac{4000k^2\log(k)}{\gamma\varepsilon\cdot n}$. Thus,
Line~\ref{alg:l22normest} of Alg.~\ref{alg:mainoracle} will not happen. 

Note that 
\begin{align*}
	\Delta_{uv}
	=\min\{\norm{\ppp_u^t \DD^{-1/2}- \ppp_v^t \DD^{-1/2}}^2_2,\norm{\ppp_u^t \DD^{-1/2}+ \ppp_v^t \DD^{-1/2}}^2_2\},
\end{align*}
where
\begin{align*}
	\norm{\ppp_u^t \DD^{-1/2}- \ppp_v^t \DD^{-1/2}}^2_2
	= \langle \ppp_u^t \DD^{-1/2}, \ppp_u^t \DD^{-1/2} \rangle - 2 \langle \ppp_u^t \DD^{-1/2}, \ppp_v^t \DD^{-1/2} \rangle + \langle \ppp_v^t \DD^{-1/2}, \ppp_v^t \DD^{-1/2} \rangle,
\end{align*}
and
\begin{align*}
	\norm{\ppp_u^t \DD^{-1/2}+ \ppp_v^t \DD^{-1/2}}^2_2
	= \langle \ppp_u^t \DD^{-1/2}, \ppp_u^t \DD^{-1/2} \rangle + 2 \langle \ppp_u^t \DD^{-1/2}, \ppp_v^t \DD^{-1/2} \rangle + \langle \ppp_v^t \DD^{-1/2}, \ppp_v^t \DD^{-1/2} \rangle.
\end{align*}
Since our estimates $X_{vv},X_{uv}, X_{uu}$ approximate
$\langle \ppp_u^t \DD^{-1/2}, \ppp_u^t \DD^{-1/2} \rangle$,
$\langle \ppp_u^t \DD^{-1/2}, \ppp_v^t \DD^{-1/2} \rangle$
$\langle \ppp_v^t \DD^{-1/2}, \ppp_v^t \DD^{-1/2}\rangle$ within an additive error
$\frac{1}{20nd}$, respectively, we can approximate $\Delta_{uv}$ within an
additive error at most $4\cdot \frac{1}{20 nd}=\frac{1}{5nd}$, i.e., the
estimate $\delta_{uv}$ (at Line~\ref{alg:deltaestimate} of
Alg.~\ref{alg:mainoracle}) satisfies that 
$\abs{\delta_{uv}-\Delta_{uv}} \leq \frac{1}{5nd}$.

Now recall that each cluster $U$ satisfies that $|U|\geq  \frac{\gamma}{k} n$
for some $\gamma =\Omega(1)$ and that 
$\betaout< \frac{\varepsilon\betain^2 }{C'\log(k) k^7d^3\log n}$.
Note that the precondition of Lem. \ref{lem:close} is satisfied. Now let $S_U=S
\cap U$, and let $u,v\in S$. Then:
\begin{itemize}
	\item If $u,v$ belong to the same cluster, by Lem.~\ref{lem:close}, we know
		that $\Delta_{uv}\leq \frac{1}{4nd}$. Then it holds that $\delta_{uv}\leq
		\Delta_{uv}+\frac{1}{5nd} < \frac{1}{2nd}$. Thus, an edge $(u,v)$ will be
		added to $H$ (at line \ref{alg:addedge} of Alg.~\ref{alg:mainoracle}). 
	\item If $u,v$ belong to two different clusters, by
		Lem.~\ref{lem:dissimilar-distributions}, we know  $\Delta_{uv}\geq
		\frac{1}{nd}$.  Then it holds that
		$\delta_{uv}\geq \Delta_{uv}-\frac{1}{5nd} > \frac{1}{2nd}$. Thus, an
		edge $(u,v)$ will not be added to $H$. 
\end{itemize}

Therefore, with probability at least
$1-\frac{1}{30}-\frac{1}{n^2}-\frac{1}{30}\geq 0.9$,
the similarity graph $H$ has the following properties: 
\begin{enumerate}
    \item all vertices in $V(H)$ (i.e., $S$) are good,
	\item all vertices in $S$ that belong to the same cluster $U$ form a clique, denoted by $H_U$,
	\item there is no edge between any two cliques $H_{U_i}$ and $H_{U_j}$ that
		correspond to two different clusters $U_i,U_j$,
	\item there are exactly $k$ cliques in $H$, each corresponding to some cluster.
\end{enumerate}

Now let us consider a membership query, i.e., the subroutine
\textsc{WhichCluster}({$G,v,H,\ell$}) for some vertex $v\in V$. We will show
that any good vertex $v$ will be correctly classified. In the following, we will
assume that $v$ is good. 

Since all the vertices in $S$ are good, we know that for any vertex $u\in
U_v\cap S$, by Lem. \ref{lem:close}, $\Delta_{uv}\leq \frac{1}{4nd}$, and by
the same argument as above, with probability at least $1-1/n^2$, the estimate
$\delta_{uv}$ (at Line~\ref{alg:estimatedeltaagain} of
Alg.~\ref{alg:answeringquery}) satisfies that $\delta_{uv}<\frac{1}{2nd}$.
Thus, the label for $v$ outputted by \textsc{WhichCluster} will be the same as
the $\ell(u)$, the label of $u$.  

On the other hand, for any other vertex $u\in S\setminus U_v$, by
Lem.~\ref{lem:dissimilar-distributions},  $\Delta_{uv}\geq \frac{1}{nd}$, and
by the same argument as above, with probability at least $1-1/n^2$, the estimate
$\delta_{uv}$ (at Line~\ref{alg:estimatedeltaagain} of
Alg.~\ref{alg:answeringquery}) satisfies that $\delta_{uv}>\frac{1}{2nd}$.
This further implies that the label of $v$ will be different from the label of $u$.

Thus, all good vertices are correctly classified with probability at least $1-\frac{2}{30}-\frac{n}{n^2}\geq 0.9$. Assuming this holds, then the set of misclassified
vertices is a subset of all bad vertices, which implies that there exists a
permutation $\pi:[k]\to[k]$ such that 
\[
	|P_{\pi(i)}\triangle U_i| 
	\leq |B|\leq O\left(\frac{\varepsilon}{\log(k)}\right)\cdot |U_i|. 
\]

\textbf{Running time.} We first note that by
Lem.~\ref{lem:estimate-dot-product}, the subroutine
\textsc{EstDotProd}($u,v,t,\alpha$) (i.e., Alg.~\ref{alg:estimate-dot-product})
runs in time $O(\frac{d^2k^{1.5}t\log n}{\alpha^{1.5}}\cdot\sqrt{n})
	=O(\sqrt{n}\poly(\frac{kd\cdot\log n}{\varepsilon\betain}))$. 

For the algorithm \textsc{BuildOracle}, it invokes the subroutine
\textsc{EstDotProd} for $O(s^2)$ times and uses the outputted estimates to
construct the similarity graph $H$, which in total takes
$O(s^2\cdot \frac{d^2k^{1.5}t\log n}{\alpha^{1.5}}\cdot\sqrt{n})
	=O(\sqrt{n}\poly(\frac{kd\cdot\log n}{\varepsilon\betain}))$ time. 

For the algorithm \textsc{WhichCluster}, it invokes the subroutine
\textsc{EstDotProd} for $O(s)$ times and uses uses the outputted estimates to
answer, which in total takes
$O(s\cdot \frac{d^2k^{1.5}t\log n}{\alpha^{1.5}}\cdot\sqrt{n})
	=O(\sqrt{n}\poly(\frac{kd\cdot\log n}{\varepsilon\betain}))$ time. 
\end{proof}

\section{Implementation Details}
\label{sec:implementation}
We describe the practical implementations of our oracle data structures. We also
discuss an \emph{un}signed oracle and a heuristic algorithm for biclustering.
Furthermore, we discuss how to practically determine the parameters for our
algorithms.

\textbf{Practical changes to our signed oracle.}
We start by giving some details on the implementation of our algorithm and the
changes that we have made compared to the theoretical version.

First, we do not set $\tDelta_{uv}$ as described in
Eqn.~\eqref{eq:delta_uv}. Instead, we follow the intuition from
Sec.~\ref{sec:analysis-intuition} and use the vectors $\mathbf{r}_u^t$ rather
than $\ppp_u^t$ in the definition of $\tDelta_{uv}$.
Recall that $\mathbf{r}_u^t = \abs{\ppp_u^t}$, where the absolute values are
taken component-wise. Therefore, we change Line~\ref{line:m_x} in
Alg.~\ref{alg:estimate-dot-product} to $\mm_x\gets\abs{(\mm_x^+ - \mm_x^-) \DD^{-1/2}}$.
Preliminary experiments (not reported here) showed that this provides slightly
better results than when using the original choice of $\tDelta_{uv}$.
Furthermore, we only run the subroutine \textsc{EstDotProd} once (rather than
$h=O(\lg n)$ times).

Next, for a \textsc{WhichCluster}($v$) query, the theoretical
algorithm returns that $v$ belongs to the cluster of vertex~$u\in S$ if
$\tDelta_{uv}\leq\frac{1}{2dn}$. %
However, in practice the the upper bound $\frac{1}{2dn}$ is not a
suitable choice. Thus, we assign $v$ to the cluster of $u=\arg\min_{w\in S}
\tDelta_{wv}$.

\textbf{Seeded and unseeded initialization.}
We consider two different initialization strategies: (1) when ground-truth
seed nodes are available and (2) when use a randomized initialization.

In Case~(1), when a small set of ground-truth seed nodes is available for each
ground-truth cluster, we skip the preprocessing from
Alg.~\ref{alg:mainoracle} and take the vertex labels provided from the
ground-truth seed nodes; we do not perform any other preprocessing.

In Case~(2), we randomly sample a set $S$ of vertices as in
Alg.~\ref{alg:mainoracle}. However, we build the auxiliary graph $H$
differently. Recall that in Alg.~\ref{alg:mainoracle}, we inserted
all edges $(u,v) \in S\times S$ into~$H$ with $\delta_{uv}\leq \frac{1}{2dn}$.
Preliminary experiments indicated that this upper bound is not a good
choice in practice. Instead, we insert edges into~$H$ until it has $k$
connected components (note that, initially, $H$~has $\abs{S}$~connected
components).  More concretely, we compute the pairwise distances $\delta_{uv}$
for all $u,v\in S$.  Then we iterate over these distances in non-decreasing
order and insert the corresponding edges into $H$ until $H$ has exactly $k$
connected components.  To obtain more robust distance estimates for
$\delta_{uv}$, we compute 5~samples of $\delta_{uv}$, $u,v\in S$, and take the
median; we only do this during the preprocessing phase for this
algorithm (and \emph{not} for queries as per
Alg.~\ref{alg:answeringquery}).

\textbf{Heuristic biclustering oracle.}
So far, we considered oracles for finding polarized communities
$U_1,\dots,U_k$ (see Def.~\ref{def:clustering}). However, we did not consider partitioning each
$U_i$ into biclusters $(V_{2i-1},V_{2i})$, that reveal the polarized
groups in $U_i$. We now present a heuristic \emph{biclustering oracle}
for this purpose.

The heuristic biclustering oracle works exactly as the clustering 
oracle, with the following two changes. First, we do not take absolute values
when computing $\mm_x$, i.e., in Line~\ref{line:m_x} in
Alg.~\ref{alg:estimate-dot-product} we set
$\mm_x\gets (\mm_x^+ - \mm_x^-) \DD^{-1/2}$.
Second, when computing $\delta_{uv}$, we now set
$\delta_{uv} \gets X_{uu}+X_{vv}-2X_{uv}$. 
These changes correspond to our intuition from Sec.~\ref{sec:analysis-intuition} that
$\mm_u$ approximates $\ppp_u^t \DD^{-1/2}$ and that
$\delta_{uv}\approx \norm{\ppp_u^t - \ppp_v^t}_2^2$ is small \emph{iff} $u$
and $v$ are from the same bicluster $V_j$. This intuition is also
supported by our analysis via Lem.s~\ref{lem:dissimilar-distributions}
and~\ref{lem:small-diff}. However, this is only a heuristic because
it appears challenging to prove that our estimate $\delta_{uv}$ is
large if $u\in V_1$ and $v\in V_2$; the main challenge is that we are only
allowed a query time of $\tO(\sqrt{n})$.

\textbf{Unsigned oracle.}
To evaluate our oracles, it will be interesting to compare against an
\emph{unsigned oracle}, i.e., an oracle which ignores the edge signs and only
considers the underlying unsigned graph.  To this end, we consider unsigned
versions of our clustering oracle and our heuristic biclustering oracle.
Algorithmically, the only change is that we assume that all edges have sign~$+$.
The resulting unsigned oracle is almost identical to the oracle
in~\cite{czumaj15testing}.

\textbf{Parameter tuning.}
To run our algorithms, one has to determine two crucial parameters: the length
and the number of random walks. For both of them, our analysis requires the
parameters $\alpha$ and $\betain$, which are not available in practice and it
seems infeasible to estimate them.  Therefore, we briefly describe how parameter
tuning can be performed to obtain good choices for the length and the number of
random walks. %

Given a graph~$G$, suppose that for a small set of vertices $V_\text{labeled}$
we know their ground-truth communities. Now we build the oracle for several
different parameters for the length and number of random walks. For each
parameter setting, we run \textsc{WhichCluster}($v$) for all $v\in
V_\text{labeled}$ and check if $v$ was classified correctly. At the end, we pick
the parameter setting with the most correct answers.

Observe that the above procedure does not require a full clustering of $G$ and
can be used even when $V_\text{labeled}$ is small.  Further observe that we
could also split $V_\text{labeled}$ into a training set (used for the seed
nodes), a validation set (used for determining the best parameters) and a test
set (for estimating the overall accuracy).

\section{Experiments on Synthetic Data}
\label{sec:experiments-synthetic}
We evaluate our algorithms on synthetic datasets.  We generated random graphs by
starting with an empty graph and partitioning $n$ vertices into equally-sized
clusters $U_1,\dots,U_k$ with $U_i = (V_{2i-1},V_{2i})$ for
all $i\in[k]$. We inserted edges~$(u,v)$ with the following probabilities: 
$\pinner$ if $u,v \in V_i$, $\pcross$ if $u \in V_{2i-1}$ and $v\in V_{2i}$, $q$
if $u \in U_i, v\not\in U_i$. For all inserted edges, we
set their sign to $+$ ($-$) with probability $\psign$ if $u,v \in V_i$
($u\in V_{2i-1}, v\in V_{2i}$) and to sign $-$ ($+$) with probability
$1-\psign$ otherwise.  If $u\in U_i, v\not\in U_i$ then we set the sign to $+$ with
probability $\qsign$ and to $-$ with probability $1-\qsign$.  When not stated
otherwise, we set $n=2000$, $k=6$, $\pinner=0.8$, $\pcross=0.4$, $q=0.05$,
$\psign=0.8$ and $\qsign=0.9$.
For each experiment we have created $5$~random graphs and we report average
accuracies and their standard deviations.

\emph{Clustering experiments.}
We present our results for finding clusters
$U_1,\dots,U_k$ in Fig.~\ref{fig:experiments}. 
We ran the oracles with $t=2$~random walk steps and $R=400$~random walks, unless
stated otherwise. The seeded oracles and \polarSeeds obtained 6~seed vertices
from each $U_i$; the unseeded oracles randomly sampled $3k$ seeds.

\begin{figure*}[t!]
  \begin{center}
  \subfigure{
	  \includegraphics[width=3\smallfigwidth]{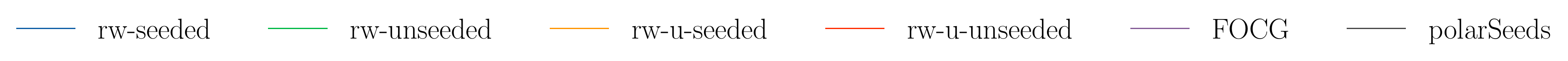}
  }
  \addtocounter{subfigure}{-1}

  \centering
  \subfigure[\small Vary $n$]{
    \label{fig:results_randomGraph_n_accuracy}
    \includegraphics[height=\smallfigheight,width=\smallfigwidth]{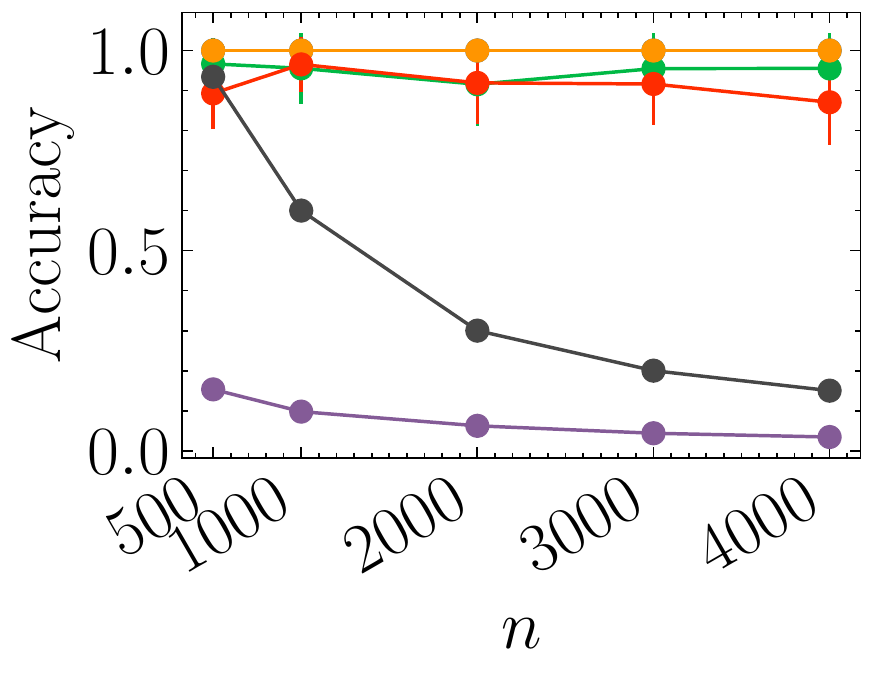}
  }\hspace*{\smallfigsep}
  \subfigure[\small Vary $k$]{
    \label{fig:results_randomGraph_k_accuracy}
    \includegraphics[height=\smallfigheight,width=\smallfigwidth]{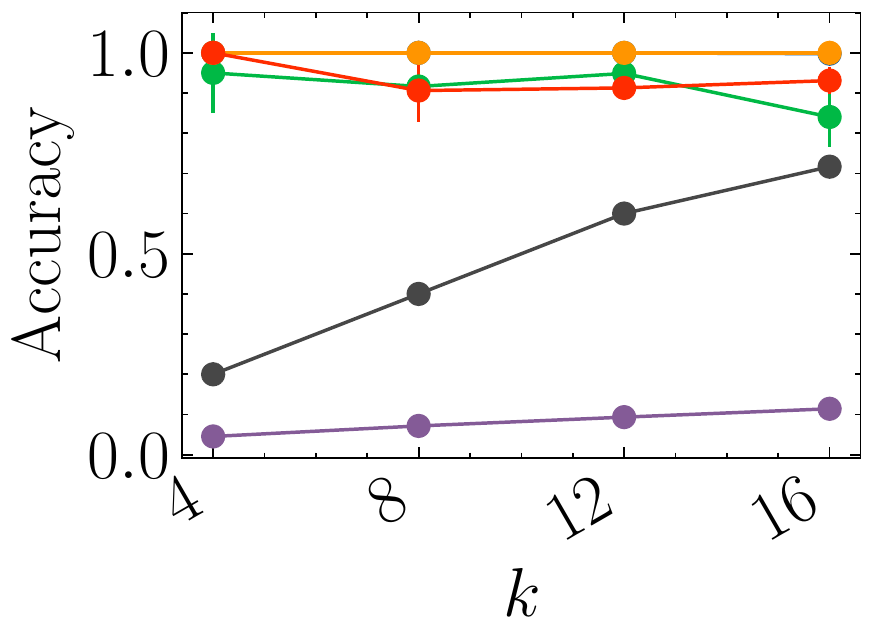}
  }\hspace*{\smallfigsep}
  \subfigure[\small Vary \#steps]{
    \label{fig:results_randomGraph_numSteps_accuracy}
    \includegraphics[height=\smallfigheight,width=\smallfigwidth]{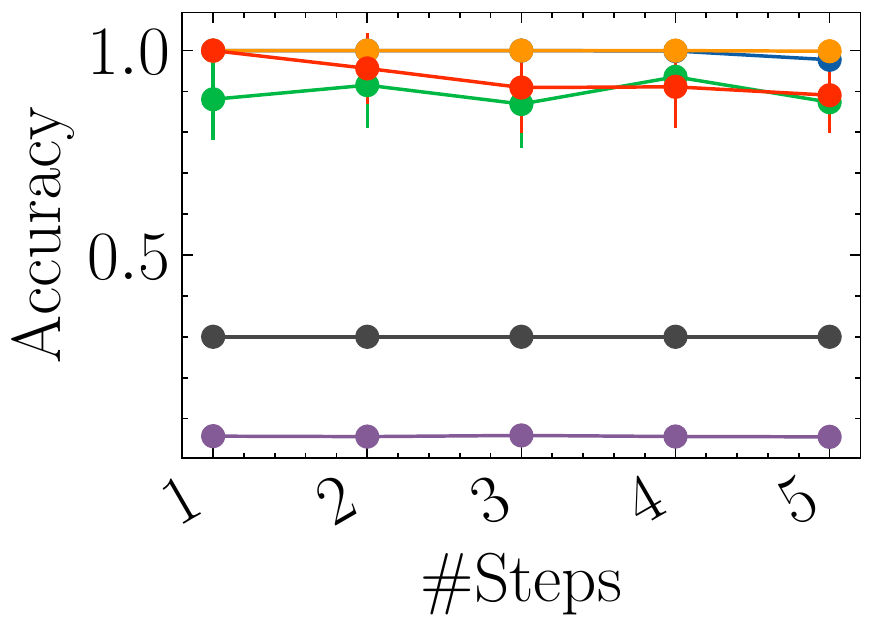}
  }\hspace*{\smallfigsep}
  \subfigure[\small Vary \#walks]{
    \label{fig:results_randomGraph_numWalks_accuracy}
    \includegraphics[height=\smallfigheight,width=\smallfigwidth]{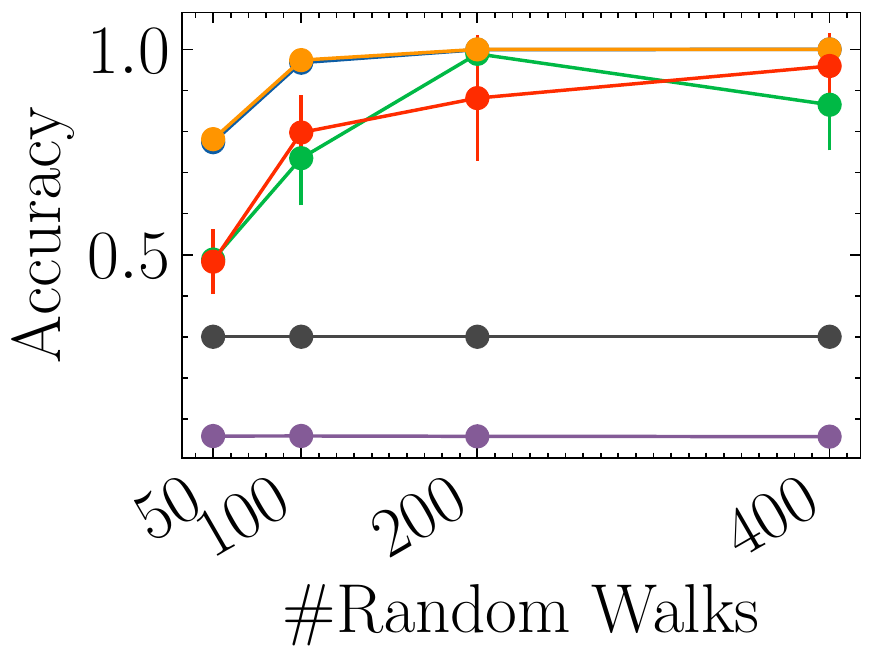}
  } \\
  \subfigure[\small Vary $n$: avg.\ query time]{
    \label{fig:results_randomGraph_n_time}
    \includegraphics[height=\smallfigheight,width=\smallfigwidth]{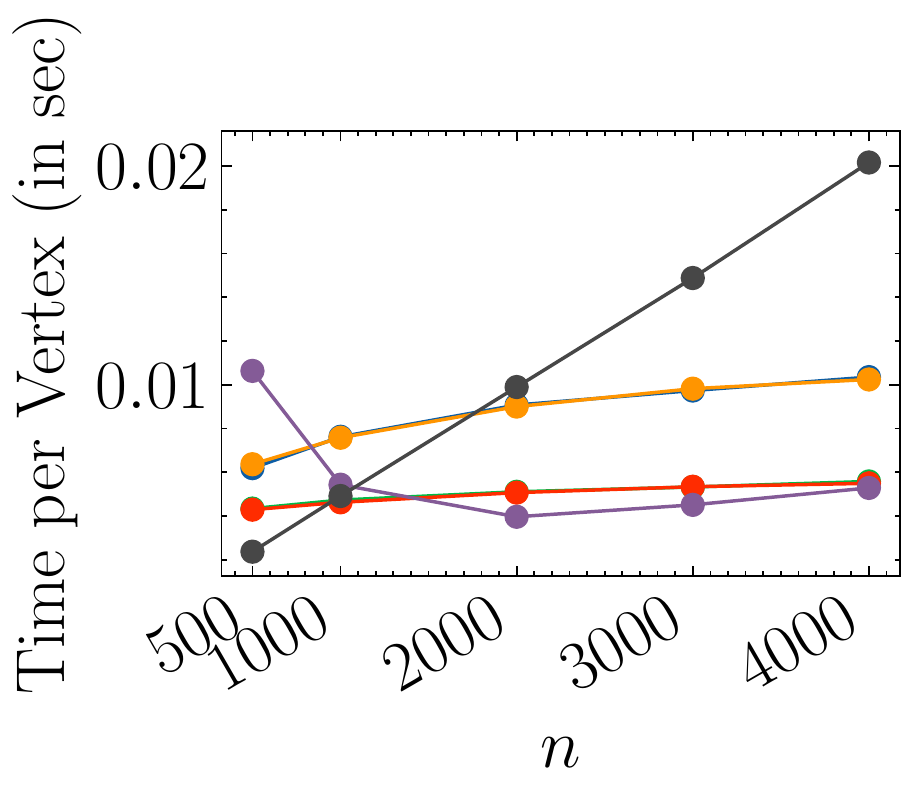}
  }\hspace*{\smallfigsep}
  \subfigure[\small Vary $k$: time]{
    \label{fig:results_randomGraph_k_time}
    \includegraphics[height=\smallfigheight,width=\smallfigwidth]{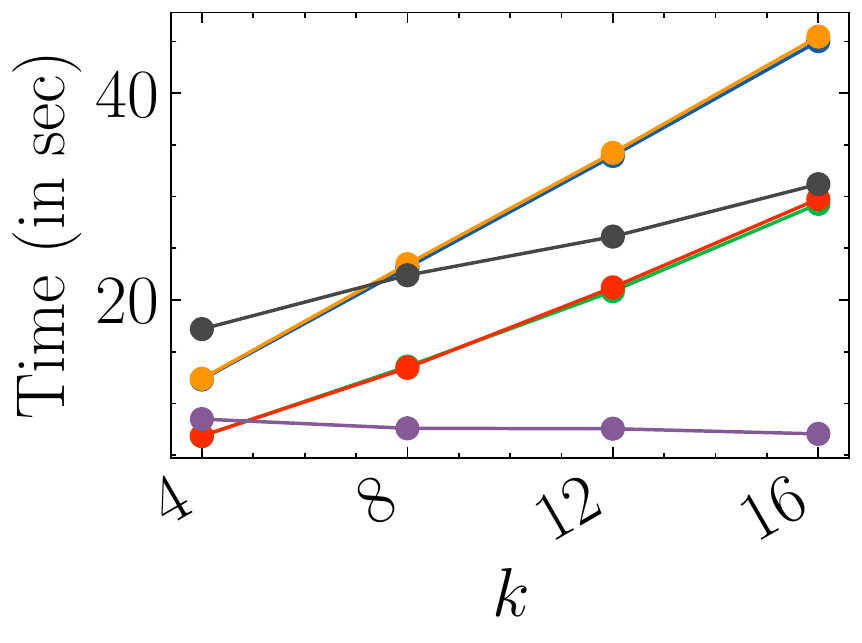}
  }\hspace*{\smallfigsep}
  \subfigure[\small Vary \#steps: time]{
    \label{fig:results_randomGraph_numSteps_time}
    \includegraphics[height=\smallfigheight,width=\smallfigwidth]{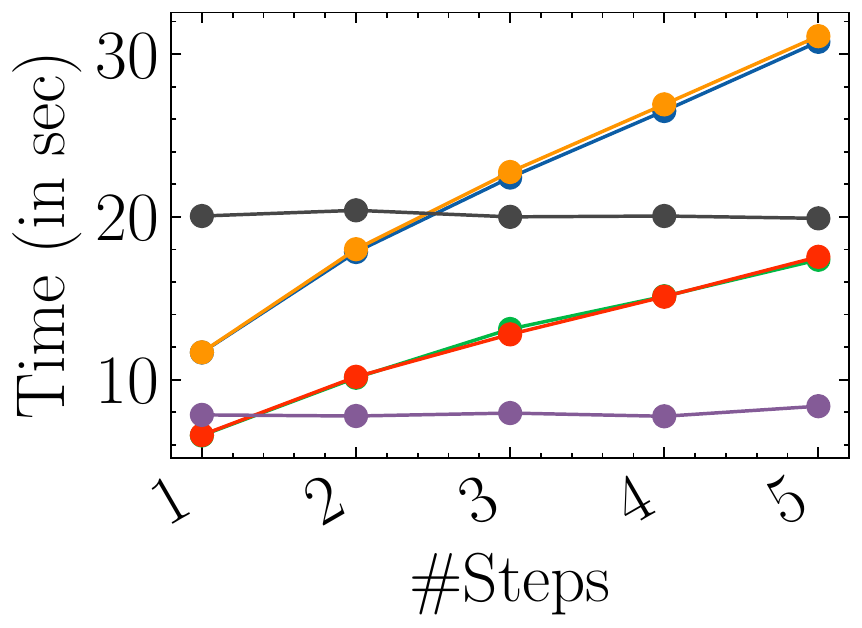}
  }\hspace*{\smallfigsep}
  \subfigure[\small Vary \#walks: time]{
    \label{fig:results_randomGraph_numWalks_time}
    \includegraphics[height=\smallfigheight,width=\smallfigwidth]{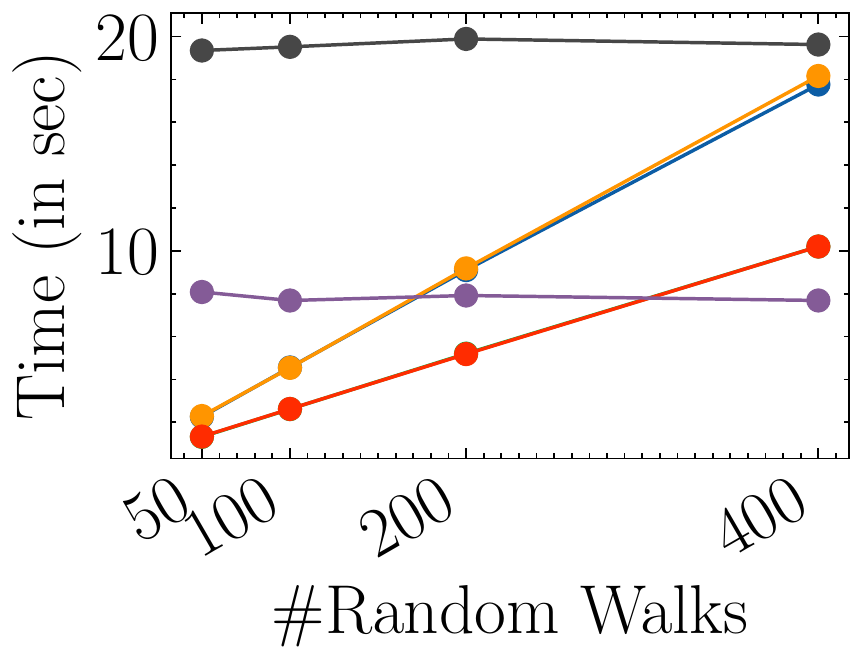}
  } \\

  \end{center}
  \caption{Clustering results on synthetic data.
	  We consider varying number of vertices~$n$
		  (Figs.~\subref{fig:results_randomGraph_n_accuracy}
		  and \subref{fig:results_randomGraph_n_time}),
		varying number of communities~$k$
			(Figs.~\subref{fig:results_randomGraph_k_accuracy} and
			\subref{fig:results_randomGraph_k_time}), 
		varying random walk length
			(Figs.~\subref{fig:results_randomGraph_numSteps_accuracy} and
			\subref{fig:results_randomGraph_numSteps_time}),
		  and varying number of random walks
			  (Figs.~\subref{fig:results_randomGraph_numWalks_accuracy} and
			   \subref{fig:results_randomGraph_numWalks_time}).
	  Figs.~\subref{fig:results_randomGraph_n_accuracy}--\subref{fig:results_randomGraph_numWalks_accuracy}
	  report the clustering accuracies,
	  Fig.~\subref{fig:results_randomGraph_n_time} reports the running time per
	  vertex in seconds, and
	  Figs.~\subref{fig:results_randomGraph_k_time}--\subref{fig:results_randomGraph_numWalks_time}
	  report total running times in seconds.
	  Markers are mean values and error bars are one standard deviation
	  over 5 datasets. }
  \label{fig:experiments}
\end{figure*}
In Fig.~\ref{fig:results_randomGraph_n_accuracy} we vary
the number of vertices~$n$ while keeping $k=6$ fixed. \rwseeded and \rwuseeded
deliver almost perfect accuracy, i.e., they classify almost all vertices
correctly; we note that in the plot, the lines of \rwseeded and \rwuseeded are
essentially identical and thus the line for \rwseeded is hard to see.
\rwunseeded and \rwuunseeded also deliver good results. Furthermore, \polarSeeds
works well when the clusters are small (for $n=500$ there are 83 vertices in
each cluster) but its performance decays as the clusters get larger. \FOCG
generally returns clusterings of low quality because it returns many clusters of
very small sizes.

In Fig.~\ref{fig:results_randomGraph_k_accuracy}, we fix
$n$ and vary $k=4,8,12,16$. Again, we observe a similar behavior as before: our
oracles outperform the competitors, and the competitors improve for smaller
clusters ($k$~larger).

We also varied the parameters for the oracles. In
Fig.~\ref{fig:results_randomGraph_numSteps_accuracy}, we
set the number of random walk steps to $1,2,3,4,5$. We see that even with
very short random walks, the algorithms deliver very good results. However, 
as the number of steps increases, the solution quality slightly decreases (see,
e.g., \rwseeded or \rwuunseeded).  This confirms the theoretical analysis of
Lem.s~\ref{lem:close} and~\ref{lem:dissimilar-distributions}.

In Fig.~\ref{fig:results_randomGraph_numWalks_accuracy},
we set the random walks lengths to $50,100,200,4000$.  \rwseeded and \rwuseeded
return excellent clusterings when at least $200 \approx 4.5\cdot\sqrt{n}$ random
walks are performed.

In 
Figs.~\ref{fig:results_randomGraph_n_time}--\subref{fig:results_randomGraph_numWalks_time}
we report the running times of the algorithms. Our oracles scale
linearly in %
the number of steps
(Fig.~\ref{fig:results_randomGraph_numSteps_time}), and the
number of random walks
(Fig.~\ref{fig:results_randomGraph_numWalks_time}).
Furthermore, since our number of seed nodes depends on the number of
communities~$k$, the oracles scale linearly in~$k$
(Fig.~\ref{fig:results_randomGraph_k_time}).
In
Fig.~\ref{fig:results_randomGraph_n_time} we report the
running time, normalized by the number of vertices in the graph; for our oracle
data structures this corresponds to the time they spend on each query. We observe
that the query times of the oracles increases only very
moderately as the number of vertices~$n$ increases;  we blame this slight
increase on the internal data structures (such as hash maps) that we use to
store our graphs. This is in contrast to \polarSeeds, for which the running time
per vertex is increasing
(Fig.~\ref{fig:results_randomGraph_n_time}).
For \FOCG we observe that it scales linearly in the number of vertices (since in
Fig.~\ref{fig:results_randomGraph_n_time} the average
time per vertex is nearly constant for $n\geq1000$) and its running time
slightly decreases as it finds better communities
(Figs.~\ref{fig:results_randomGraph_k_accuracy}
 and~\ref{fig:results_randomGraph_k_time})

\emph{Biclustering experiments.}
In Fig.~\ref{fig:bi-experiments} we present our results for finding biclusters
$(V_1,V_2),(V_3,V_4),\dots,(V_{2k-1},V_{2k})$. Thus, we run the biclustering
versions of the algorithms.  The oracles used $t=2$~random walk steps and
$R=600$~random walks, unless stated otherwise. The seeded oracles and \polarSeeds
obtained 3~seed vertices from each $V_i$; the unseeded oracles randomly sampled
$6k$~seed vertices in the preprocessing. 

\begin{figure*}[t!]
  \begin{center}
  \subfigure{
	  \includegraphics[width=3\smallfigwidth]{legend_randomGraph_k_accuracy_biclustering}
  }
  \addtocounter{subfigure}{-1}

  \centering
  \subfigure[\small Vary $n$]{
    \label{fig:bi-results_randomGraph_n_accuracy_biclustering}
    \includegraphics[height=\smallfigheight,width=\smallfigwidth]{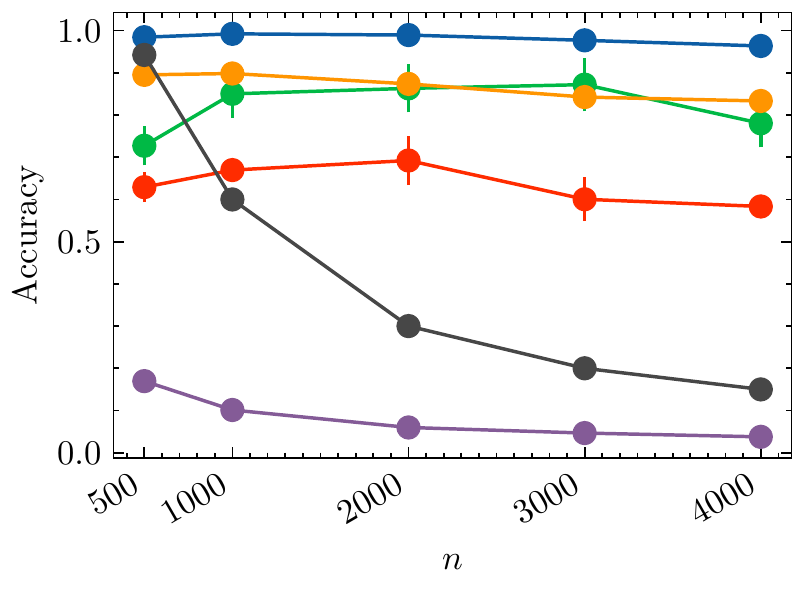}
  }\hspace*{\smallfigsep}
  \subfigure[\small Vary $k$]{
    \label{fig:bi-results_randomGraph_k_accuracy_biclustering}
    \includegraphics[height=\smallfigheight,width=\smallfigwidth]{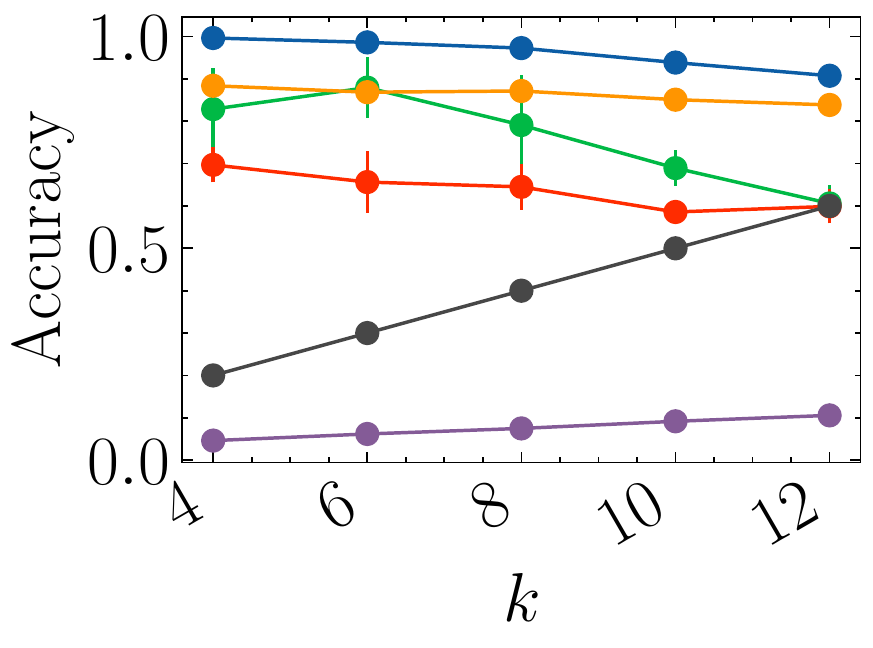}
  }\hspace*{\smallfigsep}
  \subfigure[\small Vary \#steps]{
    \label{fig:bi-results_randomGraph_numSteps_accuracy_biclustering}
    \includegraphics[height=\smallfigheight,width=\smallfigwidth]{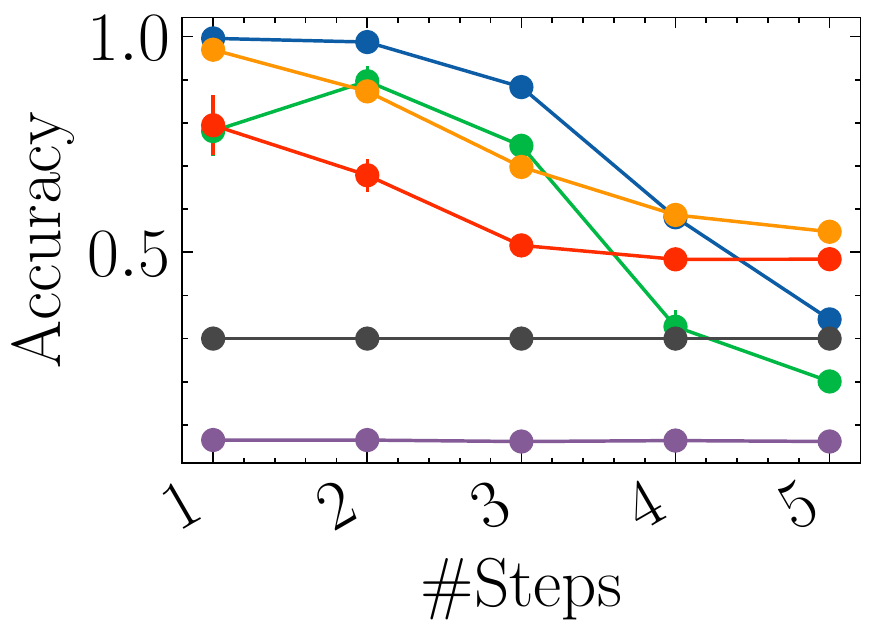}
  }\hspace*{\smallfigsep}
  \subfigure[\small Vary \#walks]{
    \label{fig:bi-results_randomGraph_numWalks_accuracy_biclustering}
    \includegraphics[height=\smallfigheight,width=\smallfigwidth]{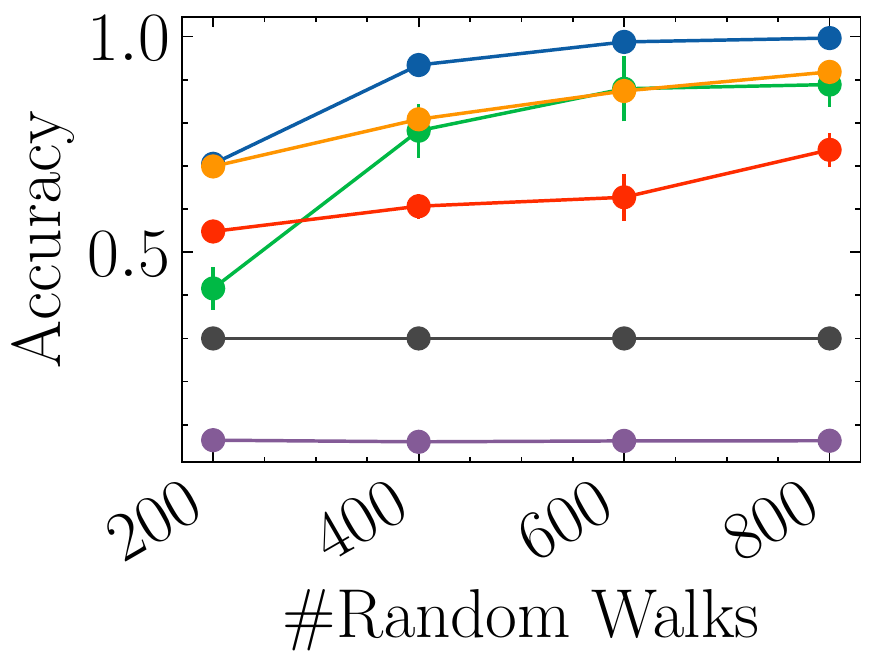}
  } \\
  \end{center}
  \caption{Biclustering results on synthetic data.
	  We vary the number of vertices~$n$
		  (Fig.~\subref{fig:bi-results_randomGraph_n_accuracy_biclustering}),
		 the number of communities~$k$
			 (Fig.~\subref{fig:bi-results_randomGraph_k_accuracy_biclustering}),
		 the random walk length
			 (Fig.~\subref{fig:bi-results_randomGraph_numSteps_accuracy_biclustering})
		and  the number of random walks
		(Fig.~\subref{fig:bi-results_randomGraph_numWalks_accuracy_biclustering}).
	  We report the achieved accuracies;
	  markers are mean values over 5 different datasets, and error bars are one
	  standard deviation over the 5 datasets. }
  \label{fig:bi-experiments}
\end{figure*}

Again, our oracles obtain better results than the baseline algorithms, which
typically return too small clusters.
Furthermore, the signed oracles \rwseeded and \rwunseeded
outperform the unsigned oracles \rwuseeded and \rwuunseeded, resp. This shows that the
edge signs are necessary to split the clusters $U_i$ into
biclusters $V_{2i-1}$ and $V_{2i}$. 
Compared to the clustering setting from before,
the biclustering algorithms are more sensitive to the number steps
(Fig.~\ref{fig:bi-results_randomGraph_numSteps_accuracy_biclustering}),
and they also require more random walks
(Fig.~\ref{fig:bi-results_randomGraph_numWalks_accuracy_biclustering}).

\emph{Conclusion.}
Our experiments suggest that our oracles outperform the baselines when the
clusters are large. Also, to recover the biclusters $(V_{2i-1},V_{2i})$, it is
necessary to use the edge signs.  Furthermore, the seeded methods outperform the
unseeded methods.

\end{document}